\documentclass[12pt,a4paper]{article}
\usepackage[utf8]{inputenc}
\usepackage[english]{babel}

\usepackage{amsmath, amssymb, amsfonts, physics, braket, mathtools, bigints, geometry}
\usepackage{amsthm}
\usepackage{pgfplots, subcaption, floatrow, footnote, adjustbox,float,fancyvrb, colonequals}
\usepackage{graphicx, grffile, epsfig, listings}
\usepackage{verbatim, dsfont, accents}
\usepackage{textcomp}
\usepackage{pdfpages}

\usepackage[dvipsnames]{xcolor}
\usepackage[toc,page]{appendix}
\usepackage{authblk}
\usepackage[bookmarksnumbered=true]{hyperref}
\usepackage{tikz}
\usetikzlibrary{decorations.pathreplacing, patterns}

\usepackage{capt-of, caption} %for captions in minipages
\numberwithin{equation}{section}
\title{Momentum Distribution of a Fermi Gas with Coulomb Interaction in the Random Phase Approximation}

\author[1,*,***]{Niels Benedikter}
\author[2,**]{Sascha Lill}
\author[3,*]{Diwakar Naidu}
\affil[1]{ORCID: \href{https://orcid.org/0000-0002-1071-6091}{0000-0002-1071-6091}, e--mail: \href{mailto:niels.benedikter@unimi.it}{niels.benedikter@unimi.it}}
\affil[2]{ORCID: \href{https://orcid.org/0000-0002-9474-9914}{0000-0002-9474-9914}, e--mail: \href{mailto:sali@math.ku.dk}{sali@math.ku.dk}}
\affil[3]{ORCID: \href{https://orcid.org/0009-0000-5567-4529}{0009-0000-5567-4529}, e--mail: \href{mailto:diwakar.naidu@unimi.it}{diwakar.naidu@unimi.it}}
\affil[*]{Università degli Studi di Milano, Via Cesare Saldini 50, 20133 Milano, Italy}
\affil[**]{University of Copenhagen, Universitetsparken 5, 2100 Copenhagen, Denmark}
\affil[***]{External Scientific Member of Basque Center for Applied Mathematics, Alameda de Mazarredo 14, 48009 Bilbao, Bizkaia, Spain}

\newcommand{\cE}{\mathcal{E}}
\newcommand{\cF}{\mathcal{F}}

\newcommand{\cK}{\mathcal{K}}
\newcommand{\cN}{\mathcal{N}}
\newcommand{\cO}{\mathcal{O}}

\newcommand{\fR}{\mathfrak{R}}

\newcommand{\CCC}{\mathbb{C}}
\newcommand{\NNN}{\mathbb{N}}

\newcommand{\TTT}{\mathbb{T}}
\newcommand{\ZZZ}{\mathbb{Z}}
\newcommand{\Zbb}{\mathbb{Z}}

\newcommand{\ulambda}{\underline{\lambda}}

\newcommand{\Bog}{\textnormal{Bog}}
\newcommand{\corr}{\textnormal{corr}}

\newcommand{\di}{\textnormal{d}}

\newcommand{\ex}{\mathrm{ex}}
\newcommand{\F}{\mathrm{F}}
\newcommand{\FS}{\mathrm{FS}}
\newcommand{\GS}{\mathrm{gs}}

\newcommand{\HS}{\mathrm{HS}}
\newcommand{\I}{\mathrm{I}}
\newcommand{\II}{\mathrm{II}}

\newcommand{\RPA}{\mathrm{RPA}}

\newcommand{\kF}{k_\F}

\newcommand{\BFc}{B_\F^c}

\newcommand{\fock}{\mathcal{F}}
\newcommand{\Ncal}{\mathcal{N}}

\newcommand{\Nbb}{\mathbb{N}}

\newcommand{\tagg}[1]{ \stepcounter{equation} \tag{\theequation}
\label{#1} } % add tag and label in align*-environments

\newcommand{\Zstar}{\mathbb{Z}^3} %Make this \Z^3 \setminus \{0\} if needed.

\newcommand{\Rbb}{\mathbb{R}}

\newcommand{\R}{\mathbb{R}}
\newcommand{\C}{\mathbb{C}}
\newcommand{\N}{\mathbb{N}}
\newcommand{\Z}{\mathbb{Z}}
\newcommand{\T}{\mathbb{T}}

\newcommand{\Hcal}{\mathcal{H}}

\newcommand{\NN}{\mathcal{N}}

\newcommand\Item[1][]{%
  \ifx\relax#1\relax  \item \else \item[#1] \fi
  \abovedisplayskip=0pt\abovedisplayshortskip=0pt~\vspace*{-\baselineskip}}

\newcommand{\half}{\frac{1}{2}}
\newcommand{\eva}[1]{\left\langle #1 \right\rangle}

\theoremstyle{plain}
\newtheorem{theorem}{Theorem}[section]
\newtheorem{lemma}[theorem]{Lemma}

\newtheorem{proposition}[theorem]{Proposition}

\theoremstyle{definition}
\newtheorem{definition}[theorem]{Definition}

\newtheorem*{remarks}{Remarks}

\theoremstyle{remark}

\theoremstyle{plain}
\newtheorem*{theorem*}{Theorem}
\newtheorem*{lemma*}{Lemma}
\newtheorem*{corollary*}{Corollary}
\newtheorem*{proposition*}{Proposition}

\theoremstyle{definition}
\newtheorem*{definition*}{Definition}
\newtheorem*{problem*}{Problem}
\newtheorem*{assumption*}{Assumption}
\newtheorem*{example*}{Example}

\theoremstyle{remark}
\newtheorem*{claim*}{Claim}
\newtheorem*{remark*}{Remark}

\begin{document}
\maketitle
\begin{abstract}

We analyse the momentum distribution of a three-dimensional Fermi gas in the mean-field scaling regime in a trial state that was recently proven to reproduce the Gell-Mann--Brueckner correlation energy for Coulomb potentials. For a class of potentials including the Coulomb potential we show that the momentum distribution is given by a step profile corrected by a random phase approximation contribution. Moreover, for potentials with summable Fourier transform we provide optimal error bounds for the deviation from the random phase approximation. This refines a recent analysis by two of the authors to the physically most relevant potentials and to momenta closer to the Fermi surface. The proof relies on a double bootstrap method, improved over the earlier analysis to estimate the momentum distribution both globally over all momenta and pointwise. We argue that a similar result can be expected to hold also for the ground state.
\medskip


\medskip


\end{abstract}

\section{Introduction and Main Result}
\label{sec:intro}

We consider a quantum system of $N$ spinless fermionic particles moving on the torus $\mathbb{T}^3\coloneq \Rbb^3/ (2\pi \Zbb^3)$ of fixed side length $2\pi$. The system is described by the Hamilton operator
\begin{equation}
	H_N := - \sum_{j=1}^{N}\Delta_{x_j} + \kF^{-1} \sum_{1\leq i < j \leq N } V(x_i - x_j) \;,
\end{equation}
where $ \kF > 0 $ is the Fermi momentum. The particle number $ N $ is fixed as the number of momenta in the Fermi ball, $ N = |\{ k \in \ZZZ^3 : |k| \leq \kF \}| $. The Hamiltonian $ H_N $ acts on wave functions in the antisymmetric tensor product $L^2_{\mathrm{a}}(\T^{3N}) = \bigwedge_{j=1}^N L^2(\T^3)$.
The choice of the coupling constant as $\kF^{-1}$ represents the mean-field scaling limit, where we are interested in the asymptotics as $\kF \to \infty$.

At zero temperature, the system will be in a ground state, that is, a vector $ \Psi_{\GS} \in L^2_{\mathrm{a}}(\T^{3N}) $ which attains the ground state energy
\begin{equation} \label{eq:EGS}
	E_{\GS}
	:= \inf_{\substack{\Psi \in L^2_{\mathrm{a}}(\T^{3N}) \\||\Psi|| = 1}} \langle \Psi, H_N \Psi \rangle \;.
\end{equation}
In this article we are interested in the momentum distribution of states $ \Psi $ which are energetically close to $ \Psi_{\GS} $. The momentum distribution is the expectation value of the operator representing the number of fermions with momentum $q \in \Zbb^3$, i.\,e.,
\begin{align}
	n(q) \coloneq \eva{\Psi, a^*_q a_q \Psi} \;,
\end{align}
where $ a_q^*$ and $a_q $ are the fermionic creation and annihilation operators. As the ground state of interacting many-body systems is very difficult to compute, we focus our attention on a trial state $ \Psi_N $ which is expected to capture the ground state's properties at least to the first non-trivial order beyond mean-field theory. This is analogous to the analysis conducted in \cite{BL25} based on the collective bosonization methods of \cite{BNP+20,BNP+21}, however we now consider the ``patchless'' trial state of~\cite{CHN23}. As long as the analysis of the true ground state remains elusive, we believe that studying two different constructions of low-energy states and obtaining consistent expressions for the momentum distribution adds plausibility to the conjecture that the obtained momentum distribution is actually close to the one of the true ground state. Both methods, that of \cite{BL25} and \cite{CHN23}, are based on the idea of considering particle--hole pairs as approximately bosonic quasiparticles. The approach of \cite{BNP+20,BNP+21}, based on collective degrees of freedom averaged over patches on the Fermi surface, leads to estimates with a more pronounced bosonic nature; the approach of \cite{CHN23} avoids the use of cutoffs which in \cite{BL25} somewhat obscured the result at momenta very close to the Fermi surface.

In the non-interacting case ($ V=0 $), the ground state is given by a Slater determinant of $ N $ plane waves (which we also call the Fermi ball state)
\begin{equation}
	\Psi_{\FS}(x_1, x_2, \ldots, x_N) \coloneq \frac{1}{\sqrt{N!}}\text{det}\left(\frac{1}{(2\pi)^{3/2}}e^{ik_j\cdot x_i}\right)^N_{j,i=1} \;.
\end{equation}
The momenta $ k_j \in \Zbb^3$ are chosen to minimize the kinetic energy $ \sum_{j=1}^N |k_j|^2 $. To avoid a degenerate ground state, we assume that they fill up a Fermi ball
\begin{equation}
	B_{\F} \coloneq \{ k \in \ZZZ^3 : |k| \leq \kF \} \;, \qquad
	|B_{\F}| = N \qquad \textnormal{for some } \kF > 0 \;.
\end{equation}
The relation between the Fermi momentum $ \kF \in \Rbb$ and the particle number is inverted as
\begin{equation}
	\kF = \left(\frac{3}{4\pi}\right)^\frac{1}{3}N^\frac{1}{3} + \mathcal{O}(1) \qquad \textnormal{for } N\to \infty \;.
 \end{equation}
In the Fermi ball state the momentum distribution is the indicator function
\begin{equation}
	\langle \Psi_{\FS}, a_q^* a_q \Psi_{\FS} \rangle
	= \mathds{1}_{B_{\F}}(q) \;.
\end{equation}

\paragraph{The Fermionic Mean-Field Scaling Limit in Context} The mean field scaling of fermionic systems, in the sense of scaling down the interaction by a coupling constant of order $N^{-1/3}$ compared to the kinetic energy, was introduced by \cite{NS81} in the context of the derivation of the Vlasov equation from many-body quantum dynamics. Their result was generalized to a large class of interaction potentials by \cite{Spo81}. Much later, in this scaling, a derivation of the time-dependent Hartree--Fock equation for short times was given by \cite{EESY04}. This line of research was continued by \cite{BPS14b} who derived the time-dependent Hartree--Fock equation for arbitrary times (that proof was later simplified by \cite{BD23}). In the following the derivation of the time-dependent Hartree--Fock equation was generalized to fermions with Coulomb interaction \cite{PRSS17}, fermions with relativistic dispersion relation \cite{BPS14a}, fermions in mixed states \cite{BJP+16}, and fermions in magnetic fields \cite{BBMN25}. This also reinvigorated research on the derivation of the Vlasov equation, yielding results admitting increasingly singular potentials \cite{BPSS16, Saf18,Saf20,Saf20a,LS23,Saf21,CLS24}. Recently results have also been obtained extending the derivation of time-dependent Hartree--Fock theory to extended Fermi gases \cite{FPS23,FPS25}, compared to the mean-field limit in fixed volume.

In parallel, also the static properties of fermionic systems were studied. The ground state energy to the precision of Hartree--Fock theory was obtained by \cite{GS94} based on correlation inequalities developed in \cite{Bac92,Bac93} -- actually for the thermodynamic limit, but the method applies equally to the simpler mean-field scaling regime. The first result beyond Hartree--Fock theory, was achieved by a rigorous version of second order perturbation theory by \cite{HPR20}. A systematic treatment of the dominant correction to Hartree--Fock theory to all perturbative orders, in agreement with the random phase approximation, was achieved as an upper bound by \cite{BNP+20}. The corresponding lower bounds were proven for small interaction potential by \cite{BNP+21} and then by \cite{BPSS23} for arbitrary interaction potential. All these results relied on collective bosonization of particle--hole pairs delocalized in momentum space over patches on the Fermi surface. An alternative approach avoiding patches and bosonizing directly individual particle--hole pairs was developed by \cite{CHN23,CHN23a} and generalized to Coulomb interaction in \cite{CHN24}. This is the approach which gives rise to the trial state in the present paper.

The progress on the derivation of the random phase approximation for the ground state energy lead also to new results for the dynamics: a Fock space norm approximation for the dynamics of bosonizable degrees of freedom was proven in \cite{BNP+22} (see also the review \cite{Ben22} for a more detailed discussion of this context). Partial results for the spectrum and excited states were obtained by \cite{Ben21,CHN22}.

Apart from the mean-field scaling regime another well-studied case of the interacting Fermi gas is the dilute regime. In this context should be mentioned the analysis of the ground state energy of the spin-balanced case \cite{Gia23,Gia24} which has been inspired by bosonization and Bogoliubov theory \cite{FGHP21a,Gia23b} and recently been pushed to the order of the Huang--Yang formula \cite{GHNS24,GHNS25}, and the study of the ground state energy in the spin-polarized case \cite{LS24,LS24a,LS24b,LS24c}. Recently also the momentum distribution in a trial state approximating the ground state to Huang--Yang precision was computed for the dilute Fermi gas \cite{BGLL26}.

\subsection{Main Result}
\label{subsec:mainresult}

We also introduce the complement of the Fermi ball
\begin{equation}
   B_\text{F}^c := \Zbb^3 \setminus B_\text{F} \;.
\end{equation}
We will prove that the momentum distribution $ n(q) $ of a specific trial state energetically very close to the ground state, for $ q \in B_{\F}^c $ is approximately given by the random phase approximation
\begin{equation} \label{eq:nqb}
	n^{\RPA}(q)
	\coloneq \sum_{\ell \in \Zstar}\mathds{1}_{L_{\ell}}(q) \; \frac{1}{\pi}\int_0^\infty \frac{g_\ell (t^2-\lambda^2_{\ell,q}) (t^2 + \lambda^2_{\ell,q})^{-2}}{1 + 2g_\ell \sum_{p \in L_{\ell}}\lambda_{\ell,p} (t^2+\lambda^2_{\ell,p})^{-1}} \mathrm{d}t \;,
\end{equation}
where the lens $ L_\ell \in \Z^3 $, the excitation energy $ \lambda_{\ell,p} > 0 $, and $ g_\ell \ge 0 $ are defined by
\begin{equation} \label{eq:Lell}
	L_\ell \coloneq B_{\F}^c \cap (B_{\F} + \ell) \;, \qquad
	\lambda_{\ell,p} \coloneq \half (|p|^2 - |p-\ell|^2) \;, \qquad
	g_\ell \coloneq \frac{\kF^{-1} \hat{V}(\ell) }{2 (2 \pi)^3} \;,
\end{equation}
and the Fourier transform of the potential follows the convention
\begin{equation} \label{eq:Fourierconvention}
	\hat{V}(\ell) := \int V(x) e^{-i k \cdot x} \di x \;.
\end{equation}
For $ q \in B_{\F} $, we analogously have
\begin{equation}\label{eq:inside}
\begin{split}
	n(q) & \approx 1 - n^{\RPA}(q) \;, \\
	n^{\RPA}(q) & 	\coloneq \sum_{\ell \in \Zstar}\mathds{1}_{L_{\ell}}(q+\ell) \; \frac{1}{\pi}\int_0^\infty \frac{g_\ell (t^2-\lambda^2_{\ell,q+\ell}) (t^2 + \lambda^2_{\ell,q+\ell})^{-2}}{1 + 2g_\ell \sum_{p \in L_{\ell}}\lambda_{\ell,p} (t^2+\lambda^2_{\ell,p})^{-1}} \mathrm{d}t \;.
\end{split}
\end{equation}
Therefore $ n^{\RPA}(q) $ represents the deviation of $ n(q) $ from the indicator function $ \mathds{1}_{B_{\F}}(q) $ we found in the non-interacting case. The scaling of $ n^{\RPA}(q) $ and the error terms will depend on the energy distance from the Fermi surface that we measure by
\begin{equation} \label{eq:eq}
	e(q) \coloneq \abs{|q|^2 - \left( \inf_{p \in B_{\F}^c} |p|^2 - \half \right)} \;.
\end{equation}
Since $ q \in \Z^3 $, the numbers $|p|^2$, $|p-\ell|^2$, and $|q|^2 $ are integer and thus
\begin{equation}
	\lambda_{\ell,p} \geq \half \;, \qquad \textnormal{and} \qquad e(q) \ge \half \;.
\end{equation}
We introduce $ e(q) $ because it gives rise to the important improved lower bound
\begin{equation}
   \lambda_{\ell,q} \ge \half e(q) \;.
\end{equation}

Our main result shows the existence of a trial state that both minimizes the energy to the precision of the random phase approximation and has a momentum distribution as expected for a Fermi liquid, with a jump at the Fermi momentum and a correction to the height of the jump as predicted in the physics literature \cite{DV60}.

\begin{theorem}[Main result] \label{thm:main}
Assume that
\begin{equation} \label{hyp:Coulomb}
\hat{V} \ge 0 \text{ is radial, decreasing, and there exists $ C > 0 $ such that }
 \hat{V}(\ell) \le C |\ell|^{-2} \quad \forall \ell \in \Z^3\;.
\end{equation}
Then, for $\kF \to \infty $, there exists a sequence of trial states  $ \Psi_N \in L^2_{\mathrm{a}}(\T^{3N}) $ such that
\begin{itemize}
\item $ \Psi_N $ is energetically close to the ground state in the sense that for any $ \delta \in (0,\frac{1}{6}) $ there exists a constant $ C_\delta > 0 $ such that for all $k_\F$ we have
\begin{equation} \label{eq:main1}
	\eva{\Psi_N, H_N \Psi_N} - E_{\GS}
	\le C_\delta \kF^{1-\frac 16 + \delta} \;,
\end{equation}
\item and there exists a family of operators $\mathcal{E}(q)$ on Fock space such that for all $ \kF $ and $ q \in \Z^3 $, we have
\begin{equation} \label{eq:main2}
	n(q) = \eva{\Psi_N, a_q^* a_q \Psi_N}
	= \begin{cases}
	n^{\RPA}(q) + \cE(q) & \quad
		\textnormal{for } |q| \ge \kF \\
	1 - n^{\RPA}(q) + \cE(q) & \quad
		\textnormal{for } |q| < \kF
	\end{cases}
\end{equation}
with $ n^{\RPA}(q) $ defined in~\eqref{eq:nqb}, and where the error term is bounded by
\begin{equation}\label{eq:main_error}
	\lvert \cE(q)\rvert \le C_\delta \kF^{-1 - \frac 16 + \delta} e(q)^{-1} \;.
\end{equation}
\end{itemize}
If additionally we have
\begin{equation}\label{hyp:ell1}
   \sum_{\ell \in \Zstar} \hat{V}(\ell) < \infty\;,
\end{equation}
then the error term is even bounded by
\begin{equation} \label{eq:main_improvederror}
	\lvert \cE(q)\rvert \le C_\delta \kF^{-2 +\delta} e(q)^{-1} \;.
\end{equation}
\end{theorem}

\begin{remarks}
\begin{enumerate}

\item In Lemma~\ref{lem:nqb_bounds}, we show that the bosonization contribution is bounded by
\begin{equation} \label{eq:nrpa}
	\lvert n^{\RPA}(q)\rvert \le C \kF^{-1} e(q)^{-1} \;,
\end{equation}
which we expect to be optimal in the mean-field scaling.
For comparison, in~\cite{BL25}, the scaling of the leading term is $ n^{\RPA}(q) \sim \kF^{-2} $ without $ e(q)^{-1} $. This is because~\cite{BL25} assumes $ |\ell| \le C $ and $ |\ell \cdot q| \ge c $ in the analogue of~\eqref{eq:nqb}, leading to $ e(q) \sim \kF $. In other words, it includes only momenta at distances $ ||q|-\kF| \sim 1 $ from the Fermi surface, while the present result admits distances as small as $ \kF^{-1} $.

\item While \eqref{eq:nqb} can formally be obtained by a rather straight-forward computation under the assumption that bosonization was exact, the proof requires a control of the error terms that is much more subtle than previously developed. This is due to the requirement of controlling the momentum distribution pointwise, for which we have developed a novel bootstrap method in \cite{BL25}, which in the present paper we further refine to a double bootstrap argument.

\item Even though our result concerns a trial state, we expect that it captures the behavior of the ground state. Our result is obtained as the vacuum expectation value of the operator expansion \eqref{eq:operatorexpansion}. The main difference in evaluating this expansion for the ground state $\Psi_\textnormal{gs}$ is that we have to add $\langle \psi, (a^*_q a_q + a^*_{-q} a_{-q})\psi \rangle$ with $\psi := e^{S} \Psi_\textnormal{gs}$, instead of $\psi = \Omega$ for the trial state. As a consequence of the results on the ground state energy, compare for example \cite[Eq.~(11.9)]{BNP+21}, one can prove estimates of the form
\[
	\langle \psi, \mathbb{H}_0 \psi \rangle \leq C \kF^{1- \frac{1}{48}}	\;.
\]
The contribution of particle excitations to this expectation value may be written as $\sum_{p \in \BFc} \lvert \lvert p\rvert^2 - \kF^2 \rvert \langle \psi, a^*_p a_p \psi\rangle$; holes have a similar non-negative contribution. \emph{Assuming} rotational symmetry was not broken (which is difficult to prove), it is plausible that $\langle \psi, a^*_p a_p \psi\rangle$ is approximately equal to its average over the spherical shell between the radii $\lvert p\rvert$ and $\lvert p\rvert + 1$; let us call this average $n^\textnormal{avg}(\lvert p\rvert)$. Then
\[
	C \kF^{1-\frac{1}{48}} \geq \langle \psi, \mathbb{H}_0 \psi \rangle \simeq \sum_{|p| \in \Nbb} \lvert \lvert p\rvert^2 - \kF^2 \rvert \, C |p|^2 \, n^\textnormal{avg}(\lvert p\rvert) \;,
\]
where we relied on the number of momenta in the shell being of order $\lvert p\rvert^2$. This could then be resolved for $n^\textnormal{avg}(\lvert p\rvert) \leq C \kF^{1-\frac{1}{48}} \left( e(p) \lvert p\rvert^{2}\right)^{-1}$. For momenta $\lvert p\rvert$ just outside the Fermi ball, so of order $\kF$, this suggests that
\[
	\langle \psi, a^*_p a_p \psi\rangle \simeq n^\textnormal{avg}(\lvert p\rvert) \leq C e(p)^{-1} \kF^{-1-\frac{1}{48}} \;,
\]
which adds a term that is subleading compared to the RPA-term \eqref{eq:nrpa}.
% , though the estimate would not be as strong as \eqref{eq:main_error} or \eqref{eq:main_improvederror}.

\item In absence of rotational symmetry, to get meaningful bounds by the approach just discussed one would have to compute the ground state energy to order $N^0$ since the gap of the kinetic energy is very small, namely $\lambda_{\ell,p} \geq \frac{1}{2}$.

\end{enumerate}
\end{remarks}

%%%%%%%%%%%%%

In analogy to~\cite[Section~1.1]{BL25} we may approximate $ \lambda_{\ell,q} \approx \kF |\ell| |\hat{\ell} \cdot \hat{q}| $ with $ \hat{q} := {q}/{|q|} $, then substitute $ \mu := t \kF |\ell| $ and formally take the continuum limit of sums to integrals.
Thus
\begin{align*}
	n^{\RPA}(q)
	&\approx \int_{\R^3} \di \ell \; \mathds{1}_{L_{\ell}}(q) \; \frac{\hat{V}(\ell) \kF^{-2}}{(2 \pi)^4 |\ell|}
		\int_0^\infty \di \mu \frac{(\mu^2-|\hat{\ell} \cdot \hat{q}|^2) (\mu^2 + |\hat{\ell} \cdot \hat{q}|^2)^{-2}}{1 + Q_\ell(\mu)} \;, \\
	Q_\ell(\mu) &:= \frac{\hat{V}(\ell)}{(2 \pi)^2} \left( 1 - \mu \arctan \left( \frac{1}{\mu} \right) \right) \;. \tagg{eq:Q}
\end{align*}
This agrees with \cite{BL25} up to a factor $ (2 \pi)^3 \kappa $ multiplying $\hat{V}$, explained by our choice of coupling constant $k_\F^{-1}$ compared to the coupling constant $N^{-1/3} = \kappa^{-1} \kF^{-1}$ in \cite{BL25}, as well as the absent $ (2 \pi)^{-3} $ in the Fourier convention~\eqref{eq:Fourierconvention}.

\medskip

Analogously to the exchange contribution to the ground state energy derived in~\cite{CHN23,CHN24}, an exchange contribution $\frac{1}{4}n^{\ex,1}(q)$ appears by normal ordering the bosonization errors, see Lemma~\eqref{lem:normalordering_errors}. However, Proposition~\ref{prop:main} shows that it is subleading.

\medskip

%%%%%%%%%%%%%%

Hypothesis~\eqref{hyp:Coulomb} on $ V $ is only needed to ensure validity of the ground state energy formula~\eqref{eq:main1}, which readily follows from~\cite{CHN23,CHN24}. Our main work is the computation of the momentum distribution, which does not require $\hat{V} $ to be radial or decreasing.

\begin{proposition}[Momentum distribution] \label{prop:main}
Let $ \Psi_N $ be the trial state defined in~\eqref{eq:Psitrial}--\eqref{eq:T} and assume that
\begin{equation}\label{hyp:alpha}
\hat{V}(\ell) = \hat{V}(-\ell) \ge 0 \quad \forall \ell \in \Z^3 \text{ and } \sum_{\ell \in \Zstar} \hat{V}(\ell)^2 |\ell|^\alpha < \infty \text{ for some } \alpha \in (0,2) \;.
\end{equation}
Let $\varepsilon > 0$. Then, there exist a family of operators $\mathcal{E}(q)$ on Fock space and $n^{\textnormal{ex},1}(q) \in \Rbb $  as well as constants $C, C_\varepsilon > 0 $, such that for all $ \kF $ and $ q \in \Z^3 $ we have
\begin{equation} \label{eq:main_prop_2}
	n(q) = \eva{\Psi_N, a_q^* a_q \Psi_N}
	= \begin{cases}
	n^{\RPA}(q) + \frac{1}{4}n^{\ex,1}(q) + \cE(q) & \quad
		\textnormal{for } |q| \ge \kF \\
	1 - n^{\RPA}(q) - \frac{1}{4}n^{\ex,1}(q) + \cE(q) & \quad
		\textnormal{for } |q| < \kF\;,
	\end{cases}
\end{equation}
with the error bounds
\begin{equation}
	\lvert \cE(q)\rvert \le C \kF^{-1 - \frac 16 \alpha} e(q)^{-1} \qquad \text{and} \qquad \lvert n^{\ex,1}(q)\rvert \leq C_\varepsilon \kF^{-1 -\frac{\alpha}{2} + \varepsilon } e(q)^{-1} \;.
\end{equation}
If the additional hypothesis~\eqref{hyp:ell1} holds, then
\begin{equation} \label{eq:main_prop_improvederror}
	\lvert \cE(q)\rvert \le C_\varepsilon \kF^{-2 +\varepsilon} e(q)^{-1} \qquad \text{and} \qquad
	\lvert n^{\ex,1}(q)\rvert \le
		C \kF^{-2} e(q)^{-2} \;.
\end{equation}
\end{proposition}

Once this proposition is established, the main theorem follows immediately.

\begin{proof}[Proof of Theorem~\ref{thm:main}]
We choose as $ \Psi_N $ the trial state defined in~\eqref{eq:Psitrial}--\eqref{eq:T}.
From~\cite[Corr.~1.3]{CHN24}, we know that $ E_{\GS} = E_{\FS} + E_{\corr} + \cO(\kF^{1 - \frac{1}{6} + \delta}) $, provided hypothesis~\eqref{hyp:Coulomb} holds, where $ E_{\FS} $ is the Hartree--Fock energy and $ E_{\corr}$ the correlation energy, see~\cite[Eqs.~(1.2) and (1.11)]{CHN24}. Moreover, \cite[Thm.~1.1]{CHN23} gives us $ \eva{\Psi_N, H_N \Psi_N} \le E_{\FS} + E_{\corr} + \cO(\kF^{1 - \frac{1}{2}}) $, hence $ \eva{\Psi_N, H_N \Psi_N} \le E_{\GS} + \cO(\kF^{1 - \frac{1}{6}+ \delta}) $.

Proposition~\ref{prop:main} yields the statement about the momentum distribution because \eqref{hyp:Coulomb} implies \eqref{hyp:alpha} for $\alpha = 1-6\delta$. Thereafter we choose $\varepsilon = \frac{\alpha}{3}$. The exchange contribution $n^\textnormal{ex}(q)$ can trivially be absorbed in the error bounds \eqref{eq:main_error} and \eqref{eq:main_improvederror}.
\end{proof}

The rest of this article is dedicated to proving Proposition~\ref{prop:main}.
In Section~\ref{sec:trialstate}, we review the construction of the trial state from \cite{CHN23}. In Section~\ref{sec:extraction}, we derive the iterated Duhamel expansion for the momentum distribution and identify the leading order. Section~\ref{sec:prelim_bounds} provides preliminary error estimates used in the later sections. In Section~\ref{subsec:manybody_estimates} we bound the error terms using $\sup_q n(q)$. In Section \ref{sec:leading_order_analysis} we compute the leading order. The proof is concluded in Section~\ref{sec:mainthmproof} by a novel double bootstrap, which first achieves strong bounds globally, i.\,e., on $\sup_{q \in \Zbb^3} n(q)$, then locally, on $n(q)$.

\section{Trial State Construction}
\label{sec:trialstate}

Here we review the trial state construction of~\cite{CHN23}. We introduce the fermionic Fock space
\begin{equation}
	\cF \coloneq \bigoplus_{N=0}^\infty L^2(\T^3)^{\bigwedge N} \;,
\end{equation}
with $ \bigwedge $ denoting the antisymmetric tensor product. The vector
\begin{equation}
\Omega = (1,0,0,\ldots) \in \cF
\end{equation}
is called the vacuum. To each momentum $ q \in \ZZZ^3 $, we assign a plane wave
\begin{equation}
	f_q \in L^2(\TTT^3) \;, \qquad
	f_q(x) \coloneq (2 \pi)^{-3/2} e^{i q \cdot x} \;,
\end{equation}
and the associated creation and annihilation operators
\begin{equation}
	a^*_q \coloneq a^*(f_q) \;, \qquad
	a_q \coloneq a(f_q) \;.
\end{equation}
These operators satisfy the canonical anticommutation relations (CAR)
\begin{equation} \label{eq:CAR}
	\{a_q, a_{q'}^*\} = \delta_{q, q'} \;, \qquad
	\{a_q, a_{q'}\} = \{a_q^*, a_{q'}^*\} = 0 \qquad \text{for all } q, q' \in \ZZZ^3\;.
\end{equation}
Moreover they are bounded in operator norm by $ \Vert a_q^* \Vert \leq 1$ and $\Vert a_q \Vert \le 1 $.

The Hamiltonian may be easily generalized to Fock space by letting $H_n$ act on the $n$-particle sectors $L^2(\T^3)^{\bigwedge n}$ independently for each $n \in  \Nbb$. The advantage of this approach is that we can now construct vectors $\Psi_N \in L^2(\T^3)^{\bigwedge N} \subset \fock$ as well as represent operators using creation and annihilation operators, which provides a convenient way of evaluating expectation values just by using the CAR and the fact that
\begin{equation}
a_q \Omega = 0 \qquad \textnormal{for all } q\in \Zbb^3 \;.
\end{equation}
The generalization of the Hamiltonian to Fock space can be written as
\begin{equation}
\Hcal_N = \sum_{p \in \Zbb^3} \lvert p\rvert^2 a^*_p a_p + \frac{\kF^{-1}}{2 (2\pi)^3}\sum_{k,p,q \in \Zbb^3} \hat{V}(k) a^*_{p+k} a^*_{q-k} a_q a_k\;.
\end{equation}
Note that $\Hcal_N$ still depends on $N$ through the coupling constant $\kF^{-1}$, which remains independent of the Fock space sector.

The Fermi ball state can be written as $ \Psi_{\FS} = R \Omega $ where $ R: \cF \to \cF $ is the unitary particle--hole transformation acting by
\begin{equation} \label{eq:R}
	R^* a_q^* R 	= \mathds{1}_{B_{\F}^c}(q) \, a_q^* 	+ \mathds{1}_{B_{\F}}(q) \, a_q \;.
\end{equation}
Note that $ R^{-1} = R = R^* $. The trial state of~\cite{CHN23} is constructed as
\begin{equation} \label{eq:Psitrial}
	\Psi_N := R T \Omega \;,
\end{equation}
where $ T: \cF \to \cF $ is another unitary operator motivated as follows: in the random phase approximation one considers the interaction to act predominantly by generating pair excitations, where a particle from inside the Fermi ball $ B_{\F} $ is moved by some momentum $ k \in \Z^3 \setminus \{ 0 \} $ to another momentum $ p \in B_{\F}^c $, leaving a ``hole'' in the Fermi ball at $ p-k \in B_{\F} $. After the particle--hole transformation $ R $, this corresponds to a creation of a pair of excitations at $ p $ and $ p-k $, described by the pair creation operator
\begin{equation} \label{eq:b}
	b^*_p(k) \coloneq a_p^* a_{p-k}^* 
	\qquad \textnormal{with adjoint} \qquad
	b_p(k) \coloneq a_{p-k} a_p \;.
\end{equation}
The constraint on $ (p,k) $ can be written as $ p \in L_k $.

The Hamiltonian can be particle-hole transformed to extract explicitly the mean-field contributions $E_{\FS}$ \cite[Eqs.~(1.2) and (1.11)]{CHN24} to its ground state energy
\begin{equation}
   R^* \Hcal_N R = E_{\FS} + H_{\Bog} + \textnormal{non-bosonizable remainders}
\end{equation}
with the bosonizable Bogoliubov-type Hamiltonian (compare \cite[(1.34)]{CHN23})
\begin{equation} \label{eq:HBog}
\begin{split}
	H_{\Bog}
	\coloneq \sum_{k \in \Zstar} \Big( & \sum_{p,q \in L_k} 2 (h(k) + P(k))_{p,q} b^*_p(k) b_q(k) \\
	& 		+ \sum_{p,q \in L_k} P(k)_{p,q} (b_p(k) b_{-q}(-k) + b^*_{-q}(-k) b^*_p(k)) \Big) \;,
\end{split}
\end{equation}
with matrices $ h(k) \in \CCC^{|L_k| \times |L_k|}$ and $P(k) \in \CCC^{|L_k| \times |L_k|} $ defined by
\begin{equation} \label{eq:HkPk}
\begin{aligned}
	h(k)_{p,q} \coloneq \delta_{p,q} \lambda_{k,p} \;, \qquad
	P(k)_{p,q} \coloneq \frac{\kF^{-1}\hat{V}(k) }{2 (2 \pi)^3} \;.
\end{aligned}
\end{equation}
The matrix $ P(k) $ is rank-one and can also be written as
\begin{equation}P(k) = \lvert v_k \rangle \langle v_k \rvert \;, \quad \textnormal{where} \quad v_{k,p} \coloneq g_k^{1/2} = \Big( \frac{\kF^{-1} \hat{V}(k) }{2 (2 \pi)^3} \Big)^\half \;.
\end{equation}
 As the pair operators satisfy approximate bosonic commutation relations (see Lemma~\ref{lem:paircomm}), $ H_{\Bog} $ can be approximately diagonalized to obtain the next order of the ground state energy (the random phase approximation to the correlation energy, required to obtain the precision \eqref{eq:main1} of the ground state energy) by an approximate Bogoliubov transformation \cite[Thm.~1.4]{CHN23}
\begin{equation} \label{eq:T}
	T \coloneq e^{-S} \;, \qquad
	S \coloneq \frac{1}{2}\sum_{\ell\in \Zstar}\sum_{r,s\in L_\ell}K(\ell)_{r,s}\left(b_r(\ell)b_{-s}(-\ell)-b^*_{-s}(-\ell)b^*_{r}(\ell)\right) \;,
\end{equation}
with the Bogoliubov kernel
\begin{equation} \label{eq:K}
	K(\ell) \coloneq - \half \log \bigg( h(\ell)^{-\half}
		\Big( h(\ell)^{\half} \big( h(\ell) + 2 P(\ell) \big) h(\ell)^{\half}\Big)^{\half}
		h(\ell)^{-\half} \bigg) \;.
\end{equation}
(The exponent $ S $ corresponds to $ R^* \cK R $ in \cite{CHN23} since we prefer to use a particle-hole transformation instead of the normal ordering with respect to the Fermi ball.)
The matrix $K(\ell)$ is symmetric and possesses the reflection invariance
\begin{equation}
  K(-\ell)_{-p,-q} = K(\ell)_{p,q} \qquad \textnormal{for all } \ell \in \Zbb^3 \textnormal{ and all }p,q \in L_\ell\;.
\end{equation}

\section{Duhamel Expansion}\label{sec:extraction}

In this section we expand the momentum distribution $ \langle \Psi_N, a_q^* a_q \Psi_N \rangle $ of the trial state $ \Psi_N = R e^{-S} \Omega $ constructed in the previous section.
The particle--hole transformation $ R $ acts in a simple way: by~\eqref{eq:R} we have
\begin{equation} \label{eq:momentum_dist_R_trafo}
	R^* a_q^* a_q R
	= \mathds{1}_{B_{\F}}(q) \big( 1 - a_q^* a_q  \big)
		+ \mathds{1}_{B_{\F}^c}(q)  a_q^* a_q \;.
\end{equation}
So it suffices to consider the excitation vector $ \xi \coloneq e^{-S} \Omega $ and  compute the excitation distribution $ \langle \xi, a_q^* a_q \xi \rangle $. The fundamental theorem of calculus implies the Duhamel formula
\begin{equation} \label{eq:duhamelexpansion_blueprint}
\begin{aligned}
	e^{S} a_q^* a_q e^{-S}
& = a_q^* a_q
		+ \int_0^1 \di \lambda_1 \,   [S, e^{\lambda_1 S} a_q^* a_q e^{-\lambda_1 S}] \;.
\end{aligned}
\end{equation}
Iterating this formula leads to the expansion in Proposition~\ref{prop:finexpan}, the main result of this section. The multicommutators are computed using the CAR~\eqref{eq:CAR}, where we extract $ n^{\RPA}(q) $ as the terms that could be expected treating the pair operators as exactly bosonic.

\subsection{Extraction of the Bosonized Contribution}
\label{sec:extraction_bos}

To compute the multicommutators in~\eqref{eq:duhamelexpansion_blueprint} we use the CAR, which will produce quadratic quasi-bosonic expressions.

\begin{definition} \label{def:Q}
Let $A=(A(\ell))_{\ell \in \Zstar} $ be a family of symmetric operators with $A(\ell): \ell^2(L_\ell)\rightarrow \ell^2(L_\ell)$. The quadratic quasi-bosonic operators are defined as
\begin{equation}
\begin{split}
	Q_1(A)&\coloneq 2 \sum_{\ell \in \Zstar}\sum_{r,s \in L_{\ell}}A(\ell)_{r,s} b^*_r(\ell)b_{s}(\ell) \;,\\
	Q_2(A)&\coloneq \sum_{\ell \in \Zstar}\sum_{r,s \in L_{\ell}}A(\ell)_{r,s} \left(b_r(\ell)b_{-s}(-\ell)+b^*_{-s}(-\ell)b^*_{r}(\ell)\right) \;.
\end{split}
\end{equation} 
\end{definition}
Our $ Q_1 $ corresponds to $ 2 \tilde Q_1 $ in~\cite{CHN22}, while the definition of $ Q_2 $ is identical to \cite{CHN22}. (Due to the particle-hole transformation, our $ a_p $ agrees with $ R^* c_p R $ in~\cite{CHN22,CHN23,CHN24}.) The pair operators satisfy the following approximately bosonic commutator relations (compare to \cite[(1.66)]{CHN22} and \cite[Lemma~4.1]{BNP+20}):

\begin{lemma}[Approximate CCR]\label{lem:paircomm}
For $k,\ell \in \Zstar$ and $p \in L_{k}$, $q\in L_{\ell}$, we have
\begin{equation}
	[b_{p}(k),b_{q}(\ell)]
	= 0 = [b^*_{p}(k),b^*_{q}(\ell)]  \;, \qquad
	[b_{p}(k),b^*_{q}(\ell)]
	= \delta_{p,q}\delta_{k,\ell} + \epsilon_{p,q}(k,\ell) \;,
\end{equation}
 with error operator
\begin{equation}
	\epsilon_{p,q}(k,\ell)
	\coloneq -\left(\delta_{p,q}a^*_{q-\ell}a_{p-k} + \delta_{p-k,q-\ell}a^*_{q}a_{p}\right) \;.
\end{equation}
The error operator satisfies
\begin{equation}
\epsilon_{p,q}(\ell,k) = \epsilon^*_{q,p}(k,\ell) \qquad \text{and} \qquad \epsilon_{p,p}(k,k)\leq 0 \;.
\end{equation}
\end{lemma}
The proof is a computation using the CAR. As a consequence we obtain the next lemma.

\begin{lemma}[Commutators of $S $ and $b^*_p(k)$]
Let the operator $S$ be defined as in \eqref{eq:T}. For $k \in \Zstar$ and $p \in L_k$ we have
\begin{equation} \label{eq:comm_Kb}
	[S, b^*_p(k)]
	= \sum_{s\in L_{k}}K(k)_{p,s}b_{-s}(-k)
		+ \mathcal{E}_{p}(k)
\end{equation}
with error operator
\begin{equation}\label{eq:commerrKb}
	\mathcal{E}_{p}(k)
	\coloneq \frac{1}{2}\sum_{\ell\in \Zstar}\sum_{r,s\in L_\ell}K(\ell)_{r,s}\left\{\epsilon_{r,p}(\ell,k),b_{-s}(-\ell)\right\} \;.
\end{equation}
\end{lemma}

For the quadratic quasi-bosonic operators, this implies the following formulas.

\begin{lemma}[Commutator of $S$ and $Q$]\label{lem:Q1Kcomm}
Let $ A = (A(\ell))_{\ell \in \Zstar} $ be a family of symmetric operators $ A(\ell) : \ell^2(L_\ell) \to \ell^2(L_\ell) $ satisfying $A(\ell)_{r,s} = A(-\ell)_{-r,-s}$. Then
\begin{equation}
\begin{aligned}
	[S, Q_1(A)] 
	&= Q_2(\{A,K\})
		+ E_{Q_1}(A) \;, \\
	[S, Q_2(A)] 
	&= Q_1\left(\{A,K\} \right) 
		+ \sum_{\ell \in \Zstar} \sum_{r \in L_{\ell}} \big\{ A(\ell), K(\ell) \big\}_{r,r}
		+ E_{Q_2}(A) \;,
\end{aligned}
\end{equation}
with the family $ \{A,K\} = (\{A(\ell),K(\ell)\})_{\ell \in \Zstar} $ and with the error operators
\begin{equation}\label{eq:errKQ}
\begin{aligned}
	E_{Q_1}(A)
	&\coloneq 2 \sum_{\ell \in \Zstar}\sum_{r,s \in L_{\ell}}A(\ell)_{r,s}\Big(\mathcal{E}_{r}(\ell)b_{s}(\ell) + b^*_{s}(\ell)\mathcal{E}^*_{r}(\ell)\Big) \;, \\
	E_{Q_2}(A)
	& \coloneq \sum_{\ell \in \Zstar}\sum_{r,s \in L_{\ell}}\Big(A(\ell)_{r,s}\big(\big\{\mathcal{E}^*_{r}(\ell), b_{-s}(-\ell)\big\}
		+ \big\{ b^*_{-s}(-\ell) , \mathcal{E}_r(\ell) \big\} \big) \\
		& \hspace{6em}
		+ \big\{A(\ell)_,K(\ell)\big\}_{r,s}\epsilon_{r,s}(\ell,\ell)\Big) \;. \\
\end{aligned} 
\end{equation}
\end{lemma}

To simplify the expansion, we introduce the $n$-fold anticommutator
\begin{equation} \label{eq:Theta}
	\Theta_K^n (A)
	\coloneq \{ K, \Theta_K^{n-1} (A) \} \qquad
	\textnormal{with} \qquad
	\Theta_K^0 (A)
	\coloneq A \;,
\end{equation}
understood pointwise for families of operators as above.
Given $q \in \Zbb^3$ we define
\begin{equation} \label{eq:Pq}
	P^q(\ell) : \ell^2(L_\ell) \to \ell^2(L_\ell) \;, \qquad
	P^q(\ell)_{r,s} \coloneq \delta_{q,r} \delta_{q,s} \qquad
	\textnormal{for } \ell \in \Zstar \;,
\end{equation}
understood as $ P^q = 0 $ if $ q \notin L_\ell $.
Moreover we define the simplex integral
\begin{equation} \label{eq:Deltan}
	\int_{\Delta^n} \di^n \ulambda
	\coloneq \int_0^1 \di \lambda_1 \int_0^{\lambda_1} \di \lambda_2 \ldots \int_0^{\lambda_{n-1}} \di \lambda_n \;, \qquad
	\ulambda \coloneq (\lambda_1, \ldots, \lambda_n) \;.
\end{equation}
 The final Duhamel expansion can then be written in the following way.

\begin{proposition}[Duhamel expansion]\label{prop:finexpan}
For $q \in B^c_{\F}$ we have
\begin{align} \label{eq:finexpan}
	\eva{\Omega, e^{S} a_q^* a_q e^{-S} \Omega} 
	&= \half\sum_{\ell\in \Zstar}\mathds{1}_{L_\ell}(q) \sum_{\substack{m=2\\m:\textnormal{ even}}}^n \frac{((2K(\ell))^m)_{q,q}}{m!}
		+ \half \sum_{m=1}^{n-1} \eva{\Omega, E_m(P^q)\Omega}\nonumber\\
	&\quad +\half \int_{\Delta^n} \di^n\underline{\lambda} \;
		\eva{\Omega, e^{\lambda_n S}Q_{\sigma(n)}(\Theta^n_{K}(P^q)) e^{-\lambda_n S} \Omega} \;,
\end{align}
where $ \sigma(n) = 1 $ if $n$ is even and $ \sigma(n) = 2 $ if $n$ is odd, and the error operator is
\begin{equation}\label{eq:errEm}
	E_m(P^q) \coloneq \int_{\Delta^{m+1}} \di^{m+1} \underline{\lambda} \;
		e^{\lambda_{m+1} S} E_{Q_{\sigma(m)}}\left(\Theta^{m}_{K}(P^q)\right) e^{-\lambda_{m+1} S} \;.
\end{equation}
\end{proposition}

In Lemma~\ref{lem:nqb_integralrecovery}, we show that the first term on the r.~h.~s. of~\eqref{eq:finexpan} converges to $ n^{\RPA}(q) $ as $ n \to \infty $. The exchange contribution $ n^{\ex}(q) $ emerges by normal ordering the operators $E_m(P^q)$.

\begin{proof}
Our trial state and excitation density are reflection symmetric, that is, if we define the spatial reflection $ \fR: \cF \to \cF $ by $ \fR^* a_q^* \fR = a^*_{-q} $ and $ \fR \Omega = \Omega $, then
\begin{equation} \label{eq:reflectionsymmetry}
	\fR e^{-S} \Omega = e^{-S} \Omega
\end{equation}
and therefore
\begin{equation}
	\eva{\Omega, e^{S} a_q^* a_q e^{-S} \Omega} = \half \eva{\Omega, e^{S} (a_q^* a_q + a_{-q}^* a_{-q}) e^{-S} \Omega} \;.
\end{equation}
The first commutator in the Duhamel expansion takes the convenient form
\begin{equation} \label{eq:firstcommutator}
	[S, a_q^* a_q] + [S, a_{-q}^* a_{-q}]
	= Q_2(\{K,\tilde{P}^q\}) \;, \qquad
	\tilde{P}^q \coloneq \half(P^q + P^{-q}) \;.
\end{equation}
We then iteratively Duhamel-expand the $ Q_1$-- and $ Q_2 $--terms using Lemma~\ref{lem:Q1Kcomm} as
\begin{align*}
	e^{\lambda S} Q_1(A) e^{-\lambda S}
	&= Q_1(A) + \int_0^{\lambda} \di \lambda' e^{\lambda' S} Q_2(\{A,K\}) e^{-\lambda' S}
		+ \int_0^{\lambda} \di \lambda' e^{\lambda' S} E_{Q_1}(A) e^{-\lambda' S} \;, \\
	e^{\lambda S} Q_2(A) e^{-\lambda S}
	&= Q_2(A) + \int_0^{\lambda} \di \lambda' e^{\lambda' S} Q_1(\{A,K\}) e^{-\lambda' S}
		+ \int_0^{\lambda} \di \lambda' e^{\lambda' S} E_{Q_2}(A) e^{-\lambda' S} \\
	&\quad + \lambda \sum_{\ell \in \Zstar} \sum_{r \in L_{\ell}} \big\{ A(\ell), K(\ell) \big\}_{r,r} \;. \tagg{eq:expand}
\end{align*}
The $ E_{Q_1}$-- and $ E_{Q_2} $--terms are not expanded further but collected in $ E_m(P^q) $. The term on the last line can be written as a trace and contributes to $n^\textnormal{RPA}(q)$. After $ n $ expansion steps
\begin{align}
	&e^{S} (a_q^* a_q + a_{-q}^* a_{-q}) e^{-S} \nonumber\\
	&= a_q^* a_q + a_{-q}^* a_{-q}
		+ \sum_{\ell\in \Zstar} \mathds{1}_{L_\ell \cup L_{-\ell}}(q) \sum_{\substack{m=2\\m:\textnormal{ even}}}^n \frac{\mathrm{Tr} \big(\Theta^m_{K(\ell)} \big( \tilde{P}^q(\ell) \big) \big)}{m!}
		+ \sum_{m=1}^{n-1} E_m(\tilde{P}^q) \nonumber\\
	&\quad+ \sum_{m=1}^{n-1}\frac{1}{m!}
		Q_{\sigma(m)}  \big( \Theta^m_{K}(\tilde{P}^q) \big)
		+\int_{\Delta^n} \di^n \underline{\lambda} \;
		e^{\lambda_n S}Q_{\sigma(n)}(\Theta^n_K (\tilde{P}^q)) e^{-\lambda_n S} \;.	\label{eq:operatorexpansion}
\end{align}
In the vacuum expectation value, $ a_q^* a_q + a_{-q}^* a_{-q} $ and the $ Q_{\sigma(m)} $--terms vanish. Thus
\begin{align}
	& \half \langle \Omega, e^{S} (a_q^* a_q + a_{-q}^* a_{-q}) e^{-S} \Omega \rangle \nonumber\\
	&= \half \sum_{\ell\in \Zstar} \mathds{1}_{L_\ell \cup L_{-\ell}}(q) \sum_{\substack{m=2\\m:\textnormal{ even}}}^n \frac{\mathrm{Tr} \big(\Theta^m_{K(\ell)} \big( \tilde{P}^q(\ell) \big) \big)}{m!}
	+ \half \sum_{m=1}^{n-1} \langle \Omega, E_m(\tilde{P}^q) \Omega \rangle \nonumber\\
	&\quad + \half \int_{\Delta^n} \di^n \underline{\lambda} \;
		\langle \Omega, e^{\lambda_n S}Q_{\sigma(n)}(\Theta^n_{K}(\tilde{P}^q)) e^{-\lambda_n S} \Omega \rangle \;.
\end{align}
Using reflection symmetry, we may replace $ \tilde{P}^q $ by $ P^q $ and restrict to $ q \in L_\ell $ because $ q \notin L_\ell$ implies $P^q(\ell) = 0 $. The result follows since by cyclicity $ \mathrm{Tr} \big(\Theta^m_{K(\ell)} \big( P^q(\ell) \big) \big) = ((2K(\ell))^m)_{q,q} $.
\end{proof}

\subsection{Normal Ordering the Error Operators}
\label{sec:extraction_ex}
As shown in \eqref{eq:errEm}, $E_m(P_q)$ depends on $E_{Q_1}$ for $m$ even and on $E_{Q_2}$ for $m$ odd. To estimate $ E_{Q_1} $ and $ E_{Q_2} $, we need to normal order them first. In this step the exchange contribution $ n^{\ex}(q) $ emerges.

\begin{lemma}[Error terms] \label{lem:normalordering_errors}
Recall definition~\eqref{eq:errKQ}.
 For $ q \in B_{\F}^c $, we may write
\begin{equation} \label{eq:EQ1EQ2extension}
\begin{split}
	E_{Q_1}(\Theta^m_{K}(P^q)) &
	= \sum_{j=1}^3 E_{Q_1}^{m,j}(q) + \mathrm{h.c.} 	\qquad (\text{for even }m \in \NNN)\;, \\
	E_{Q_2}(\Theta^m_{K}(P^q)) &
	= \bigg( \sum_{j=1}^{11} E_{Q_2}^{m,j}(q)  + \mathrm{h.c.} \bigg) + n^{\ex,m}(q) \qquad (\text{for odd }m \in \NNN)\;,
\end{split}
\end{equation}
with
\begin{align}
	E_{Q_1}^{m,1}(q)
	&\coloneq -2 \sum_{\ell, \ell_1\in \Zstar}\sum_{\substack{r\in L_{\ell} \cap L_{\ell_1}\\ s \in L_{\ell},\,s_1\in L_{\ell_1}}} \Theta^m_{K}(P^q)(\ell)_{r,s} K(\ell_1)_{r,s_1} a^*_{r-\ell_1} b^*_{s}(\ell) b^*_{-s_1}(-\ell_1) a_{r-\ell}
	\;, \nonumber\\
	E_{Q_1}^{m,2}(q)
	&\coloneq -2 \sum_{\ell, \ell_1\in \Zstar}\sum_{\substack{r\in (L_{\ell}-\ell) \cap (L_{\ell_1}-\ell_1)\\ s \in L_{\ell},\,s_1\in L_{\ell_1} }} \Theta^m_{K}(P^q)(\ell)_{r+\ell,s}K(\ell_1)_{r+\ell_1,s_1}
	a^*_{r+\ell_1}b^*_{s}(\ell) b^*_{-s_1}(-\ell_1) a_{r+\ell}
	\;, \nonumber\\
	E_{Q_1}^{m,3}(q)
	&\coloneq  2 \sum_{\ell, \ell_1\in \Zstar}\sum_{\substack{r\in L_{\ell} \cap L_{\ell_1} \cap (-L_{\ell_1}+\ell+\ell_1)\\ s \in L_{\ell}}} \Theta^m_{K}(P^q)(\ell)_{r,s}K(\ell_1)_{r,-r+\ell+\ell_1} b^*_{s}(\ell) a^*_{r-\ell_1}a^*_{r-\ell-\ell_1} \;, \label{eq:expandedEQ1}
\end{align}
and
\allowdisplaybreaks
\begin{align}
	E_{Q_2}^{m,1}(q)
	&\coloneq 2\sum_{\ell,\ell_1 \in \Zstar}\sum_{\substack{r\in L_{\ell} \cap L_{\ell_1}\\ s \in L_{\ell},\,s_1\in L_{\ell_1}}} \Theta^m_{K}(P^q)(\ell)_{r,s}K(\ell_1)_{r,s_1} a^*_{r-\ell_1}b^*_{-s_1}(-\ell_1)b_{-s}(-\ell)a_{r-\ell} \;, \nonumber\\
	E_{Q_2}^{m,2}(q)
	&\coloneq 2\sum_{\ell,\ell_1 \in \Zstar}\sum_{\substack{r\in (L_{\ell}-\ell) \cap (L_{\ell_1}-\ell_1)\\ s \in L_{\ell},\,s_1\in L_{\ell_1}}} \Theta^m_{K}(P^q)(\ell)_{r+\ell,s} K(\ell_1)_{r+\ell_1,s_1} a^*_{r+\ell_1} b^*_{-s_1}(-\ell_1) b_{-s}(-\ell) a_{r+\ell}\;, \nonumber\\
	E_{Q_2}^{m,3}(q)
	&\coloneq -2\sum_{\ell,\ell_1 \in \Zstar}\sum_{\substack{r\in L_{\ell} \cap L_{\ell_1} \cap (-L_{\ell_1}+\ell+\ell_1)\\ s \in L_{\ell}}} \Theta^m_{K}(P^q)(\ell)_{r,s} K(\ell_1)_{r,-r+\ell+\ell_1} a^*_{r-\ell_1}a^*_{r-\ell-\ell_1}b_{-s}(-\ell)\;, \nonumber\\
	E_{Q_2}^{m,4}(q)
	&\coloneq -2 \sum_{\ell,\ell_1 \in \Zstar}\sum_{\substack{r\in L_{\ell} \cap L_{\ell_1}\cap (-L_{\ell}+\ell+\ell_1)\\s_1\in L_{\ell_1}}} \Theta^m_{K}(P^q)(\ell)_{r,-r+\ell+\ell_1} K(\ell_1)_{r,s_1} b^*_{-s_1}(-\ell_1)a_{r-\ell-\ell_1}a_{r-\ell}\;, \nonumber\\
	E_{Q_2}^{m,5}(q)
	&\coloneq - 2\sum_{\ell,\ell_1 \in \Zstar}\sum_{\substack{r\in L_{\ell} \cap L_{\ell_1}\\ s \in (L_{\ell}-\ell) \cap (L_{\ell_1}-\ell_1)}} \Theta^m_{K}(P^q)(\ell)_{r,s+\ell}K(\ell_1)_{r,s+\ell_1}a^*_{r-\ell_1}a^*_{-s-\ell_1} a_{-s-\ell}a_{r-\ell}\;, \nonumber\\
	E_{Q_2}^{m,6}(q)
	&\coloneq -\sum_{\ell,\ell_1 \in \Zstar}\sum_{r,s\in L_{\ell} \cap L_{\ell_1}} \Theta^m_{K}(P^q)(\ell)_{r,s}K(\ell_1)_{r,s}a^*_{r-\ell_1}a^*_{-s+\ell_1} a_{-s+\ell}a_{r-\ell}\;, \nonumber\\
	E_{Q_2}^{m,7}(q)
	&\coloneq -\sum_{\ell,\ell_1 \in \Zstar}\sum_{\substack{r,s\in (L_{\ell}-\ell) \cap (L_{\ell_1}-\ell_1)}} \Theta^m_{K}(P^q)(\ell)_{r+\ell,s+\ell} K(\ell_1)_{r+\ell_1,s+\ell_1} a^*_{r+\ell_1}a^*_{-s-\ell_1}a_{-s-\ell}a_{r+\ell}\;, \nonumber\\
	E_{Q_2}^{m,8}(q)
	&\coloneq -2\sum_{\ell,\ell_1 \in \Zstar}\sum_{\substack{r\in L_{\ell} \cap L_{\ell_1} \cap (-L_{\ell}+\ell+\ell_1) \\\cap (-L_{\ell_1}+\ell+\ell_1)}} \Theta^m_{K}(P^q)(\ell)_{r,-r+\ell+\ell_1}K(\ell_1)_{r,-r+\ell+\ell_1} a^*_{r-\ell_1}a_{r-\ell_1}\;, \nonumber\\
	E_{Q_2}^{m,9}(q)
	&\coloneq -2\sum_{\ell,\ell_1 \in \Zstar} \sum_{\substack{r\in L_{\ell}\cap L_{\ell_1} \cap (-L_{\ell}+\ell +\ell_1) \\\cap (-L_{\ell_1}+\ell+\ell_1)}} \Theta^m_{K}(P^q)(\ell)_{r,-r+\ell+\ell_1}K(\ell_1)_{r,-r+\ell+\ell_1} a^*_{r-\ell-\ell_1}a_{r-\ell-\ell_1} \;, \nonumber\\
	E_{Q_2}^{m,10}(q)
	&\coloneq \sum_{\ell \in \Zstar} \sum_{r\in L_{\ell}}\Theta^{m+1}_{K}(P^q)(\ell)_{r,r} a^*_{r-\ell}a_{r-\ell} \;, \nonumber\\
	E_{Q_2}^{m,11}(q)
	&\coloneq \sum_{\ell \in \Zstar} \sum_{r\in L_{\ell}}\Theta^{m+1}_{K}(P^q)(\ell)_{r,r} a^*_{r}a_{r} \;, \label{eq:expandedEQ2}
\end{align}
\allowdisplaybreaks[1]
as well as
\begin{align}
	n^{\ex,m}(q)
	&\coloneq 2 \sum_{\ell,\ell_1 \in \Zstar}\sum_{r\in L_{\ell} \cap L_{\ell_1} \cap (-L_{\ell}+\ell+\ell_1) \cap (-L_{\ell_1}+\ell+\ell_1 )} \!\!\!\Theta^m_{K}(P^q)(\ell)_{r,-r+\ell+\ell_1}K(\ell_1)_{r,-r+\ell+\ell_1} \;. \label{eq:nqexm}
\end{align}
\end{lemma}
\begin{proof}
Follows by a lengthy but straightforward computation using the CAR. (Alternatively, this can be conveniently computed using Friedrichs diagrams~\cite{BL25}.)
\end{proof}

\section{Preliminary Bounds}
\label{sec:prelim_bounds}

We collect some estimates to be used for estimating the error terms in~\eqref{eq:finexpan}.

\begin{lemma}[Kinetic sums] \label{lem:lambdainverse}
Let $ \ell \in \Zstar $ and recall the definitions~\eqref{eq:Lell} and~\eqref{eq:eq} of $ \lambda_{\ell,r} $ and $ e(r) $. Consider any set $ S \subset \Z^3 $ with $ |S| \le C_0 \kF^3 $ for some $ C_0 > 0 $. Then, given $ \varepsilon > 0 $, there exist $ C, C_{C_0, \varepsilon} > 0 $ such that
\begin{equation} \label{eq:lambdainverse}
	\sum_{r \in L_\ell} \lambda_{\ell,r}^{-1} \le C \kF \;, \qquad
	\sum_{r \in S} e(r)^{-1} \le C_{\varepsilon, C_0} \kF^{1+\varepsilon} \;.
\end{equation}
\end{lemma}
\begin{proof}
The first statement is~\cite[Prop.~A.2]{CHN22}, the second was proven with $ C \kF^3 $ replaced by $ |\overline{B_{2 \kF}(0)} \cap \Z^3| $ in~\cite[Lemma~3.2]{CHN24}. It generalizes to any set $S$ with $ |S| \le C_0 \kF^3 $ by writing $ S $ as a disjoint union of a uniformly bounded numbers of sets each with at most $ |\overline{B_{2 \kF}}(0) \cap \Z^3| $ points.
\end{proof}

\begin{definition}
For $ \ell \in \Zstar$ and $A(\ell) : \ell^2(L_\ell) \to \ell^2(L_\ell)$ we define the norms
\begin{equation}
\begin{aligned}
	\norm{A(\ell)}_{\max,2}
	\coloneq \bigg(\sum_{p \in L_\ell}
	\sup\limits_{q \in L_\ell}
	\abs{A(\ell)_{p,q}}^2\bigg)^\half \;, \qquad
	\norm{A(\ell)}_{\mathrm{max,1}}
	\coloneq \sum_{p \in L_\ell}
	\sup\limits_{q \in L_\ell}
	\abs{A(\ell)_{p,q}} \;.
\end{aligned}
\end{equation}
Moreover we have the Hilbert--Schmidt norm $ \norm{A(\ell)}_{\HS} := \Big( \sum_{p,q \in L_\ell} |A(\ell)_{p,q}|^2 \Big)^{\half} $.
\end{definition}

We have the following estimates for the Bogoliubov kernel \eqref{eq:K}.
\begin{lemma}[Bounds on $ K $]\label{lem:normsk}
Let $ \ell \in \Zstar $, $ m \in \mathbb{N} \setminus \{0\}$, and $ r,s \in L_\ell $. Then
\begin{equation} \label{eq:K_element_bounds}
	|(K(\ell)^m)_{r,s}|
	\le \frac{(C \hat{V}(\ell))^m \kF^{-1}}{\lambda_{\ell,r} + \lambda_{\ell,s}} \;.
\end{equation}
Moreover, we have the estimates
\begin{equation} \label{eq:K_max_bounds}
\begin{aligned}
	&\Vert K(\ell)^m \Vert_{\max,2}
	&\le \; &(C \hat{V}(\ell))^m \kF^{-\half} \;, \qquad
	&&\Vert K(\ell)^m \Vert_{\max,1}
	&&\le (C \hat{V}(\ell))^m \;, \\
	&\norm{K(\ell)^m}_{\HS}
	&\le \; &(C \hat{V}(\ell))^m \;,
\end{aligned} 
\end{equation}
as well as for $ q \in L_\ell $ the estimates
\begin{equation} \label{eq:e(q)_extraction_bounds}
\begin{split}
	|(K(\ell)^m)_{r,q}|
	& \le (C \hat{V}(\ell))^m \kF^{-1} e(q)^{-1} \;, \\
	\left( \sum_{r \in L_\ell} |(K(\ell)^m)_{r,q}|^2 \right)^{\half}
	& \le (C \hat{V}(\ell))^m \kF^{-\half} e(q)^{-\half} \;.
\end{split}
\end{equation}
\end{lemma}
\begin{proof}
From~\cite[Prop.~7.10]{CHN23} we retrieve \eqref{eq:K_element_bounds} for $ m = 1 $. For $ m \ge 2 $, we proceed by induction: Suppose \eqref{eq:K_element_bounds} holds until $ m-1 $. Then, using \eqref{eq:lambdainverse}, we get
\begin{align*}
		|(K(\ell)^m)_{r,s}|
		&\le \sum_{r' \in L_\ell}
		|(K(\ell)^{m-1})_{r,r'}| \;
		|K(\ell)_{r',s}| \tagg{eq:notimportant1}
		\le (C \hat{V}(\ell))^m \kF^{-2} \sum_{r' \in L_\ell}
		\frac{1}{\lambda_{\ell, r} + \lambda_{\ell, r'}}
		\frac{1}{\lambda_{\ell, r'} + \lambda_{\ell, s}} \\
		&\le (C \hat{V}(\ell))^m \kF^{-2} \sum_{r' \in L_\ell}
		\frac{1}{\lambda_{\ell, r'}} \frac{1}{\lambda_{\ell, r} + \lambda_{\ell, s}}
		\le (C \hat{V}(\ell))^m \kF^{-1}
		\frac{1}{\lambda_{\ell, r} + \lambda_{\ell, s}} \;.
\end{align*}
The first bound in \eqref{eq:e(q)_extraction_bounds} follows analogously, using $\lambda_{\ell,q}^{-1} \leq 2 e(q)^{-1}$.

The first bound in~\eqref{eq:K_max_bounds} follows using $ \lambda_{\ell,q}^{-1} \le 2 $ by
\begin{align*}
	\sum_{q \in L_\ell} \sup_{r \in L_\ell} |(K(\ell)^m)_{q,r}|^2
	&\le \sum_{q \in L_\ell} (C \hat{V}(\ell))^{2m} \kF^{-2} \sup_{r \in L_\ell} (\lambda_{\ell,q} + \lambda_{\ell,r})^{-2} \\
	&
	\le (C \hat{V}(\ell))^{2m} \kF^{-2} \sum_{q \in L_\ell} \lambda_{\ell,q}^{-1} \sup_{r \in L_\ell} \lambda_{\ell,q}^{-1} \le (C \hat{V}(\ell))^{2m} \kF^{-1} \;. \tagg{eq:notimportant}
\end{align*}
The second bound in \eqref{eq:e(q)_extraction_bounds} follows analogously using $\lambda_{\ell,q}^{-1} \leq 2 e(q)^{-1}$.
The second and third bound in~\eqref{eq:K_max_bounds} are easier.
% from
% \begin{displaymath}
% 	\norm{K(\ell)^m}_{\HS}^2
% 	\le \sum_{r,s \in L_\ell} (C \hat{V}(\ell))^{2m} \kF^{-2} (\lambda_{\ell,r} + \lambda_{\ell,s})^{-2}
% 	\le (C \hat{V}(\ell))^{2m} \kF^{-2} \Big( \sum_{r \in L_\ell} \lambda_{\ell,r}^{-1} \Big)^2
% 	\le (C \hat{V}(\ell))^{2m}
% \end{displaymath}
% and
% \begin{displaymath}
% 	\normmaxi{K(\ell)^m}
% 	\leq \sum_{r \in L_\ell} \sup_{q \in L_\ell} (C \hat{V}(\ell))^{m} \kF^{-1} (\lambda_{\ell,r} + \lambda_{\ell,q})^{-1}
% 	\le (C \hat{V}(\ell))^{m} \kF^{-1} \sum_{r \in L_\ell} \lambda_{\ell,r}^{-1} \leq (C \hat{V}(\ell))^{m} \,.\nonumber
% \end{displaymath}
% This concludes the proof.
\end{proof}

Next, we collect some elementary estimates involving the fermionic number operator
\begin{equation} \label{eq:cN}
	\cN \coloneq \sum_{q \in \Z^3} a_q^* a_q \;.
\end{equation}

\begin{lemma}[Quadratic operator estimates]\label{lem:estQ2}
Let $A = (A(\ell))_{\ell \in \Zstar}$ be a family of symmetric operators $ A(\ell) : \ell^2(L_\ell) \to \ell^2(L_\ell) $. Then for $ \Psi \in \cF $ we have
\begin{equation} \label{eq:Qest}
\begin{aligned}
	|\eva{\Psi,Q_1(A)\Psi}|
	&\leq 2\sum_{\ell\in \Zstar}\norm{A(\ell)}_{\HS}\eva{\Psi,\mathcal{N} \Psi} \;, \\
	|\eva{\Psi,Q_2(A)\Psi}|
	&\leq 2\sum_{\ell\in \Zstar}\norm{A(\ell)}_{\HS}\eva{\Psi,(\mathcal{N}+1) \Psi} \;.
\end{aligned}
\end{equation}
\end{lemma}

\begin{proof}
For the first bound, see \cite[Prop.~4.7]{CHN22}; the second follows analogously.
\end{proof}

The next estimate generalizes \cite[Prop.~5.8]{CHN22}, conceptually going back to \cite{BJP+16,BPS14b} in the context of the derivation of the time-dependent Hartree--Fock equation. It shows that the expectation value $ \langle \Omega, e^{\lambda S} (\mathcal{N} + 1)^m e^{-\lambda S} \Omega \rangle$ does not grow with $N$.

\begin{lemma}[Gr\"onwall estimate]\label{lem:gronNest}
Let $S$ be defined as in \eqref{eq:T}. For every $ m \in \NNN $, there exists a constant $ C_m > 0 $ such that for all $ \lambda\in [0,1]$ we have
\begin{equation}\label{eq:gronest}
	e^{\lambda S} (\mathcal{N} +1)^m e^{-\lambda S}
	\leq C_m (\NN+1)^m \;.
\end{equation}
\end{lemma}
\begin{proof}
First, by the pull-through formula $a^*_k \Ncal = (\Ncal - 1) a^*_k$, we have
\begin{align}
	& \left[(\NN+4)^m, b^*_{-s}(-\ell)b^*_{r}(\ell)\right] \nonumber\\
	&= \left( (\NN+4)^m - \NN^m \right) b^*_{-s}(-\ell)b^*_{r}(\ell) \nonumber\\
	&= \left( \left(\NN+4\right)^m - \NN^m \right)^\half b^*_{-s}(-\ell)b^*_{r}(\ell) \left( \left(\NN+8\right)^m - \left(\NN+4\right)^m \right)^\half \;.
\end{align}
For $ \Psi_0 \in \cF $ and $ \Psi_\lambda \coloneq e^{-\lambda S} \Psi_0 $, using the definition~\eqref{eq:T} of $ S $, then the Cauchy--Schwarz inequality, and then $ b_{-s}(-\ell) = a_{-s+\ell} a_{-s} $ with $ \Vert a_{-s+\ell} \Vert_{\textnormal{op}} \le 1 $, we get
\begin{align*}
	&\left|\frac{\di}{\di\lambda}\eva{\Psi_0, e^{\lambda S} (\mathcal{N}+4)^m e^{-\lambda S} \Psi_0 }\right|
	= \left| \eva{\Psi_0, e^{\lambda S} \left[S, (\NN+4)^m\right] e^{-\lambda S} \Psi_0}\right|\nonumber\\
	&\leq \sum_{\ell\in \Zstar}
		\sum_{r,s\in L_\ell} \abs{\eva{ b_{-s}(-\ell) \left( \left(\NN+4\right)^m - \NN^m \right)^\half \Psi_\lambda, K(\ell)_{r,s} b^*_{r}(\ell) \left( \left(\NN+8\right)^m - \left(\NN+4\right)^m \right)^\half \Psi_\lambda }}\nonumber\\
	&\leq \sum_{\ell\in \Zstar}
		\Bigg( \sum_{s\in L_\ell} \norm{ a_{-s} \left( \left(\NN+4\right)^m - \NN^m \right)^\half \Psi_\lambda}^2 \Bigg)^{\half} \times \nonumber\\
		&\quad \times \Bigg( \sum_{s\in L_\ell}
			\norm { \sum_{r\in L_\ell} K(\ell)_{r,s} b^*_{r}(\ell) \left( \left(\NN+8\right)^m - \left(\NN+4\right)^m \right)^\half \Psi_\lambda}^2 \Bigg)^{\half} \;. \nonumber
\end{align*}
Now, using $ \sum_{s \in L_\ell} \norm{a_{-s} \Phi}^2 \leq \Vert \cN^{\half} \Phi \Vert^2 $ for all $ \Phi \in \cF $, and \cite[Prop.~4.2]{CHN22} we recover
\begin{equation}
	\norm{\sum_{r\in L_\ell} K(\ell)_{r,s} b^*_{r}(\ell) \Phi}^2
	\le \sum_{r\in L_\ell} |K(\ell)_{r,s}|^2
		\norm{(\cN + 1)^{\half} \Phi}^2 \;.
\end{equation}
Moreover, there exists $ C > 0 $ depending on $ m $, such that
\begin{equation}
	\left( \left(\NN+4\right)^m - \NN^m \right)
	\leq \left(\NN+4\right)^{m-1} \;, \quad
	\left( \left(\NN+8\right)^m - \left(\NN+4\right)^m \right)
	\leq C \left(\NN+4\right)^{m-1} \;,
\end{equation}
so by Lemma~\ref{lem:normsk}, we get
\begin{align}
	&\left|\frac{\di}{\di\lambda}\eva{\Psi_0, e^{\lambda S} (\mathcal{N}+4)^m e^{-\lambda S} \Psi_0 }\right| \nonumber\\
	&\leq \sum_{\ell\in \Zstar}
		\norm{ \NN^\half \left( \left(\NN+4\right)^m - \NN^m \right)^\half \Psi_\lambda}
		\norm{K(\ell)}_{\HS}
		\norm{ (\NN+1)^\half \left( \left(\NN+8\right)^m - \left(\NN+4\right)^m \right)^\half \Psi_\lambda } \nonumber\\
	&\leq C \sum_{\ell\in \Zstar}
		\norm{K(\ell)}_{\HS}
		\norm{ \left(\NN+4\right)^\frac{m}{2} \Psi_\lambda}^2
	\leq C \eva{\Psi_0, e^{\lambda S} (\mathcal{N}+4)^m e^{-\lambda S} \Psi_0 } \;.
\end{align}
We conclude using Gr\"onwall's lemma and $ (\cN+1)^m \le (\cN+4)^m \le C (\cN+1)^m $.
\end{proof}

\section{Many-Body Error Estimates}
\label{subsec:manybody_estimates}

We now turn to estimating the two errors of the expansion in Proposition~\ref{prop:finexpan}, namely the bosonization error $ E_m $ and the expansion tail.

\subsection{Vanishing of the Expansion Tail}
\label{subsec:tailestimate}

The next proposition shows that the expansion tail vanishes as $ n \to \infty $.
\begin{proposition}[Tail]\label{prop:headerr}
Let $S$ be defined as in \eqref{eq:T} and recall definitions \eqref{eq:Theta}, \eqref{eq:Pq}, and \eqref{eq:Deltan}. There exists a constant $C$ such that for all $q \in B^c_{\F}$ and all $n \in \Nbb$ we have
\begin{equation}\label{eq:headest}
	\half \abs{\int_{\Delta^n} \di^n\underline{\lambda} \;
		\eva{\Omega, e^{\lambda_n S}Q_{\sigma(n)}(\Theta^n_{K}(P^q)) e^{-\lambda_n S} \Omega} }
	\leq C \frac{2^n}{n!} \sum_{\ell \in \Zstar} \norm{K(\ell)}^n_{\mathrm{op}} \;.
\end{equation}
\end{proposition}
The proof uses the following lemma.
\begin{lemma}[Iterated anticommutator]\label{lem:multicommest}
Let $ \ell \in \Zstar $. For any symmetric operator $ A(\ell): \ell^2(L_\ell) \to \ell^2(L_\ell) $ and $\Theta^n_K$ the $ n $-fold anticommutator as in \eqref{eq:Theta}, we have
\begin{equation}
	\norm{\Theta^{n}_K(A)(\ell)}_{\HS}
	\leq 2^n \norm{K(\ell)}^{n}_{\mathrm{op}}\norm{A(\ell)}_{\HS} \;.
\end{equation}
\end{lemma}
\begin{proof}
Using $\norm{AB}_{\HS} \leq \norm{A}_{\mathrm{op}} \norm{B}_{\HS}$, the bound follows by induction from
\[
	\norm{\Theta^{n}_K(A)(\ell)}_{\HS}
	= \norm{\left\{K(\ell),\Theta^{n-1}_K(A)(\ell)\right\}}_{\HS}
% 	\leq 2 \norm{K(\ell)\Theta^{n-1}_K(A)(\ell) }_{\HS} \\
	\leq 2 \norm{K(\ell)}_{\mathrm{op}}\norm{\Theta^{n-1}_K(A)(\ell)}_{\HS} \;. \qedhere
\]
\end{proof}

\begin{proof}[Proof of Proposition~\ref{prop:headerr}]
Combining Lemmas~\ref{lem:estQ2} and~\ref{lem:multicommest}, we have
\begin{align}
	&\half \abs{\int_{\Delta^n} \di^n \ulambda \;
		\eva{\Omega, e^{\lambda_n S} Q_{\sigma(n)}(\Theta^n_{K}(P^q)) e^{-\lambda_n S} \Omega} } \nonumber\\
	&\leq 2^n \int_{\Delta^n} \di^n \ulambda \sum_{\ell \in \Zstar} \norm{K(\ell)}^n_{\mathrm{op}} \norm{P^q(\ell)}_{\HS} 
		\abs{\eva{\Omega, e^{\lambda_n S} (\NN +1) e^{-\lambda_n S} \Omega}} \;.
\end{align}
With $ \norm{P^q}_{\HS} = 1$, Lemma~\ref{lem:gronNest}, and $ \int_{\Delta^n} \di^n\underline{\lambda} = \frac{1}{n!} $, we get
\begin{align*}
	&\half \abs{\int_{\Delta^n} \di^n \ulambda \;
		\eva{\Omega, e^{\lambda_n S} Q_{\sigma(n)}(\Theta^n_{K}(P^q)) e^{-\lambda_n S} \Omega} }
	\\
	& \leq C 2^n \int_{\Delta^n} \di^n \ulambda \sum_{\ell \in \Zstar} \norm{K(\ell)}^n_{\mathrm{op}} \eva{\Omega,(\NN+1)\Omega} =  C \frac{2^n}{n!} \sum_{\ell \in \Zstar} \norm{K(\ell)}^n_{\mathrm{op}} \;. \qedhere
\end{align*}
\end{proof}

\subsection{Bosonization Error Estimates}
\label{subsec:bos_error}

The largest part of our analysis addresses the error terms $ E_m $ in Proposition~\ref{prop:finexpan}. Following \cite{BL25}, we estimate them by a double bootstrap argument.

\begin{definition}[Bootstrap Quantity] We define the \emph{bootstrap quantity} as
\begin{equation} \label{eq:Xi}
	\Xi \coloneq \sup\limits_{q \in \Z^3} \sup\limits_{\lambda \in [0,1]}\expval{\Omega, e^{\lambda S} a^*_q a_q e^{-\lambda S} \Omega} \;.
\end{equation}
\end{definition}

The idea of our double bootstrap argument is as follows. As the starting point, $ 0 \le a_q^* a_q \le 1 $ implies the trivial estimate $ 0 \le \Xi \le 1 $. Obviously $n(q) \leq \Xi$. In the proof of the main result we will expand $n(q)$ around $n^\textnormal{RPA}(q)$ and control the deviation using $\Xi$. Then taking a supremum  of $n(q)$ over $q$ we can resolve for the improved bound $\Xi \leq C \kF^{-1}$ (this is the global bootstrap). Secondly, we return to the expansion and, by a similar argument as before, derive an improved bound proportional to $e(q)^{-1}$ on $n(q)$ at a fixed point $q \in \Zbb^3$ (this is the local bootstrap). To make this work, it is crucial that whenever possible in the error bounds given in the following propositions, we keep $\langle \Omega, e^{\lambda S} a^*_q a_q e^{-\lambda S} \Omega \rangle$ and use $\Xi$ only where momenta different from $q$ appear as arguments of the creation and annihilation operators. Finally, with the obtained bound on $n(q)$ we can once more return to estimate the error terms of the expansion of $n(q)$ around $n^\textnormal{RPA}(q)$, thereby obtaining our main result.

\begin{proposition} \label{prop:finalEmest}
Let $ \sum_{\ell \in \Zstar} \hat{V}(\ell)^2 |\ell|^\alpha < \infty $ for some $ \alpha > 0 $. Recall $E_m(P^q)$ from \eqref{eq:errEm} with $ \Theta^n_K $, $ P^q $, and $ \sigma(n) $ defined within and above Proposition~\ref{prop:finexpan}. Then, for
\[
	\xi_\lambda \coloneq e^{- \lambda S} \Omega \;,
\]
given $ \varepsilon > 0 $, there exist constants $ C, C_\varepsilon > 0 $ such that for all $ m \in \NNN $, $ q \in B_{\F}^c $, and $ \gamma \ge 0 $,
\begin{align} \label{eq:finalEmest_Coulomb}
	\abs{\half \eva{\Omega, E_m(P^q) \Omega} - \frac{\delta_{m,1}}{4} n^{\ex,1}(q)}
	&\leq C_\varepsilon \frac{C^m}{m!}
		\bigg( e(q)^{-1}\!\left( \kF^{-\frac 32 + \varepsilon}
		\!+ \kF^{-1 - \frac{\alpha \gamma}{2}}
		\!+ \kF^{-1 + \frac{3-\alpha}{2} \gamma} \Xi^\half
		\!+ \kF^{-1+\varepsilon} \Xi^\half \right) \nonumber\\
	&\hspace{4.5em} + e(q)^{-\half} \kF^{-1} \sup_{\lambda \in [0,1]} \eva{\xi_\lambda, a_q^* a_q \xi_\lambda}^{\half} \bigg) \;,
\end{align}
Moreover, if $ \sum_{\ell \in \Zstar} \hat{V}(\ell) < \infty $, we have the estimate
\begin{align} \label{eq:finalEmest}
	&\abs{\eva{\Omega, E_m(P^q) \Omega}} \nonumber\\
	&\leq C_\varepsilon \frac{C^m}{m!}
		\bigg( e(q)^{-1}
		\Big( \kF^{-2}
		+ \kF^{-\frac{3}{2}} \Xi^\half
		+ \kF^{-1} \Xi^{1-\varepsilon} \Big)
		+ e(q)^{-\half} \kF^{-1} \Xi^{\half - \varepsilon}
		\sup_{\lambda \in [0,1]} \eva{\xi_\lambda, a_q^* a_q \xi_\lambda}^{\half}  \bigg)\;.
\end{align}
\end{proposition}

To prove this bound, we write
\begin{equation} \label{eq:errEm2}
	\abs{\eva{\Omega, E_m(P^q) \Omega }}
	\le \int_{\Delta^{m+1}} \di^{m+1} \underline{\lambda} \;
		\abs{\eva{\xi_{\lambda_{m+1}}, E_{Q_{\sigma(m)}}\left(\Theta^{m}_{K}(P^q)\right) \xi_{\lambda_{m+1}}}} \;,
\end{equation}
where we recall the terms \eqref{eq:EQ1EQ2extension} of $ E_{Q_1}(\Theta^m_{K}(P^q)) $ and $ E_{Q_2}(\Theta^m_{K}(P^q)) $.
After the bootstrap, we will have $ \eva{\xi_\lambda, a_q^* a_q \xi_\lambda} \sim n^{\RPA}(q) \sim \kF^{-1} e(q)^{-1} $, hence all errors scale like $ e(q)^{-1} $.

In Proposition~\ref{prop:finEQ1est} to Lemma~\ref{lem:estnqex} we estimate the individual contributions before summarizing them to complete the proof of Proposition~\ref{prop:finalEmest}.

\subsubsection{Error Estimates for Even $m$}

Here, only the case $ m \ge 2 $ occurs, which will make bounds slightly easier, since $ \sum_\ell \hat{V}(\ell)^m < \infty $ is always true for our assumptions on the potential.

\begin{proposition}[Estimate for $E_{Q_1}(\Theta^m_{K}(P^q))$]\label{prop:finEQ1est}
Let $ \sum_{\ell \in \Zstar} \hat{V}(\ell)^2 < \infty $. For $\xi_\lambda = e^{-\lambda S} \Omega$, given $ \varepsilon > 0 $, there exist constants $ C, C_\varepsilon > 0 $ such that for all even $ m \in \NNN $, $ m \ge 2 $, $ \lambda \in [0,1] $, and $ q \in B_{\F}^c $,
\begin{align} \label{eq:finalEQ1est_Coulomb}
	\abs{\eva{\xi_\lambda, E_{Q_1}\!\left(\Theta^m_K(P^q)\right) \xi_\lambda}}
	&\leq C_\varepsilon C^m \left( \kF^{-\frac 32 + \varepsilon}
		+ \kF^{-1 + \varepsilon} \Xi^\half \right)
		e(q)^{-1}
		+ C^m \kF^{-1} \eva{\xi_\lambda, a_q^* a_q \xi_\lambda}^{\half} e(q)^{-\half} \;.
\end{align}
If $ \sum_{\ell \in \Zstar} \hat{V}(\ell) < \infty $, we have the even stronger bound
\begin{align} \label{eq:finalEQ1est}
	\abs{\eva{\xi_\lambda, E_{Q_1}\!\left(\Theta^m_K(P^q)\right) \xi_\lambda}}
	&\leq C_\varepsilon C^m \left(
		\kF^{-\frac{3}{2}} \Xi^\half
		+ \kF^{-1}\Xi^{1-\varepsilon} \right) e(q)^{-1} \nonumber\\
	&\quad + C_\varepsilon C^m \kF^{-1} \Xi^{\half - \varepsilon} \eva{\xi_\lambda, a_q^* a_q \xi_\lambda}^{\half} e(q)^{-\half} \;.
\end{align}
\end{proposition}

To prove this proposition, we need to estimate $ E^{m,1}_{Q_1}(q) $, $ E^{m,2}_{Q_1}(q) $, and $ E^{m,3}_{Q_1}(q) $, which is done in Lemma~\ref{lem:EQ111} and Lemma~\ref{lem:EQ112}. These two lemmas in turn rely on the following lemma.

\begin{lemma}[H\"older estimate using the bootstrap quantity] \label{lem:Xi_halfminusepsilon}
For any $ \varepsilon > 0 $ and $ a \in \N $, there exists some $ C_{a,\varepsilon} > 0 $ such that for all $ \lambda \in [0,1] $ and all $ q \in \Z^3 $ we have
\begin{equation} \label{eq:Xi_halfminusepsilon}
	\Vert a_q (\cN + 1)^a \xi_\lambda \Vert
	\le C_{a,\varepsilon} \Xi^{\half-\varepsilon} \;.
\end{equation}
\end{lemma}

\begin{proof}
We iteratively apply the following bound, which follows from $ [\cN, a_q^* a_q] = 0 $:
\begin{equation}
	\Vert a_q (\cN + 1)^a \xi_\lambda \Vert^2
	= \eva{\xi_\lambda, (\cN + 1)^{2a} a_q^* a_q \xi_\lambda}
	\le \Vert a_q (\cN + 1)^{2a} \xi_\lambda \Vert \Xi^{\frac 12} \;.
\end{equation}
After $ n $ iterations,
\begin{equation}
	\Vert a_q (\cN + 1)^a \xi_\lambda \Vert
	\le \Vert a_q (\cN + 1)^{2^n a} \xi_\lambda \Vert^{2^{-n}} \Xi^{\half (1-2^{-n})} \;.
\end{equation}
We conclude using $ \Vert a_q \Vert \le 1 $ and Lemma~\ref{lem:gronNest}, and choosing $ n $ large enough.
\end{proof}

\begin{lemma} \label{lem:EQ111}
Let $ \sum_{\ell \in \Zstar} \hat{V}(\ell)^2 < \infty $ and recall definition~\eqref{eq:expandedEQ1} of $ E_{Q_1}^{m,j}(q) $. For $\xi_\lambda = e^{-\lambda S} \Omega$, there exists some $ C > 0 $ such that for all $ \lambda \in [0,1] $, even $ m \in \NNN $, $ m \ge 2 $, and $ q \in B_{\F}^c $,
\begin{align} \label{eq:estEQ111_Coulomb}
	\abs{\eva{\xi_\lambda,\left(E^{m,1}_{Q_1}(q) + E^{m,2}_{Q_1}(q) + \mathrm{h.c.}\right) \xi_\lambda }} 
	&\leq C^m \kF^{-1} \Xi^{\half} e(q)^{-1}
		\norm{ (\NN+1)^2 \xi_\lambda}  \nonumber\\
	&\quad + C^m \kF^{-1} \eva{\xi_\lambda, a_q^* a_q \xi_\lambda}^{\half} e(q)^{-\half} \norm{(\NN+1)^2 \xi_\lambda} \;.
\end{align}
If $ \sum_{\ell \in \Zstar} \hat{V}(\ell) < \infty $, given $ \varepsilon > 0 $, there exists some $ C_\varepsilon > 0 $ such that
\begin{align} \label{eq:estEQ111}
	\abs{\eva{\xi_\lambda,\left(E^{m,1}_{Q_1}(q) + E^{m,2}_{Q_1}(q) + \mathrm{h.c.}\right) \xi_\lambda }} 
	&\leq C_\varepsilon C^m \left(
		\kF^{-\frac{3}{2}} \Xi^\half
		+ \kF^{-1}\Xi^{1-\varepsilon} \right) e(q)^{-1}
		\norm{ (\NN+1)^2 \xi_\lambda} \nonumber\\
	&\quad + C_\varepsilon C^m \kF^{-1} \Xi^{\half - \varepsilon} \eva{\xi_\lambda, a_q^* a_q \xi_\lambda}^{\half} e(q)^{-\half} \;.
\end{align}
\end{lemma}
\begin{proof}
We start with estimating $ E^{m,1}_{Q_1}(q) $. Splitting the anticommutator in $ E^{m,1}_{Q_1}(q) $ as
\begin{equation} \label{eq:q-q}
	\Theta^m_K(P^q)(\ell)
	= \sum_{j=0}^m {{m}\choose j} K(\ell)^{m-j} P^q(\ell) K(\ell)^{j} \;,
\end{equation}
with $ K(\ell)^0 = \mathds{1} $, we obtain
\begin{equation} \label{eq:EQ1111}
\begin{aligned}
	& \abs{\eva{\xi_\lambda,\left(E^{m,1}_{Q_1}(q) +  \mathrm{h.c.}\right) \xi_\lambda }}
	\le 4 \sum_{j=0}^m {{m}\choose j} \sum_{\ell,\ell_1  \in \Zstar}\!\! \mathds{1}_{L_\ell}(q) | \I_j(\ell, \ell_1)| \;,
	\end{aligned}
\end{equation}
where
\begin{equation}
\begin{aligned}
& 	\I_j(\ell, \ell_1)
	\coloneq \sum_{\substack{r\in L_{\ell} \cap L_{\ell_1}\\ s \in L_{\ell},s_1\in L_{\ell_1}}}
		\eva{\xi_\lambda, K^{m-j}(\ell)_{r,q} K^{j}(\ell)_{q,s} K(\ell_1)_{r,s_1} a^*_{r-\ell_1} b^*_{s}(\ell) b^*_{-s_1}(-\ell_1) a_{r-\ell} \xi_\lambda} \;. \\
\end{aligned}
\end{equation}
We need three different strategies for $ j = 0 $, for $ 1 \le j \le m-1 $, and for $ j = m $. The general strategy is to apply the Cauchy--Schwarz inequality, estimate the $ K $-matrices by Lemma~\ref{lem:normsk} and then either eliminate annihilation operators by $ \norm{a_p} \le 1 $ or bound them by $ \sum_{p \in \Z^3} \norm{a_p \Psi}^2 = \Vert \cN^{\half} \Psi \Vert^2 $. In the first case $ j = 0 $, we start by
\begin{align}
	&\sum_{\ell,\ell_1 \in \Zstar} \mathds{1}_{L_\ell}(q) |\I_0(\ell, \ell_1)| \nonumber\\
	&\le \sum_{\ell,\ell_1 \in \Zstar} \mathds{1}_{L_\ell}(q) \times \nonumber\\
	&\quad \times \sum\limits_{r \in L_\ell \cap L_{\ell_1}} \abs{\eva{ \sum\limits_{s_1 \in L_{\ell_1}} K(\ell_1)_{r,s_1} b_{-s_1}(-\ell_1) b_{q}(\ell) a_{r-\ell_1} (\NN+1)^{\frac 32} (\NN+1)^{-\frac 32} \xi_\lambda, K^{m}(\ell)_{r,q} a_{r-\ell} \xi_\lambda }}\nonumber\\
	&\leq \sum_{\ell,\ell_1 \in \Zstar} \mathds{1}_{L_\ell}(q) \Bigg( \sum_{r,s_1 \in L_{\ell_1}} |K(\ell_1)_{r,s_1}|^2
		\sum_{s_1' \in L_{\ell_1}} \Bigg\Vert b_{-s_1'}(-\ell_1) b_{q}(\ell) a_{r-\ell_1} (\NN+1)^{-\frac 32}\xi_\lambda \Bigg\Vert^2\Bigg)^\half \times\nonumber\\
	&\quad \times \Bigg( \sum\limits_{r \in L_\ell} |K^{m}(\ell)_{r,q}|^2 \norm{ a_{r-\ell} (\NN+5)^{\frac 32}\xi_\lambda }^2\Bigg)^\half \nonumber\\
	&\leq \sum_{\ell,\ell_1 \in \Zstar} \Bigg( \norm{K(\ell_1)}_{\max,2}^2 \sum\limits_{r, s_1' \in \Z^3} \norm{ a_{-s_1'} a_{-s_1' + \ell_1} a_q a_{q-\ell} a_{r-\ell_1} (\NN+1)^{-\frac 32}\xi_\lambda}^2\Bigg)^\half \times\nonumber\\
	& \quad \times (C \hat{V}(\ell))^m \kF^{-1} e(q)^{-1} \norm{ \NN^\half(\NN+5)^{\frac 32}\xi_\lambda } \nonumber\\
	&\leq \sum_{\ell \in \Zstar} \kF^{-\half} \Bigg( \sum_{\ell_1 \in \Zstar} \hat{V}(\ell_1)^2 \Bigg)^{\half}
		\Bigg( \sum_{r,\ell_1,s_1' \in \Z^3}
		\norm{ a_{r-\ell_1} a_{-s_1'+\ell_1} a_{-s_1'} a_q  (\NN+1)^{-\frac 32}\xi_\lambda}^2\Bigg)^\half \times\nonumber\\
	& \quad \times (C \hat{V}(\ell))^m \kF^{-1} e(q)^{-1} \norm{ (\NN+5)^2 \xi_\lambda } \nonumber\\
	&\leq C^m \kF^{-\frac 32} e(q)^{-1}
		\norm{ a_q \xi_\lambda} \norm{(\NN+5)^2 \xi_\lambda}
	\leq C^m \kF^{-\frac 32} e(q)^{-1} \Xi^\half
	 	\norm { (\NN+1)^2 \xi_\lambda } \;.\label{eq:estEQ1111}
\end{align}
Note that in the second last line, we were able to use $ \sum_\ell \hat{V}(\ell)^m < \infty $, since $ m \ge 2 $.
The order in which we estimate annihilation operators by $ \cN^\half $ matters: We first took care of the sum over $r$ and the operator $a_{r-\ell_1}$, then the sum over $\ell_1$ with the second annihilation operator, and finally the sum over $ s_1' $.
Moreover, we absorbed a constant $ C $ into $ C^m $, and in the last line, we used $ (\cN+5)^m \le C (\cN+1)^m $ and the definition~\eqref{eq:Xi} of the bootstrap quantity $ \Xi $.

In the second case $ 1 \le j \le m-1 $, we have an operator $ a_{r-\ell} $ instead of $ a_q $ and obtain the bootstrap quantity $ \Xi $:
\begin{align}
	&\sum_{\ell,\ell_1 \in \Zstar} \mathds{1}_{L_\ell}(q) |\I_j(\ell, \ell_1)| \nonumber\\
	&\leq \sum_{\ell,\ell_1 \in \Zstar} \Bigg( \sum\limits_{s \in L_\ell} \abs{K^j(\ell)_{q,s}}^2
		\sum\limits_{r, s_1 \in L_{\ell_1}} \abs{K(\ell_1)_{r,s_1}}^2
		\sum\limits_{s' \in L_\ell} \sum\limits_{s_1' \in L_{\ell_1}} \norm{a_{r-\ell_1} b_{s'}(\ell) b_{-s_1'}(-\ell_1) \xi_\lambda}^2 \Bigg)^\half \times \nonumber\\
	&\quad \times \mathds{1}_{L_\ell}(q) \Bigg( \sum\limits_{r\in L_{\ell}} \abs{K^{m-j}(\ell)_{r,q}}^2 \norm{a_{r-\ell} \xi_\lambda }^2 \Bigg)^\half\nonumber\\
	&\leq \kF^{-\frac 32} e(q)^{-1}
		\sum_{\ell \in \Zstar} (C \hat{V}(\ell))^m
		\Bigg( \sum_{\ell_1 \in \Zstar} \norm{K(\ell_1)}_{\max,2}^2 \Bigg)^{\half} 
		\Bigg( \sum\limits_{r,\ell_1,s_1',s' \in \Z^3} \norm{a_{r-\ell_1} a_{-s_1'+\ell_1} a_{-s_1'} a_{s'} \xi_\lambda}^2 \Bigg)^\half
		\Xi^\half \nonumber\\
	&\leq C^m \kF^{-\frac 32} e(q)^{-1}
		\norm{(\NN+1)^2 \xi_\lambda} \Xi^\half \;. \label{eq:estEQ1112}
\end{align}
For the third case $ j = m $, we proceed similarly
\begin{align}
	& \sum_{\ell,\ell_1 \in \Zstar} \mathds{1}_{L_\ell}(q) |\I_m(\ell, \ell_1)| \nonumber\\
	&\leq \sum_{\ell,\ell_1 \in \Zstar} \mathds{1}_{L_\ell \cap L_{\ell_1}}(q)
		\Bigg\Vert \sum\limits_{s\in L_{\ell}, s_1 \in L_{\ell_1}} K^m(\ell)_{q,s}K(\ell_1)_{q,s_1} b_{-s_1}(-\ell_1) b_{s}(\ell) a_{q-\ell_1}\xi_\lambda \Bigg\Vert
		\norm{ a_{q-\ell}\xi_\lambda }\nonumber\\
	&\leq \sum_{\ell \in \Zstar} \mathds{1}_{L_\ell}(q) \Bigg(\sum\limits_{s \in L_{\ell}} \abs{K^m(\ell)_{q,s}}^2\Bigg)^\half 
		\Bigg(\sum_{\ell_1 \in \Zstar} \mathds{1}_{L_{\ell_1}}(q) \sum\limits_{s_1 \in L_{\ell_1}} \abs{K(\ell_1)_{q,s_1}}^2\Bigg)^\half \times \nonumber\\
	&\quad \times \Bigg( \sum_{\ell_1, s_1, s \in \Z^3} \norm{ a_{-s_1+\ell_1} a_{-s_1} a_s \xi_\lambda}^2 \Bigg)^{\half} \Xi^\half \nonumber\\
	&\leq C^m \kF^{-1} e(q)^{-1} \norm{ (\NN+1)^{\frac 32} \xi_\lambda}\Xi^\half \;. \label{eq:estEQ1113_Coulomb}
\end{align}
As later, $ \Xi \sim \kF^{-1} $, the r.~h.~s. here is only of order $ \sim \kF^{-\frac 32} $, in contrast to~\eqref{eq:estEQ1111} and~\eqref{eq:estEQ1112}, where it is $ \sim \kF^{-2} $. Nevertheless, for $ \sum_{\ell_1} \hat{V}(\ell_1) < \infty $, we can achieve a stronger estimate of order $ \sim \kF^{-2 + \varepsilon} $, using Lemma~\ref{lem:Xi_halfminusepsilon}:
\begin{align}
	|\I_m(\ell, \ell_1)|
	&\leq \Bigg(\sum\limits_{s \in L_{\ell}} \abs{K^m(\ell)_{q,s}}^2\Bigg)^\half \Bigg(\sum\limits_{s_1 \in L_{\ell_1}} \abs{K(\ell_1)_{q,s_1}}^2\Bigg)^\half \norm{ a_{q-\ell_1} (\NN+1)\xi_\lambda} \norm{ a_{q-\ell}\xi_\lambda }\nonumber\\
	&\leq (C \hat{V}(\ell))^m \hat{V}(\ell_1) \kF^{-1} e(q)^{-1} \sup_{q' \in \Z^3}\norm{ a_{q'} (\NN+1) \xi_\lambda}\Xi^\half\nonumber\\
	&\leq C_\varepsilon (C \hat{V}(\ell))^m \hat{V}(\ell_1) \kF^{-1} e(q)^{-1} \Xi^{1-\varepsilon} \;. \label{eq:estEQ1113}
\end{align}
Summing up the estimates and using $\sum_{j=1}^{m-1} {{m}\choose j} \le C^m $ concludes the proof for $ E^{m,1}_{Q_1}(q) $.

The estimate for $ E^{m,2}_{Q_1}(q) $ (compare~\eqref{eq:expandedEQ1}) is analogous, except for $ j = m $: Splitting $ E^{m,2}_{Q_1}(q) $ as in~\eqref{eq:EQ1111}, and estimating the analogous $ \I_m $-term via~\eqref{eq:estEQ1113_Coulomb} would result in a factor of $ e(q)^{-\half} e(q-\ell+\ell_1)^{\half} $ instead of $ e(q)^{-1} $. However, $ a_{q-\ell} $ gets replaced by $ a_q $, so we can recover a factor $ \eva{\xi_\lambda, a_q^* a_q \xi_\lambda}^{\half} $, which we expect to eventually scale like $ \sim n^{\RPA}(q)^{\half} \sim \kF^{-\half} e(q)^{-\half} $, providing the missing factor of $ e(q)^{-\half} $:
\begin{align}
		&\sum_{\ell,\ell_1 \in \Zstar} \mathds{1}_{L_\ell}(q) |\I_m(\ell, \ell_1)| \nonumber\\
		&\leq \sum_{\ell,\ell_1 \in \Zstar} \mathds{1}_{L_\ell \cap (L_{\ell_1}+\ell-\ell_1) }(q)
		\Bigg(\sum\limits_{s \in L_{\ell}} \abs{K^m(\ell)_{q,s}}^2\Bigg)^\half
		\Bigg(\sum\limits_{s_1 \in L_{\ell_1}} \abs{K(\ell_1)_{q-\ell+\ell_1,s_1}}^2\Bigg)^\half \times \nonumber\\
		&\quad \times \Bigg( \sum_{s,s_1 \in \Z^3} \norm{ a_{-s_1} a_{-s_1+\ell_1} a_s a_{s-\ell} a_{q-\ell+\ell_1} \xi_\lambda}^2
		\norm{ a_q \xi_\lambda }^2 \Bigg)^{\half} \nonumber\\
		&\leq \kF^{-1} e(q)^{-\frac 12}
		\Bigg(\sum_{\ell \in \Zstar} (C \hat{V}(\ell))^{2m} \Bigg)^\half \!
		\Bigg(\sum_{\ell_1 \in \Zstar} \hat{V}(\ell_1)^2 \Bigg)^\half \!
		\Bigg( \sum_{\ell, s, \ell_1, s_1 \in \Z^3}  \!\!\norm{ a_{s-\ell} a_s a_{-s_1+\ell_1} a_{-s_1} \xi_\lambda}^2\Bigg)^{\half} \!
		\norm{ a_q \xi_\lambda }  \nonumber\\
		&\leq C^m \kF^{-1} e(q)^{-\frac 12}
			\norm{ (\NN+1)^2 \xi_\lambda }
			\eva{\xi_{\lambda},a_q^* a_q\xi_{\lambda}}^{\half} \;. \label{eq:estEQ1113_Coulomb_bis}
\end{align}
In case $ \sum_\ell \hat{V}(\ell) < \infty $, with Lemma~\ref{lem:Xi_halfminusepsilon}, we may again obtain a stronger estimate, while still extracting $ \eva{\xi_\lambda, a_q^* a_q \xi_\lambda}^{\half} $:
\begin{align}
		|\I_m(\ell, \ell_1)|
		&\leq \Bigg(\sum\limits_{s \in L_{\ell}} \abs{K^m(\ell)_{q,s}}^2\Bigg)^\half
		\Bigg(\sum\limits_{s_1 \in L_{\ell_1}} \abs{K(\ell_1)_{q-\ell+\ell_1,s_1}}^2\Bigg)^\half
		\norm{ a_{q-\ell+\ell_1} (\NN+1) \xi_\lambda}
		\norm{ a_q \xi_\lambda }\nonumber\\
		&\leq (C \hat{V}(\ell))^m \hat{V}(\ell_1) \kF^{-1} e(q)^{-\half}
		\sup_{q' \in \Z^3} \norm{ a_{q'} (\NN+1) \xi_\lambda }
		\eva{\xi_{\lambda},a_q^* a_q\xi_{\lambda}}^{\half} \nonumber\\
		&\leq C_\varepsilon (C \hat{V}(\ell))^m
		\hat{V}(\ell_1)
		\kF^{-1} e(q)^{-\half} \Xi^{\half-\varepsilon} \eva{\xi_{\lambda},a_q^* a_q\xi_{\lambda}}^{\half} \;.  		\label{eq:estEQ1113_bis}
\end{align}
This concludes the proof.
\end{proof}

\begin{lemma} \label{lem:EQ112}
Let $ \sum_{\ell \in \Zstar} \hat{V}(\ell)^2 < \infty $. For $\xi_\lambda = e^{-\lambda S} \Omega$, given $ \varepsilon > 0 $, there exist $ C, C_\varepsilon > 0 $ such that for all $ \lambda \in [0,1] $, even $ m \in \NNN $, $ m \ge 2 $, and $ q \in B_{\F}^c $,
\begin{align} \label{eq:estEQ112_Coulomb}
	\abs{\eva{\xi_\lambda,\left(E^{m,3}_{Q_1}(q) + \mathrm{h.c.}\right) \xi_\lambda }}
	\leq C_\varepsilon C^m \left( \kF^{-\frac 32 + \varepsilon}
		+ \kF^{-1 + \varepsilon} \Xi^\half \right)
		e(q)^{-1}
		\norm{(\NN+1) \xi_\lambda } \;.
\end{align}
If $ \sum_{\ell \in \Zstar} \hat{V}(\ell) < \infty $, then
\begin{align} \label{eq:estEQ112}
	\abs{\eva{\xi_\lambda,\left(E^{m,3}_{Q_1}(q) + \mathrm{h.c.}\right) \xi_\lambda }}
	\leq C^m \kF^{-\frac{3}{2}} \Xi^{\half} e(q)^{-1}
		\norm{(\NN+1)^\half \xi_\lambda } \;.
\end{align}
\end{lemma}

\begin{proof}
As in the proof of Lemma~\ref{lem:EQ111}, we split
\begin{equation} \label{eq:EQ1121}
\begin{aligned}
	& \abs{\eva{\xi_\lambda,\left(E^{m,3}_{Q_1}(q) +  \mathrm{h.c.}\right) \xi_\lambda }}
	\le 4 \sum_{j=0}^m {{m}\choose j} \sum_{\ell,\ell_1 \in \Zstar}\!\! \mathds{1}_{L_\ell}(q) |\I_j(\ell, \ell_1)| \;,
	\end{aligned}
\end{equation}
where
\begin{equation}
\begin{aligned}
	& \I_j(\ell, \ell_1)
	\coloneq \sum_{\substack{r\in L_{\ell} \cap L_{\ell_1} \cap (-L_{\ell_1}+\ell+\ell_1)\\ s \in L_{\ell}}}
		\eva{\xi_\lambda, K^{m-j}(\ell)_{r,q} K^{j}(\ell)_{q,s}K(\ell_1)_{r,-r+\ell+\ell_1} a_{r-\ell-\ell_1} a_{r-\ell_1} b_{s}(\ell) \xi_\lambda} \;. \\
\end{aligned}
\end{equation}
Again we need three slightly different strategies for $ j = 0 $, for $ 1 \le j \le m-1 $, and for $ j = m $. For the first case $ j = 0 $, we insert $1 = (\NN+1)^{-\half}(\NN+1)^{\half}$, followed by the Cauchy--Schwarz inequality. Then, we estimate the $ K $-matrices by Lemma~\ref{lem:normsk} and use $ \norm{a_p} \le 1 $ and $ \sum_{p \in \Z^3} \norm{a_p \Psi}^2 = \Vert \cN^\half \Psi \Vert^2 $ so that
\begin{align}
	&\sum_{\ell,\ell_1 \in \Zstar} \mathds{1}_{L_\ell}(q) |\I_0(\ell, \ell_1)| \nonumber\\
	&\leq \sum_{\ell,\ell_1 \in \Zstar} \mathds{1}_{L_\ell}(q) \sum\limits_{\substack{r\in L_{\ell} \cap L_{\ell_1} \cap (-L_{\ell_1}+\ell+\ell_1)}}
		\norm{ (\NN+5)^{\half} \xi_\lambda} \times \nonumber\\
	&\quad \times \norm{ K^m(\ell )_{r,q} K(\ell_1)_{r,-r+\ell+\ell_1} a_{r-\ell-\ell_1} a_{r-\ell_1} b_{q}(\ell) (\NN+1)^{-\half} \xi_\lambda }\nonumber\\
	 &\leq \sum_{\ell \in \Zstar} (C \hat{V}(\ell))^m \kF^{-1} e(q)^{-1}
	 	\norm{ (\NN+5)^{\half} \xi_\lambda}
	 	\sum_{r \in L_\ell} \Bigg( \sum_{\ell_1 \in \Zstar} \mathds{1}_{L_{\ell_1} \cap (-L_{\ell_1}+\ell+\ell_1)}(r) |K(\ell_1)_{r,-r+\ell+\ell_1}|^2 \Bigg)^{\half} \times \nonumber\\
	 &\quad \times \Bigg( \sum\limits_{\ell_1 \in \Z^3} \norm{ a_{r-\ell_1} a_q (\NN+1)^{-\half} \xi_\lambda }^2 \Bigg)^\half \nonumber\\
	 &\leq \sum_{\ell \in \Zstar} (C \hat{V}(\ell))^m \kF^{-1} e(q)^{-1}
	 	\norm{ (\NN+5)^{\half} \xi_\lambda}
	 	\sum_{r \in L_\ell} \Bigg( \sum_{\ell_1 \in \Zstar} \frac{\mathds{1}_{L_{\ell_1} \cap (L_{-\ell_1}+\ell+\ell_1)}(r) \hat{V}(\ell_1)^2}{(\lambda_{\ell_1,r} + \lambda_{\ell_1,-r+\ell+\ell_1})^2} \kF^{-2} \Bigg)^{\half} \norm{ a_q \xi_\lambda } \nonumber\\
	 &\leq \sum_{\ell \in \Zstar} (C \hat{V}(\ell))^m \kF^{-1} e(q)^{-1}
	 	\norm{ (\NN+5)^{\half} \xi_\lambda}
	 	\sum_{r \in L_\ell} e(r)^{-1} \kF^{-1} \Bigg( \sum_{\ell_1 \in \Zstar} \hat{V}(\ell_1)^2 \Bigg)^{\half} \Xi^\half \;,
\end{align}
where we used $ \lambda_{\ell_1,r} \ge \frac{1}{2} e(r) $ and $ \norm{ a_q \xi_\lambda} \le \Xi^\half $, compare~\eqref{eq:Xi}. Then, with~\eqref{eq:lambdainverse}
\begin{equation}\label{eq:estEQ1121_Coulomb}
	\sum_{\ell,\ell_1 \in \Zstar} \mathds{1}_{L_\ell}(q) |\I_0(\ell, \ell_1)|
	\leq C_\varepsilon C^m \kF^{-1+\varepsilon} e(q)^{-1}
	 	\norm{ (\NN+5)^{\half} \xi_\lambda}
	 	\Xi^{\half}	\;.
\end{equation}
As we will show that $ \Xi \sim \kF^{-1} $, this bound will later be of order $ \sim \kF^{-\frac 32 + \varepsilon} $. For $ \sum_{\ell_1} \hat{V}(\ell_1) < \infty $, we can even achieve an estimate of order $ \sim \kF^{-2} $ via
\begin{align}
	&|\I_0(\ell, \ell_1)| \nonumber\\
	&\leq \sum\limits_{r\in L_{\ell} \cap L_{\ell_1}\cap (-L_{\ell_1}+\ell+\ell_1)}
		\norm{ (\NN+5)^{\half} \xi_\lambda} 
		\norm{ K^m(\ell )_{r,q} K(\ell_1)_{r,-r+\ell+\ell_1} a_{r-\ell-\ell_1} a_{r-\ell_1} b_{q}(\ell) (\NN+1)^{-\half} \xi_\lambda }\nonumber\\
	 &\leq (C \hat{V}(\ell))^m \hat{V}(\ell_1) \kF^{-1} e(q)^{-1}
	 	\norm{ (\NN+5)^{\half} \xi_\lambda} \norm{K(\ell_1) }_{\max,2}
	 	\Bigg( \sum\limits_{r\in \Z^3} \norm{ a_{r-\ell_1} a_q (\NN+1)^{-\half} \xi_\lambda }^2 \Bigg)^\half \nonumber\\
	 &\leq (C \hat{V}(\ell))^m \hat{V}(\ell_1)
	 	\kF^{-\frac 32} e(q)^{-1}
	 	\norm{(\NN+5)^{\half} \xi_\lambda}
	 	\Xi^{\half} \;.
\label{eq:estEQ1121}
\end{align}
The estimate for the second case $ 1 \le j \le m-1 $ follows a similar strategy, using $ \Vert \xi_\lambda \Vert = 1 $, $ \lambda_{\ell_1,r} \ge \frac{1}{2} e(r) $ and~\eqref{eq:lambdainverse}:
\begin{align}
	&\sum_{\ell,\ell_1 \in \Zstar} \mathds{1}_{L_\ell}(q) |\I_j(\ell, \ell_1)| \nonumber\\
	&\leq \norm{ \xi_\lambda} \sum_{\ell,\ell_1 \in \Zstar} \mathds{1}_{L_\ell}(q)
		\sum\limits_{\substack{r\in L_{\ell} \cap L_{\ell_1}  \cap (-L_{\ell_1} + \ell + \ell_1)\\s\in L_{\ell}}}
		\norm{K^{m-j}(\ell)_{r,q} K^j(\ell)_{q,s} K(\ell_1)_{r,-r+\ell+\ell_1} a_{r-\ell-\ell_1} a_{r-\ell_1} b_{s}(\ell) \xi_\lambda }\nonumber\\
	&\leq \sum_{\ell \in \Zstar}
		(C \hat{V}(\ell))^{m-j} \kF^{-1} e(q)^{-1}
		\sum\limits_{r \in L_\ell} \Bigg(\sum_{\ell_1 \in \Zstar} \mathds{1}_{L_{\ell_1} \cap (-L_{\ell_1}+\ell+\ell_1)}(r) \abs{ K(\ell_1)_{r,-r+\ell+\ell_1} }^2\Bigg)^\half \times\nonumber\\ 
	&\quad \times \mathds{1}_{L_\ell}(q) \sum\limits_{s\in L_{\ell}} \Bigg( \sum_{\ell_1 \in \Z^3} \norm{K^{j}(\ell)_{q,s} a_{r-\ell_1} b_{s}(\ell) \xi_\lambda }^2 \Bigg)^\half \nonumber\\
	&\leq \sum_{\ell \in \Zstar}
		(C \hat{V}(\ell))^{m-j} \kF^{-2} e(q)^{-1}
		\sum\limits_{r \in L_\ell} e(r)^{-1} \Bigg(\sum_{\ell_1 \in \Zstar} \hat{V}(\ell_1)^2 \Bigg)^\half
		\mathds{1}_{L_\ell}(q) \sum\limits_{s\in L_{\ell}} \norm{K^{j}(\ell)_{q,s} a_s (\NN+1)^{\half} \xi_\lambda } \nonumber\\
	&\leq C_\varepsilon \sum_{\ell \in \Zstar}
		(C \hat{V}(\ell))^{m-j} \kF^{-1+\varepsilon} e(q)^{-1}
		\norm{K^j(\ell)}_{\max,2}
		\Bigg( \sum\limits_{s\in L_{\ell}} \norm{a_s (\NN+1)^{\half} \xi_\lambda }^2 \Bigg)^{\half} \nonumber\\
	&\leq C_\varepsilon C^m \kF^{-\frac 32 + \varepsilon} e(q)^{-1} \norm{(\NN+1) \xi_\lambda } \;. \label{eq:estEQ1122_Coulomb}
\end{align}
Again, for $ \sum_{\ell \in \Zstar} \hat{V}(\ell) $, we get a simpler and stronger estimate:
\begin{align}
	&|\I_j(\ell, \ell_1)| \nonumber\\
	&\leq \norm{ (\NN+5)^{\half} \xi_\lambda}
		\sum_{\substack{r\in L_{\ell} \cap L_{\ell_1} \cap (-L_{\ell_1} + \ell + \ell_1)\\s\in L_{\ell}}}
		\norm{K^{m-j}(\ell)_{r,q} K^j(\ell)_{q,s} K(\ell_1)_{r,-r+\ell+\ell_1} a_{r-\ell_1} a_s (\NN+1)^{-\half} \xi_\lambda }\nonumber\\
	&\leq \norm{ (\NN+5)^{\half} \xi_\lambda} 
		(C \hat{V}(\ell))^{m-j} \kF^{-1} e(q)^{-1}
		\Bigg(\sum_{r\in L_{\ell_1} \cap (-L_{\ell_1} + \ell + \ell_1)}\abs{ K(\ell_1)_{r,-r+\ell+\ell_1} }^2\Bigg)^\half \times \nonumber\\
	&\quad \times \sum_{s\in L_{\ell}} \Bigg( \sum_{r \in \Z^3}\norm{K^{j}(\ell)_{q,s} a_{r-\ell_1} a_s (\NN+1)^{-\half} \xi_\lambda }^2 \Bigg)^\half \nonumber\\
	&\leq \norm{ (\NN+5)^{\half} \xi_\lambda}
		(C \hat{V}(\ell))^{m-j} \kF^{-1} e(q)^{-1}
		\norm{K(\ell_1)}_{\max,2}
		\sum_{s\in L_{\ell}}\abs{K^{j}(\ell)_{q,s}}
		\norm{a_s \xi_\lambda }		
	\nonumber\\
	&\leq \norm{(\NN+5)^\half \xi_\lambda }
		(C \hat{V}(\ell))^{m-j} \kF^{-1} e(q)^{-1}
		\norm{ K(\ell_1) }_{\max,2}
		\norm{ K^{j}(\ell)}_{\mathrm{max,1}} \Xi^\half \nonumber \\
	&\leq \norm{(\NN+5)^\half \xi_\lambda }
		(C \hat{V}(\ell))^m
		\hat{V}(\ell_1)
		\kF^{-\frac 32} e(q)^{-1} \Xi^\half \;.
\end{align}
Finally, for the case $ j = m $,
\begin{align}
	&\sum_{\ell,\ell_1 \in \Zstar} \mathds{1}_{L_\ell}(q) |\I_m(\ell, \ell_1)| \nonumber\\
	&\leq \sum_{\ell,\ell_1 \in \Zstar} \mathds{1}_{L_\ell \cap L_{\ell_1} \cap (-L_{\ell_1} + \ell + \ell_1)}(q) \norm{\xi_\lambda}
		\sum\limits_{s \in L_{\ell}}
		\norm{ K^m(\ell)_{q,s} K(\ell_1)_{q,-q+\ell+\ell_1} a_{q-\ell-\ell_1} a_{q-\ell_1} b_{s}(\ell) \xi_\lambda } \nonumber\\
	&\leq C
		\sum_{\ell,\ell_1 \in \Zstar} \mathds{1}_{L_\ell}(q)
		\norm{K^m(\ell)}_{\max,2} \hat{V}(\ell_1) \kF^{-1} e(q)^{-1}
		\Bigg(\sum\limits_{s \in \Z^3} \norm{ a_{q-\ell_1} a_s \xi_\lambda }^2\Bigg)^\half \nonumber\\
	&\leq C^m
		\Bigg( \sum_{\ell_1 \in \Zstar} \hat{V}(\ell_1)^2 \Bigg)^{\half}
		\kF^{-\frac 32} e(q)^{-1}
		\Bigg(\sum\limits_{\ell_1,s \in \Z^3} \norm{ a_{q-\ell_1} a_s \xi_\lambda }^2\Bigg)^\half \nonumber\\
	&\leq C^m \kF^{-\frac 32} e(q)^{-1} \norm{(\NN+1) \xi_\lambda } \;, \label{eq:estEQ1123_Coulomb}
\end{align}
and for $ \sum_{\ell} \hat{V}(\ell) < \infty $, we again get a stronger bound:
\begin{align}
	&|\I_m(\ell, \ell_1)| \nonumber\\
	&\leq \mathds{1}_{L_{\ell_1} \cap (-L_{\ell_1} + \ell + \ell_1)}(q) \norm{(\NN+5)^{\half} \xi_\lambda}
		\sum_{s \in L_{\ell}}
		\norm{ K^m(\ell)_{q,s} K(\ell_1)_{q,-q+\ell+\ell_1} a_{q-\ell_1} a_s (\NN+1)^{-\half} \xi_\lambda } \nonumber\\
	&\leq \norm{(\NN+5)^{\half} \xi_\lambda}
		\norm{K^m(\ell)}_{\max,2}
		C \hat{V}(\ell_1) \kF^{-1} e(q)^{-1}
		\Bigg(\sum_{s \in \Z^3}  \norm{ a_{q-\ell_1} a_s (\NN+1)^{-\half} \xi_\lambda }^2\Bigg)^\half \nonumber\\
	&\leq \norm{(\NN+5)^\half \xi_\lambda }
		(C \hat{V}(\ell))^m
		\hat{V}(\ell_1)
		\kF^{-\frac 32} e(q)^{-1} \Xi^\half \;. \label{eq:estEQ1123}
\end{align}
Adding up all bounds yields the result.
\end{proof}

\begin{proof}[Proof of Proposition~\ref{prop:finEQ1est}]
We sum the estimates from Lemmas~\ref{lem:EQ111} and~\ref{lem:EQ112} and use that $ \norm{(\NN+1)^2 \xi_\lambda} \le C $ according to Lemma \ref{lem:gronNest}.
\end{proof}

\subsubsection{Error Estimates for Odd $m$}

\begin{proposition}[Estimate for $E_{Q_2}(\Theta^m_{K}(P^q))$]\label{prop:finEQ2est}
Let $ \sum_{\ell \in \Zstar} \hat{V}(\ell)^2 |\ell|^\alpha < \infty $ for some $ \alpha > 0 $ and recall $ E_{Q_2}(\Theta^m_K(P^q)) $ from~\eqref{eq:EQ1EQ2extension}. For $\xi_\lambda = e^{-\lambda S} \Omega$, given $ \varepsilon > 0 $ there exist constants $ C, C_\varepsilon > 0 $ such that for all odd $ m \in \NNN $, $ \lambda \in [0,1] $, and $ q \in B_{\F}^c $, the following bound is true for any $ \gamma \ge 0 $:
\begin{align}
	&\abs{\eva{\xi_\lambda,  E_{Q_2}(\Theta^m_K(P^q)) \xi_\lambda} - \delta_{m,1} n^{\ex,1}(q)} \label{eq:finEQ2est_Coulomb} \\
	&\leq C_\varepsilon C^m \left( \kF^{-\frac 32 + \varepsilon}
		+ \kF^{-1 - \frac{\alpha \gamma}{2}}
		+ \kF^{-1+ \frac{3-\alpha}{2} \gamma} \Xi^\half
		+ \kF^{-1+\varepsilon} \Xi^\half \right) e(q)^{-1} + C^m \kF^{-1} \eva{\xi_\lambda, a_q^* a_q \xi_\lambda}^{\half} e(q)^{-\half}  \,. \nonumber
\end{align}
If $ \sum_{\ell \in \Zstar} \hat{V}(\ell) < \infty $, we have the stronger bound
\begin{equation}\label{eq:finEQ2est}
\begin{split}
	&\abs{\eva{\xi_\lambda, E_{Q_2}\!\left(\Theta^m_K(P^q)\right) \xi_\lambda}} \\
	&\leq C_\varepsilon C^m \left( \kF^{-2}
		+ \kF^{-\frac{3}{2}} \Xi^\half
		+ \kF^{-1} \Xi^{1-\varepsilon} \right) e(q)^{-1}
		+ C_\varepsilon C^m \kF^{-1} \Xi^{\half - \varepsilon} \eva{\xi_\lambda, a_q^* a_q \xi_\lambda}^{\half} e(q)^{-\half} \;.
\end{split}
\end{equation}
\end{proposition}

Note that here, $ m = 1 $ and thus $ \sum_\ell \hat{V}(\ell)^m = \infty $ can occur for $ \alpha \to 0 $. In the end, we will choose $ \gamma = \frac 13 $ as explained below. Then, the error in~\eqref{eq:finEQ2est_Coulomb} is $ \sim \kF^{-1-\frac {\alpha}{6}} $.
To prove this proposition, we estimate the terms $ E^{m,1}_{Q_2}(q) $ to $ E^{m,11}_{Q_2} (q)$ and $ n^{\ex,m}(q) $.

\begin{lemma} \label{lem:EQ211}
Let $ \sum_{\ell \in \Zstar} \hat{V}(\ell)^2 |\ell|^\alpha < \infty $ for some $ \alpha > 0 $ and recall definition~\eqref{eq:expandedEQ2} of $ E_{Q_2}^{m,j}(q) $. For $\xi_\lambda = e^{-\lambda S} \Omega$, there exists $ C > 0 $ such that for all $ \lambda \in [0,1] $, odd $ m \in \NNN $, and $ q \in B_{\F}^c $, the following bound is true for any $ \gamma \ge 0 $:
\begin{align}
	& \abs{\eva{\xi_\lambda,\left(E^{m,1}_{Q_2}(q) + E^{m,2}_{Q_2}(q) + \mathrm{h.c.}\right) \xi_\lambda }} \label{eq:estEQ211_Coulomb}\\
	&\leq C^m \left( \kF^{-\frac 32}
		+ \kF^{-1-\frac{\alpha \gamma}{2}}
		+ \kF^{-1+ \frac{3-\alpha}{2} \gamma} \Xi^{\half} \right) e(q)^{-1}
		\norm{ (\NN+1)^{\frac 52} \xi_\lambda }^2 \nonumber\\
	&\quad + C^m \kF^{-1} \eva{\xi_\lambda, a_q^* a_q \xi_\lambda}^{\half} e(q)^{-\half} \norm{ (\NN+1)^2 \xi_\lambda } \;. \nonumber
\end{align}
If $ \sum_{\ell \in \Zstar} \hat{V}(\ell) < \infty $, given $ \varepsilon > 0 $, there exists some $ C_\varepsilon > 0 $ such that
\begin{align}
	& \abs{\eva{\xi_\lambda,\left(E^{m,1}_{Q_2}(q) + E^{m,2}_{Q_2}(q) + \mathrm{h.c.}\right) \xi_\lambda }} \label{eq:estEQ211} \\
	&\leq C_\varepsilon C^m \left( \kF^{-2}
		+ \kF^{-\frac{3}{2}} \Xi^\half
		+ \kF^{-1}\Xi^{1-\varepsilon} \right) e(q)^{-1}
		\norm { (\NN+1)^{\frac 52} \xi_\lambda }^2 \nonumber\\
	&\quad + C_\varepsilon C^m \kF^{-1} \Xi^{\half - \varepsilon} \eva{\xi_\lambda, a_q^* a_q \xi_\lambda}^{\half} e(q)^{-\half} \;. \nonumber
\end{align}
\end{lemma}
\begin{proof}
The proof is similar to the one of Lemma~\ref{lem:EQ111}: We start with $ E^{m,1}_{Q_2}(q) $, where we split the anticommutator using~\eqref{eq:q-q} to get
\begin{equation} \label{eq:EQ2111}
\begin{aligned}
	& \abs{\eva{\xi_\lambda,\left(E^{m,1}_{Q_2}(q) + \mathrm{h.c.}\right) \xi_\lambda }}
	\le 4 \sum_{j=0}^m {{m}\choose j} \sum_{\ell,\ell_1 \in \Zstar}\!\! \mathds{1}_{L_\ell}(q) |\I_j(\ell, \ell_1)| \;,
	\end{aligned}
\end{equation}
where
\begin{equation}
\begin{aligned}
	& \I_j(\ell, \ell_1)
	\coloneq \sum_{\substack{r\in L_{\ell} \cap L_{\ell_1}\\ s \in L_{\ell},s_1\in L_{\ell_1}}}
		\eva{\xi_\lambda, K^{m-j}(\ell)_{r,q} K^{j}(\ell)_{q,s} K(\ell_1)_{r,s_1} a^*_{r-\ell_1} b^*_{-s_1}(-\ell_1) b_{-s}(-\ell) a_{r-\ell} \xi_\lambda} \;. \\
\end{aligned}
\end{equation}
For the case $ j = 0 $ we use $1 = (\NN+1) (\NN+1)^{-1}$ and Lemma~\ref{lem:normsk}. Since $ m = 1 $ may occur, we have to use the Cauchy--Schwarz inequality for the sum over $\ell$ to ensure that we get at least the second power $\sum_{\ell} \hat{V}(\ell)^2$ and not just $\sum_{\ell} \hat{V}(\ell)$. Then
\begin{align}
	\sum_{\ell,\ell_1 \in \Zstar} \mathds{1}_{L_\ell}(q) |\I_0(\ell, \ell_1)|
 	&\leq \sum\limits_{\ell,\ell_1 \in \Zstar} \mathds{1}_{L_\ell}(q) \Bigg( \sum\limits_{r \in L_{\ell_1}} 
 		\Bigg\Vert \sum\limits_{s_1 \in L_{\ell_1}} K(\ell_1)_{r,s_1} b_{-s_1}(-\ell_1) a_{r-\ell_1} (\NN+1) \xi_\lambda \Bigg\Vert^2\Bigg)^\half \times\nonumber\\
 	&\quad \times \Bigg( \sum\limits_{r \in L_\ell} \norm{K^{m}(\ell)_{r,q} b_{-q}(-\ell) a_{r-\ell} (\NN+1)^{-1}\xi_\lambda }^2\Bigg)^\half \nonumber\\
 	&\leq \Bigg( \sum\limits_{\ell_1 \in \Zstar} \norm{K(\ell_1)}_{\mathrm{max,2}}^2 \Bigg)^{\half} \Bigg(
 		\sum\limits_{r, \ell_1, s_1 \in \Z^3} \norm{a_{r-\ell_1} a_{-s_1+\ell_1} a_{-s_1} (\NN+1) \xi_\lambda}^2\Bigg)^\half \times\nonumber\\
 	&\quad \times \Bigg( \sum\limits_{\ell \in \Zstar}(C \hat{V}(\ell))^{2m} \Bigg)^{\half}
 		\kF^{-1} e(q)^{-1}
 		\Bigg( \sum\limits_{\ell \in \Z^3} \norm{a_{-q + \ell} a_q (\NN+1)^{-\half} \xi_\lambda } \Bigg)^{\half} \nonumber\\
 	&\leq C^m \norm{ (\NN+1)^{\frac 52}\xi_\lambda}
 		\kF^{-\frac 32} e(q)^{-1} \Xi^{\half} \;. \label{eq:estEQ2111}
\end{align}
The case $ 1 \le j \le m-1 $ only occurs for $ m \ge 2 $, so $ \sum_\ell \hat{V}(\ell)^m < \infty $ is true for any $ \alpha > 0 $:
\begin{align}
	\sum_{\ell,\ell_1 \in \Zstar} \mathds{1}_{L_\ell}(q) |\I_j(\ell, \ell_1)|
	&\leq \sum_{\ell,\ell_1 \in \Zstar} \mathds{1}_{L_\ell}(q) \sum\limits_{r\in L_{\ell} \cap L_{\ell_1}}\! \Bigg( \sum\limits_{s \in L_\ell} \abs{K^{j}(\ell)_{q,s}}^2\Bigg)^\half \bigg( \sum\limits_{s \in L_\ell}\norm{b_{-s}(-\ell) a_{r-\ell} \xi_\lambda}^2\bigg)^\half \times\nonumber\\
		&\quad \times \Bigg( \sum\limits_{s_1 \in L_{\ell_1}}\abs{K(\ell_1)_{r,s_1}}^2\Bigg)^\half \bigg(\sum\limits_{s_1 \in L_{\ell_1}}\norm{ K^{m-j}(\ell)_{r,q} b_{-s_1}(-\ell_1)  a_{r-\ell_1} \xi_\lambda }^2\bigg)^\half
	\nonumber\\
	&\leq \sum_{\ell \in \Zstar} (C \hat{V}(\ell))^m
		\kF^{-2} e(q)^{-1}
		\bigg( \sum\limits_{r, s\in \Z^3} \norm{ a_{r-\ell} a_{-s} \xi_\lambda}^2\bigg)^\half \times \nonumber\\
		&\quad \times 
		\Bigg( \sum_{\ell_1 \in \Zstar} \hat{V}(\ell_1)^2 \Bigg)^{\half}
	\bigg(
		\sum\limits_{r, \ell_1, s_1 \in \Z^3} \norm{ a_{r-\ell_1} a_{-s_1 + \ell_1} a_{-s_1} \xi_\lambda }^2\bigg)^\half
	\nonumber\\
	&\leq C^m \kF^{-2} e(q)^{-1}
	\norm{ (\NN+1)^{\frac 32} \xi_\lambda}^2 \;. \label{eq:estEQ2112_Coulomb}
\end{align}

Finally, let us consider the case $ j = m $. For Coulomb potentials, this is the most difficult term to bound. In analogy to~\cite{CHN24}, we introduce the ball $ S := \Zstar \cap B_{\kF^{\gamma}}(0) $ with $ \gamma \ge 0 $. For $ \ell_1 $ outside the ball, we proceed as follows:
\begin{align}
	&\sum_{\ell \in \Zstar} \sum_{\ell_1 \in \Zstar \setminus S} \mathds{1}_{L_\ell}(q) |\I_m(\ell, \ell_1)| \nonumber\\
	&\leq \sum_{\ell \in \Zstar} \sum_{\ell_1 \in \Zstar \setminus S} \mathds{1}_{L_\ell \cap L_{\ell_1}}(q) \Bigg(\sum\limits_{s \in L_{\ell}} \abs{K^m(\ell)_{q,s}}^2\Bigg)^\half
		\Bigg(\sum\limits_{s_1 \in L_{\ell_1}} \abs{K(\ell_1)_{q,s_1}}^2\Bigg)^\half \times \nonumber \\
	& \quad \times \norm{ a_{q-\ell} (\NN+1)^{\half} \xi_\lambda}
		\norm{ a_{q-\ell_1} (\NN+1)^{\half} \xi_\lambda }\nonumber\\
	&\leq \kF^{-1} e(q)^{-1}
		\Bigg( \sum_{\ell \in \Zstar} (C \hat{V}(\ell))^{2m} \Bigg)^{\half}
		\Bigg( \sum_{\ell_1 \in \Zstar \setminus S} \hat{V}(\ell_1)^2 \Bigg)^{\half} \times \nonumber \\
	& \quad \times 
		\Bigg( \sum_{\ell \in \Zstar} \norm{ a_{q-\ell} (\NN+1)^{\half} \xi_\lambda}^2 \Bigg)^{\half}
		\Bigg( \sum_{\ell_1 \in \Zstar} \norm{ a_{q-\ell_1} (\NN+1)^{\half} \xi_\lambda }^2 \Bigg)^{\half} \nonumber\\
	&\leq \kF^{-1} e(q)^{-1}
		\Bigg( \sum_{\ell \in \Zstar} (C \hat{V}(\ell))^{2m} \Bigg)^{\half}
		\Bigg( \sum_{\ell_1 \in \Zstar \setminus S} \hat{V}(\ell_1)^2 \Bigg)^{\half} \norm{(\NN+1) \xi_\lambda}^2 \nonumber\\
	&\leq C^m \kF^{-1-\frac{\alpha \gamma}{2}} e(q)^{-1}
		\norm{(\NN+1) \xi_\lambda}^2 \;. \label{eq:estEQ2113_Coulomb_1}
\end{align}
Here we used that $ \sum_{|\ell_1| \ge \kF^\gamma} \hat{V}(\ell_1)^2 \le \kF^{-\alpha \gamma} \sum_{\ell_1 \in \Zstar} \hat{V}(\ell_1)^2 |\ell_1|^\alpha \le C \kF^{-\alpha \gamma} $.
For $ \ell_1 \in S $, we can achieve a better bound by extracting an additional $ \Xi^{\half} $:
\begin{align}
	&\sum_{\ell \in \Zstar} \sum_{\ell_1 \in S} \mathds{1}_{L_\ell}(q) |\I_m(\ell, \ell_1)| \nonumber\\
	&\leq \sum_{\ell \in \Zstar} \sum_{\ell_1 \in S} \mathds{1}_{L_\ell \cap L_{\ell_1}}(q) \Bigg(\sum\limits_{s \in L_{\ell}} \abs{K^m(\ell)_{q,s}}^2\Bigg)^\half
		\Bigg(\sum\limits_{s_1 \in L_{\ell_1}} \abs{K(\ell_1)_{q,s_1}}^2\Bigg)^\half
		\norm{ a_{q-\ell} (\NN+1) \xi_\lambda}
		\norm{ a_{q-\ell_1} \xi_\lambda }\nonumber\\
	&\leq \kF^{-1} e(q)^{-1}
		\Bigg( \sum_{\ell \in \Zstar} (C \hat{V}(\ell))^{2m} \Bigg)^{\half}
		\sum_{\ell_1 \in S} \hat{V}(\ell_1)
		\Bigg( \sum_{\ell \in \Z^3} \norm{ a_{q-\ell} (\NN+1) \xi_\lambda}^2 \Bigg)^{\half}
		\Xi^{\half} \nonumber\\
	&\leq C^m \kF^{-1} e(q)^{-1}
		\Bigg( \sum_{\ell_1 \in S} \hat{V}(\ell_1) \Bigg) \norm{(\NN+1)^{\frac 32} \xi_\lambda} \Xi^{\half} \nonumber\\
	&\leq C^m \kF^{-1 + \frac{3-\alpha}{2} \gamma} e(q)^{-1}
		\norm{(\NN+1)^{\frac 32} \xi_\lambda} \Xi^{\half} \;. \label{eq:estEQ2113_Coulomb_2}
\end{align}
Here, we used $ \sum_{|\ell_1| < \kF^\gamma} \hat{V}(\ell_1) \le \left( \sum_{\ell_1 \in \Zstar} \hat{V}(\ell_1)^2 |\ell_1|^\alpha \right)^{\half} \left( \sum_{|\ell_1| < \kF^\gamma} |\ell_1|^{-\alpha} \right)^{\half} \le C \kF^{\frac{3-\alpha}{2} \gamma} $.

For $ \sum_\ell \hat{V}(\ell) < \infty $, no splitting of the sum over $ \ell_1 $ is required and we gain a factor $ \Xi^{1-\varepsilon} $ by using Lemma~\ref{lem:Xi_halfminusepsilon}:
\begin{align}
	|\I_m(\ell, \ell_1)|
	&\leq \mathds{1}_{L_{\ell_1}}(q) \Bigg(\sum_{s \in L_{\ell}} \abs{K^m(\ell)_{q,s}}^2\Bigg)^\half
		\Bigg(\sum_{s_1 \in L_{\ell_1}} \abs{K(\ell_1)_{q,s_1}}^2\Bigg)^\half
		\norm{ a_{q-\ell} (\NN+1) \xi_\lambda}
		\norm{ a_{q-\ell_1} \xi_\lambda }\nonumber\\
	&\leq (C \hat{V}(\ell))^m \hat{V}(\ell_1) \kF^{-1} e(q)^{-1}
		\sup_{q' \in \Z^3} \norm{ a_{q'} (\NN+1) \xi_\lambda }\Xi^{\half} \nonumber\\
	&\leq C_\varepsilon (C \hat{V}(\ell))^m
		\hat{V}(\ell_1)
		\kF^{-1} e(q)^{-1} \Xi^{1-\varepsilon} \;. \label{eq:estEQ2113}
\end{align}
Summing up the three bounds concludes the proof for $ E^{m,1}_{Q_2}(q) $.
The bound for $ E^{m,2}_{Q_2}(q) $ is analogous, except for $ j = m $ where we proceed as in~\eqref{eq:estEQ1113_Coulomb_bis}:
\begin{align}
	&\sum_{\ell,\ell_1 \in \Zstar} \mathds{1}_{L_\ell}(q) |\I_m(\ell, \ell_1)| \nonumber\\
	&\leq \sum_{\ell,\ell_1 \in \Zstar} \mathds{1}_{L_\ell \cap (L_{\ell_1}+\ell-\ell_1) }(q)
		\Bigg(\sum\limits_{s \in L_{\ell}} \abs{K^m(\ell)_{q,s}}^2\Bigg)^\half
		\Bigg(\sum\limits_{s_1 \in L_{\ell_1}} \abs{K(\ell_1)_{q-\ell+\ell_1,s_1}}^2\Bigg)^\half \times \nonumber\\
	&\quad \times \Bigg( \sum_{s,s_1 \in \Z^3} \norm{ a_{-s} a_{-s+\ell} a_q (\NN+1)^{-1} \xi_\lambda}^2
		\norm{ a_{-s_1} a_{-s_1+\ell_1} a_{q-\ell+\ell_1} (\NN+1) \xi_\lambda }^2 \Bigg)^{\half} \nonumber\\
	&\leq \kF^{-1} e(q)^{-\frac 12}
		\Bigg(\sum_{\ell \in \Zstar} (C \hat{V}(\ell))^{2m} \Bigg)^\half
		\Bigg(\sum_{\ell_1 \in \Zstar} \hat{V}(\ell_1)^2 \Bigg)^\half \times \nonumber\\
	&\quad \times \Bigg( \sum_{\ell,s \in \Z^3} \norm{a_{-s+\ell} a_{-s} a_q (\NN+1)^{-1} \xi_\lambda}^2
		\sum_{\ell_1,s_1 \in \Z^3} \norm{a_{-s_1+\ell_1} a_{-s_1} (\NN+1) \xi_\lambda }^2 \Bigg)^{\half} \nonumber\\
	&\leq C^m \kF^{-1} e(q)^{-\frac 12} \eva{\xi_\lambda, a_q^* a_q \xi_\lambda}^{\half} \norm{ (\NN+1)^2 \xi_\lambda } \;. \label{eq:estEQ2113_Coulomb_bis}
\end{align}
In case $ \sum_\ell \hat{V}(\ell) < \infty $, we proceed similarly to~\eqref{eq:estEQ1113_bis}:
\begin{align}
	|\I_m(\ell, \ell_1)|
	&\leq \Bigg(\sum\limits_{s \in L_{\ell}} \abs{K^m(\ell)_{q,s}}^2\Bigg)^\half
		\Bigg(\sum\limits_{s_1 \in L_{\ell_1}} \abs{K(\ell_1)_{q-\ell+\ell_1,s_1}}^2\Bigg)^\half
		\norm{ a_{q-\ell+\ell_1} (\NN+1) \xi_\lambda}
		\norm{ a_q \xi_\lambda }\nonumber\\
	&\leq (C \hat{V}(\ell))^m \hat{V}(\ell_1) \kF^{-1} e(q)^{-\half}
		\sup_{q' \in \Z^3} \norm{ a_{q'} (\NN+1) \xi_\lambda }
		\eva{\xi_\lambda, a_q^* a_q \xi_\lambda}^{\half} \nonumber\\
	&\leq C_\varepsilon (C \hat{V}(\ell))^m
		\hat{V}(\ell_1)
		\kF^{-1} e(q)^{-\half} \Xi^{\half-\varepsilon} \eva{\xi_\lambda, a_q^* a_q \xi_\lambda}^{\half} \;. \label{eq:estEQ2113_bis}
\end{align}
This concludes the proof.
\end{proof}

\begin{lemma} \label{lem:EQ215}
Let $ \sum_{\ell \in \Zstar} \hat{V}(\ell)^2 < \infty$. For $\xi_\lambda = e^{-\lambda S} \Omega$, given $ \varepsilon > 0 $, there exist $ C, C_\varepsilon > 0 $ such that for all $ \lambda \in [0,1] $, odd $ m \in \NNN $, and $ q \in B_{\F}^c $,
\begin{align} \label{eq:estEQ215_Coulomb}
	\abs{\eva{\xi_\lambda,\left(E^{m,3}_{Q_2}(q) +\mathrm{h.c.}\right) \xi_\lambda }}
	\leq C_\varepsilon C^m \left( \kF^{-\frac 32}
		+ \kF^{-1 + \varepsilon} \Xi^\half \right)
		e(q)^{-1}
		\norm{(\NN+1) \xi_\lambda}^2 \;.
\end{align}
If $ \sum_{\ell \in \Zstar} \hat{V}(\ell) < \infty $, then
\begin{align} \label{eq:estEQ215}
	\abs{\eva{\xi_\lambda,\left(E^{m,3}_{Q_2}(q) + \mathrm{h.c.}\right) \xi_\lambda }}
	\leq C^m \kF^{-\frac{3}{2}} \Xi^{\half} e(q)^{-1}
		\norm{(\NN+1)^\half \xi_\lambda } \;.
\end{align}
\end{lemma}

\begin{proof}
Splitting the anticommutator in $ E^{m,3}_{Q_2}(q) $ by \eqref{eq:q-q} yields
\begin{equation} \label{eq:EQ2151}
\begin{aligned}
	&\abs{\eva{\xi_\lambda,\left(E^{m,3}_{Q_2}(q) + \mathrm{h.c.}\right) \xi_\lambda }}
	\le 4 \sum_{j=0}^m {{m}\choose j} \sum_{\ell,\ell_1 \in \Zstar}\!\! \mathds{1}_{L_\ell}(q) |\I_j(\ell, \ell_1)| \;,
	\end{aligned}
\end{equation}
where
\begin{equation}
\begin{aligned}
	& \I_j(\ell, \ell_1)
	\coloneq \sum_{\substack{r\in L_{\ell} \cap L_{\ell_1} \cap (-L_{\ell_1}+\ell+\ell_1)\\ s \in L_{\ell}}}
		\eva{\xi_\lambda, K^{m-j}(\ell)_{r,q} K^{j}(\ell)_{q,s} K(\ell_1)_{r,-r+\ell+\ell_1} a^*_{r-\ell_1} a^*_{r-\ell-\ell_1} b_{-s}(-\ell) \xi_\lambda} \;. \\
\end{aligned}
\end{equation}
For the case $j=0$, using Lemma~\ref{lem:normsk} we get
\begin{align}
	&\sum_{\ell,\ell_1 \in \Zstar} \mathds{1}_{L_\ell}(q) |\I_0(\ell, \ell_1)| \nonumber\\
	&\leq \sum_{\ell,\ell_1 \in \Zstar} \mathds{1}_{L_\ell}(q) \sum\limits_{\substack{r\in L_{\ell} \cap L_{\ell_1} \cap (-L_{\ell_1}+\ell+\ell_1)}}
		\norm{K^m(\ell)_{r,q} K(\ell_1)_{r,-r+\ell+\ell_1} a_{r-\ell-\ell_1} a_{r-\ell_1} (\NN+1)^\half \xi_\lambda} \times \nonumber\\
	&\quad \times \norm{ b_{-q}(-\ell) (\NN+1)^{-\half} \xi_\lambda}\nonumber\\
	&\leq \sum_{\ell \in \Zstar} (C \hat{V}(\ell))^m 
		\kF^{-1} e(q)^{-1} \sum\limits_{r \in L_\ell}
		\Bigg( \sum_{\ell_1 \in \Zstar} \mathds{1}_{L_{\ell_1} \cap (-L_{\ell_1}+\ell+\ell_1) }(r) |K(\ell_1)_{r,-r+\ell+\ell_1}|^2 \Bigg)^{\half} \times \nonumber\\
	&\quad \times \Bigg( \sum_{\ell_1 \in \Z^3} \norm{a_{r-\ell_1} (\NN+1)^\half \xi_\lambda}^2 \Bigg)^{\half}
		\norm{ a_{-q} a_{-q+\ell} (\NN+1)^{-\half} \xi_\lambda}\nonumber\\
	&\leq \Bigg( \sum_{\ell \in \Zstar} (C \hat{V}(\ell))^{2m} \Bigg)^{\half} 
		\kF^{-1} e(q)^{-1} \sum\limits_{r \in L_\ell} e(r)^{-1} \kF^{-1}
		\Bigg( \sum_{\ell_1 \in \Zstar} \hat{V}(\ell_1)^2 \Bigg)^{\half} \times \nonumber\\
	&\quad \times \norm{(\NN+1) \xi_\lambda}
		\Bigg( \sum_{\ell \in \Z^3} \norm{a_{-q+\ell} a_{-q} (\NN+1)^{-\half} \xi_\lambda}^2 \Bigg)^{\half} \nonumber\\
	&\leq C_\varepsilon C^m \kF^{-1 + \varepsilon} e(q)^{-1}
		\norm{ (\NN+1) \xi_\lambda } \Xi^{\half} \;. \label{eq:estEQ2151_Coulomb}
\end{align}
Note that we used~\eqref{eq:K_element_bounds} $ |K(\ell_1)_{r,-r+\ell+\ell_1}| \le C \kF^{-1} \hat{V}(\ell_1) \lambda_{\ell_1,r}^{-1} $ with $ \lambda_{\ell_1,r} \ge C e(r) $ and then~\eqref{eq:lambdainverse}.

The bound for $ \sum_{\ell} \hat{V}(\ell) < \infty $ is considerably easier and stronger:
\begin{align}
	|\I_0(\ell, \ell_1)|
	&\leq \sum_{r\in L_{\ell} \cap L_{\ell_1} \cap (-L_{\ell_1}+\ell+\ell_1)} \norm{K^m(\ell)_{r,q} K(\ell_1)_{r,-r+\ell+\ell_1} a_{r-\ell-\ell_1} a_{r-\ell_1} \xi_\lambda}\norm{ b_{-q}(-\ell) \xi_\lambda}\nonumber\\
	&\leq (C \hat{V}(\ell))^m \kF^{-1} e(q)^{-1}
		\norm{K(\ell_1)}_{\max,2} \norm{ (\NN+1)^\half \xi_\lambda } \Xi^{\half} \nonumber\\
	&\leq (C \hat{V}(\ell))^m
		\hat{V}(\ell_1)
		\kF^{-\frac 32} e(q)^{-1}
		\norm{ (\NN+1)^\half \xi_\lambda } \Xi^{\half} \;. \label{eq:estEQ2151}
\end{align}
For $ 1 \le j \le m-1 $, which only happens if $ m \ge 2 $, we proceed as follows:
\begin{align}
	&\sum_{\ell,\ell_1 \in \Zstar} \mathds{1}_{L_\ell}(q) |\I_j(\ell, \ell_1)| \nonumber\\
	&\leq \sum_{\ell,\ell_1 \in \Zstar} \mathds{1}_{L_\ell}(q) \sum\limits_{r\in L_{\ell} \cap L_{\ell_1} \cap (-L_{\ell_1}+\ell+\ell_1)}
		\norm{ K^{m-j}(\ell)_{r,q} K(\ell_1)_{r,-r+\ell+\ell_1} a_{r-\ell-\ell_1} a_{r-\ell_1} \xi_\lambda} \times \nonumber\\
	&\quad \times \sum\limits_{s \in L_{\ell}}
		\norm{ K^j(\ell)_{q,s} b_{-s}(-\ell) \xi_\lambda }\nonumber\\
	&\leq \sum_{\ell \in \Zstar} (C \hat{V}(\ell))^{m-j} \kF^{-1} e(q)^{-1}
		\sum_{r \in L_\ell}
		\Bigg( \sum_{\ell_1 \in \Zstar} \mathds{1}_{L_{\ell_1} \cap (-L_{\ell_1}+\ell+\ell_1) }(r) |K(\ell_1)_{r,-r+\ell+\ell_1}|^2 \Bigg)^{\half} \times \nonumber\\
	&\quad \times \Bigg( \sum_{\ell_1 \in \Z^3} \norm{ a_{r-\ell_1} \xi_\lambda}^2 \Bigg)^{\half}
		\norm{K^j(\ell)}_{\max,1}
		\Xi^\half \nonumber\\
	&\leq \sum_{\ell \in \Zstar} (C \hat{V}(\ell))^m \kF^{-1} e(q)^{-1}
		\sum_{r \in L_\ell} e(r)^{-1} \kF^{-1}
		\norm{ (\NN+1)^{\half} \xi_\lambda}
		\Xi^\half \nonumber\\
	&\leq C_\varepsilon C^m \kF^{-1+\varepsilon} e(q)^{-1}
		\norm{ (\NN+1)^{\half} \xi_\lambda}
		\Xi^\half \;. \label{eq:estEQ2152_Coulomb}
\end{align}
The corresponding stronger bound for $ \sum_\ell \hat{V}(\ell) < \infty $ with $ q \in L_\ell $ is
\begin{align}
	|\I_j(\ell, \ell_1)|
	&\leq \sum_{\substack{r\in L_{\ell} \cap L_{\ell_1}\\ \cap (-L_{\ell_1}+\ell+\ell_1)}}
		\norm{ K^{m-j}(\ell)_{r,q} K(\ell_1)_{r,-r+\ell+\ell_1} a_{r-\ell-\ell_1} a_{r-\ell_1} \xi_\lambda}
		\sum_{s \in L_{\ell}}
		\norm{ K^j(\ell)_{q,s} b_{-s}(-\ell) \xi_\lambda }\nonumber\\
	&\leq (C \hat{V}(\ell))^{m-j} \kF^{-1} e(q)^{-1}
		\norm{K(\ell_1)}_{\max,1} \Xi^{\half}
		\norm{K^j(\ell)}_{\max,2}
		\norm{ (\NN+1)^\half \xi_\lambda} \nonumber\\
	&\leq (C \hat{V}(\ell))^m
		\hat{V}(\ell_1)
		\kF^{-\frac 32} e(q)^{-1} \Xi^\half
		\norm{ (\NN+1)^\half \xi_\lambda} \;. \label{eq:estEQ2152}
\end{align}
Finally, for the case $ j = m $, we have
\begin{align}
	&\sum_{\ell,\ell_1 \in \Zstar} \mathds{1}_{L_\ell}(q) |\I_m(\ell, \ell_1)| \nonumber\\
	&\leq \sum_{\ell,\ell_1 \in \Zstar} \mathds{1}_{L_\ell \cap L_{\ell_1} \cap (-L_{\ell_1} + \ell + \ell_1)}(q) \norm{K(\ell_1)_{q,-q+\ell+\ell_1} a_{q-\ell-\ell_1} a_{q-\ell_1} \xi_\lambda} 
		\sum\limits_{ s \in L_{\ell}}
		\norm{ K^m(\ell)_{q,s} b_{-s}(-\ell) \xi_\lambda }\nonumber\\
	&\leq \! \Bigg( \! \sum_{\ell_1 \in \Zstar} \! \hat{V}(\ell_1)^2 \! \Bigg)^{\half} 
		\! \kF^{-\frac 32} e(q)^{-1}
		\Bigg( \! \sum_{\ell,\ell_1 \in \Z^3} \! \norm{a_{q-\ell-\ell_1} a_{q-\ell_1}  \xi_\lambda}^2 \! \Bigg)^{\half}
		\Bigg( \! \sum_{\ell \in \Zstar} (C \hat{V}(\ell))^{2m} \! \Bigg)^{\half} 
		\Bigg( \! \sum\limits_{ s \in \Z^3} \! \norm{ a_{-s} \xi_\lambda }^2 \! \Bigg)^{\half} \nonumber\\
	&\leq C^m \kF^{-\frac 32} e(q)^{-1}
		\norm{(\NN+1) \xi_\lambda}^2 \;. \label{eq:estEQ2153_Coulomb}
\end{align}
The improved bound for $ \sum_\ell \hat{V}(\ell) < \infty $ works as follows:
\begin{align}
	|\I_m(\ell, \ell_1)|
	&\leq \mathds{1}_{L_{\ell_1} \cap (-L_{\ell_1} + \ell + \ell_1)}(q) \norm{K(\ell_1)_{q,-q+\ell+\ell_1} a_{q-\ell-\ell_1} a_{q-\ell_1} \xi_\lambda} \times  \label{eq:estEQ2153} \\
	& \quad \times
		\sum_{ s \in L_{\ell}}
		\norm{ K^m(\ell)_{q,s} b_{-s}(-\ell) \xi_\lambda }\nonumber\\
	&\leq (C \hat{V}(\ell))^m
		\hat{V}(\ell_1)
		\kF^{-\frac 32} e(q)^{-1} \Xi^\half
		\norm{(\NN+1)^\half\xi_\lambda} \;. \qedhere \nonumber
\end{align}
\end{proof}

\begin{lemma} \label{lem:EQ217}
Let $ \sum_{\ell \in \Zstar} \hat{V}(\ell)^2 < \infty $. For $\xi_\lambda = e^{-\lambda S} \Omega$, there exists a constant $ C > 0 $ such that for all $ \lambda \in [0,1] $, odd $ m \in \NNN $, and $ q \in B_{\F}^c $,
\begin{align}
	\abs{\eva{\xi_\lambda,\left(E^{m,4}_{Q_2}(q) + \mathrm{h.c.}\right) \xi_\lambda }}
	\leq C^m \kF^{-1} \Xi^{\half}  e(q)^{-1}
		\norm{(\NN+1) \xi_\lambda } \;. \label{eq:estEQ217_Coulomb}
\end{align}
If $ \sum_{\ell \in \Zstar} \hat{V}(\ell) < \infty $, then
\begin{align}
	\abs{\eva{\xi_\lambda,\left(E^{m,4}_{Q_2}(q) + \mathrm{h.c.}\right) \xi_\lambda }}
	\leq C^m \kF^{-\frac 32} \Xi^{\half}  e(q)^{-1}
		\norm{(\NN+1) \xi_\lambda } \;. \label{eq:estEQ217}
\end{align}
\end{lemma}

\begin{proof}
Splitting the anticommutator in $ E^{m,4}_{Q_2}(q) $ by \eqref{eq:q-q} gives
\begin{equation} \label{eq:EQ2171}
\begin{aligned}
	& \abs{\eva{\xi_\lambda,\left(E^{m,4}_{Q_2}(q) + \mathrm{h.c.}\right) \xi_\lambda }}
	\le 4 \sum_{j=0}^m {{m}\choose j} \sum_{\ell,\ell_1 \in \Zstar}\!\! \mathds{1}_{L_\ell}(q) |\I_j(\ell, \ell_1)| \;,
	\end{aligned}
\end{equation}
where
\begin{equation}
\begin{aligned}
	& \I_j(\ell, \ell_1)
	\coloneq \sum_{\substack{r\in L_{\ell} \cap L_{\ell_1} \cap (-L_{\ell}+\ell+\ell_1)\\ s_1 \in L_{\ell_1}}} \!\!
		\eva{\xi_\lambda, K^{m-j}(\ell)_{r,q} K^{j}(\ell)_{q,-r+\ell+\ell_1} K(\ell_1)_{r,s_1} b^*_{-s_1}(-\ell_1) a_{r-\ell-\ell_1} a_{r-\ell} \xi_\lambda} \;. \\
\end{aligned}
\end{equation}
For the case $ j = 0 $, Lemma~\ref{lem:normsk} yields
\begin{align}
	&\sum_{\ell,\ell_1 \in \Zstar} \mathds{1}_{L_\ell}(q) |\I_0(\ell, \ell_1)| \nonumber\\
	&\leq \sum_{\ell,\ell_1 \in \Zstar} \mathds{1}_{L_\ell \cap (-L_\ell + \ell + \ell_1) \cap (-L_{\ell_1} + \ell + \ell_1)}(q) \sum\limits_{s_1 \in L_{\ell_1}}
		\norm{K(\ell_1)_{-q+\ell+\ell_1,s_1} b_{-s_1}(-\ell_1)  \xi_\lambda} \times \nonumber\\
	&\quad \times \norm{ K^m(\ell)_{-q+\ell+\ell_1,q}a_{-q}a_{-q+\ell_1}  \xi_\lambda } \nonumber\\
	&\leq \sum_{\ell_1 \in \Zstar}
		\Bigg( \sum_{\ell \in \Zstar} \mathds{1}_{-L_{\ell_1} + \ell + \ell_1}(q) \sum\limits_{s_1 \in L_{\ell_1}} |K(\ell_1)_{-q+\ell+\ell_1,s_1}|^2 \Bigg)^{\half}
		\Bigg( \sum\limits_{s_1 \in L_{\ell_1}} \norm{ b_{-s_1}(-\ell_1)  \xi_\lambda}^2 \Bigg)^{\half} \times \nonumber\\
	&\quad \times \Bigg( \sum_{\ell \in \Zstar} (C \hat{V}(\ell))^{2m} \Bigg)^{\half}
		e(q)^{-1} \kF^{-1}
		\norm{a_{-q}a_{-q+\ell_1}  \xi_\lambda } \nonumber\\
	&\leq C^m \Bigg( \sum_{\ell_1 \in \Zstar}
		\norm{K(\ell_1)}_{\HS}^2 \Bigg)^{\half}
		\Bigg( \sum_{\ell_1,s_1 \in \Z^3} \norm{a_{-s_1+\ell_1} a_{-s_1} \xi_\lambda}^2 \Bigg)^{\frac 12}
		\kF^{-1} e(q)^{-1} \Xi^\half \nonumber\\
	&\leq C^m \norm{(\NN+1) \xi_\lambda}
		\kF^{-1} e(q)^{-1} \Xi^\half \;. \label{eq:estEQ2171_Coulomb}
\end{align}
In case $ \sum_\ell \hat{V}(\ell) < \infty $, we get a stronger bound:
\begin{align}
	|\I_0(\ell, \ell_1)|
	&\leq \sum_{s_1 \in L_{\ell_1}}
		\norm{K(\ell_1)_{-q+\ell+\ell_1,s_1} b_{-s_1}(-\ell_1) \xi_\lambda}
		\norm{ K^m(\ell)_{-q+\ell+\ell_1,q}a_{-q}a_{-q+\ell_1} \xi_\lambda } \nonumber\\
	&\leq (C \hat{V}(\ell))^m
		\hat{V}(\ell_1)
		\kF^{-\frac 32} e(q)^{-1}
		\norm{(\NN+1)^\half\xi_\lambda} \Xi^\half \;. \label{eq:estEQ2171}
\end{align}
In case $ j = m $, the bound~\eqref{eq:estEQ2171} is analogous, while~\eqref{eq:estEQ2171_Coulomb} is replaced by
\begin{align}
	&\sum_{\ell,\ell_1 \in \Zstar} \mathds{1}_{L_\ell}(q) |\I_m(\ell, \ell_1)| \nonumber\\
	&\leq \sum_{\ell,\ell_1 \in \Zstar} \mathds{1}_{L_\ell \cap L_{\ell_1} \cap (-L_\ell + \ell + \ell_1)}(q) \sum\limits_{s_1 \in L_{\ell_1}}
		\norm{K(\ell_1)_{q,s_1} b_{-s_1}(-\ell_1)  \xi_\lambda} 
		\norm{ K^m(\ell)_{-q+\ell+\ell_1,q}a_{q-\ell-\ell_1}a_{q-\ell}  \xi_\lambda } \nonumber\\
	&\leq \sum_{\ell,\ell_1 \in \Zstar} 
		\norm{K(\ell_1)}_{\max,1} \Xi^{\half}
		k_{\F}^{-1} e(q)^{-1} (C \hat{V}(\ell))^m \norm{ a_{q-\ell-\ell_1}a_{q-\ell}  \xi_\lambda } \nonumber\\
	&\leq \Bigg( \sum_{\ell,\ell_1 \in \Zstar} \hat{V}(\ell_1)^2 (C \hat{V}(\ell))^{2m} \Bigg)^{\half}
		\Xi^{\half}
		k_{\F}^{-1} e(q)^{-1} \Bigg( \sum_{\ell_1,\ell \in \Zstar}  \norm{ a_{q-\ell-\ell_1}a_{q-\ell}  \xi_\lambda }^2 \Bigg)^{\half} \nonumber\\
	&\leq C^m \Xi^\half
		\kF^{-1} e(q)^{-1}  \norm{(\NN+1) \xi_\lambda} \;.
\end{align}
Finally, for the case $ 1 \le j \le m-1 $, which only occurs for $ m \ge 2 $:
\begin{align}
	&\sum_{\ell,\ell_1 \in \Zstar} \mathds{1}_{L_\ell}(q) |\I_j(\ell, \ell_1)| \nonumber\\
	&\leq \sum_{\ell,\ell_1 \in \Zstar} \mathds{1}_{L_\ell}(q) \sum\limits_{r\in L_{\ell} \cap L_{\ell_1} \cap (-L_{\ell}+\ell+\ell_1)}
		\sum\limits_{s_1\in L_{\ell_1}} 
		\norm{K(\ell_1)_{r,s_1} b_{-s_1}(-\ell_1) \xi_\lambda} \times \nonumber\\
	&\quad \times \norm{ K^{m-j}(\ell)_{r,q} K^j(\ell)_{q,-r+\ell+\ell_1} a_{r-\ell-\ell_1} a_{r-\ell} \xi_\lambda } \nonumber\\
	&\leq \sum_{\ell,\ell_1 \in \Zstar} \mathds{1}_{L_\ell}(q) \Bigg( \sum\limits_{r,s_1\in L_{\ell_1}} 
		|K(\ell_1)_{r,s_1}|^2 \Bigg)^{\half}
		\Bigg( \sum\limits_{s_1 \in \Z^3}
		\norm{a_{-s_1} a_{-s_1+\ell_1} \xi_\lambda}^2 \Bigg)^{\half} \times \nonumber\\
		&\quad \times \kF^{-1} e(q)^{-1} (C \hat{V}(\ell))^{m-j}
		\Bigg( \sum\limits_{r\in (-L_{\ell}+\ell+\ell_1)} 
		|K^j(\ell)_{q,-r+\ell+\ell_1}|^2
		\norm{ a_{r-\ell} \xi_\lambda }^2 \Bigg)^{\half} \nonumber\\
	&\leq  \Bigg( \sum_{\ell_1 \in \Zstar} 
		\hat{V}(\ell_1)^2 \Bigg)^{\half}
		\Bigg( \sum\limits_{\ell_1,s_1 \in \Z^3}
		\norm{a_{-s_1+\ell_1} a_{-s_1} \xi_\lambda}^2 \Bigg)^{\half}
		\kF^{-\frac 32} e(q)^{-1} \sum_{\ell \in \Zstar} (C \hat{V}(\ell))^m \Xi^\half \nonumber\\
	&\leq C^m \kF^{-\frac 32} e(q)^{-1}
		\norm{(\NN+1) \xi_\lambda} \Xi^{\half} \;. \label{eq:estEQ2172}
\end{align}
As later we will show $ \Xi \sim \kF^{-1} $, the bound \eqref{eq:estEQ2172} will be of order $ \sim \kF^{-2} $. The same is applicable for $ \sum_\ell \hat{V}(\ell) < \infty $.
\end{proof}

\begin{lemma} \label{lem:EQ213}
Let $ \sum_{\ell \in \Zstar} \hat{V}(\ell)^2 < \infty $. For $\xi_\lambda = e^{-\lambda S} \Omega$, there exists a constant $ C > 0 $ such that for all $ \lambda \in [0,1] $, odd $ m \in \NNN $, and $ q \in B_{\F}^c $,
\begin{equation}
\begin{aligned}
	\abs{\eva{\xi_\lambda,\left(E^{m,5}_{Q_2}(q) + E^{m,6}_{Q_2}(q) + E^{m,7}_{Q_2}(q) + \mathrm{h.c.}\right) \xi_\lambda }}
	\leq C^m \left( \kF^{-1} \Xi^{\half}
		+ \kF^{-\frac 32} \right)
		e(q)^{-1}
		\norm { (\NN+1)^{\frac 32} \xi_\lambda }^2 \,. \label{eq:estEQ213_Coulomb}
\end{aligned}
\end{equation}
If $ \sum_{\ell \in \Zstar} \hat{V}(\ell) < \infty $, then
\begin{equation}
\begin{aligned}
	\abs{\eva{\xi_\lambda,\left(E^{m,5}_{Q_2}(q) + E^{m,6}_{Q_2} (q) + E^{m,7}_{Q_2}(q) + \mathrm{h.c.}\right) \xi_\lambda }}
	\leq C^m \kF^{-\frac 32} \Xi^{\half} e(q)^{-1}
		\norm { (\NN+1)^{\frac 32} \xi_\lambda } \;. \label{eq:estEQ213}
\end{aligned}
\end{equation}
\end{lemma}

\begin{proof}
We start with estimating $ E^{m,5}_{Q_2}(q) $: Splitting via \eqref{eq:q-q} yields
\begin{equation} \label{eq:EQ2131}
\begin{aligned}
	& \abs{\eva{\xi_\lambda,\left(E^{m,5}_{Q_2}(q) +\mathrm{h.c.}\right) \xi_\lambda }}
	\le 4 \sum_{j=0}^m {{m}\choose j} \sum_{\ell,\ell_1  \in \Zstar}\!\! \mathds{1}_{L_\ell}(q) |\I_j(\ell, \ell_1)| \;,
	\end{aligned}
\end{equation}
where
\begin{equation}
\begin{aligned}
	& \I_j(\ell, \ell_1)
	\coloneq \sum_{\substack{r\in L_{\ell} \cap L_{\ell_1}\\ s \in (L_{\ell} - \ell) \cap (L_{\ell_1} - \ell_1)}}
		\eva{\xi_\lambda, K^{m-j}(\ell)_{r,q} K^{j}(\ell)_{q,s+\ell} K(\ell_1)_{r,s+\ell_1} a^*_{r-\ell_1} a^*_{-s-\ell_1} a_{-s-\ell} a_{r-\ell} \xi_\lambda} \;. \\
\end{aligned}
\end{equation}
For the case $ j = 0 $, the Cauchy--Schwarz inequality and Lemma~\ref{lem:normsk} implies
\begin{align}
	&\sum_{\ell,\ell_1 \in \Zstar} \mathds{1}_{L_\ell}(q) |\I_0(\ell, \ell_1)| \nonumber\\
	&\leq \sum_{\ell,\ell_1 \in \Zstar} \mathds{1}_{L_\ell \cap (L_{\ell_1}+\ell-\ell_1)}(q) \sum\limits_{r\in L_{\ell_1} \cap L_\ell}
		\norm{ K(\ell_1)_{r,q-\ell+\ell_1} a_{-q+\ell-\ell_1} a_{r-\ell_1} \xi_\lambda}
		\norm{ K^m(\ell)_{r,q} a_{-q}a_{r-\ell} \xi_\lambda } \nonumber\\
	&\leq \sum_{\ell_1 \in \Zstar}
		\Bigg( \sum_{r\in L_{\ell_1}} \sum_{\ell \in \Z^3} \mathds{1}_{L_{\ell_1}+\ell-\ell_1}(q) |K(\ell_1)_{r,q-\ell+\ell_1}|^2 \Bigg)^{\half} \times \nonumber\\
	&\quad \times \kF^{-1} e(q)^{-1} \Bigg( \sum_{\ell \in \Zstar} (C \hat{V}(\ell))^{2m} \sum_{r\in \Z^3} \norm{a_{-q+\ell-\ell_1} a_{r-\ell_1} \xi_\lambda}^2
		\norm{ a_{-q} \xi_\lambda }^2 \Bigg)^{\half} \nonumber\\
	&\leq \Bigg( \sum_{\ell_1 \in \Zstar} \hat{V}(\ell_1)^2 \Bigg)^{\half} \kF^{-1} e(q)^{-1}
		\Bigg( \sum_{\ell \in \Zstar} (C \hat{V}(\ell))^{2m} \sum_{r,\ell_1 \in \Z^3} \norm{a_{r-\ell_1} a_{-q+\ell-\ell_1} \xi_\lambda}^2 \Bigg)^{\half} \Xi^{\half} \nonumber\\
	&\leq C^m \kF^{-1} e(q)^{-1}
		\norm{ (\NN+1) \xi_\lambda} \Xi^{\half} \;. \label{eq:estEQ2131_Coulomb}
\end{align}
In case $ \sum_\ell \hat{V}(\ell) < \infty $, we get a simpler and stronger bound:
\begin{align}
	|\I_0(\ell, \ell_1)|
 	&\leq \left(\sum_{r\in L_{\ell_1}} \norm{ K(\ell_1)_{r,q-\ell+\ell_1} a_{r-\ell_1}(\NN+1)^{\half} \xi_\lambda}^2\right)^\half \times \nonumber\\
	& \quad \times \left(\sum_{r\in L_{\ell}} \norm{ K^m(\ell)_{r,q} a_{-q}a_{r-\ell} (\NN+1)^{-\half} \xi_\lambda }^2 \right)^\half \nonumber\\
	&\leq (C \hat{V}(\ell))^m \hat{V}(\ell_1) \kF^{-2} e(q)^{-1} \norm{ (\NN+1) \xi_\lambda} \norm{ a_{-q} \xi_\lambda } \nonumber \\
	&\leq (C \hat{V}(\ell))^m
		\hat{V}(\ell_1)
		\kF^{-2} e(q)^{-1}
		\norm{ (\NN+1) \xi_\lambda} \Xi^{\half} \;. \label{eq:estEQ2131}
\end{align}
In the case $ j = m $, we proceed as follows:
\begin{align}
	&\sum_{\ell,\ell_1 \in \Zstar} \mathds{1}_{L_\ell}(q) |\I_m(\ell, \ell_1)| \nonumber\\
	&\leq \sum_{\ell,\ell_1 \in \Zstar} \mathds{1}_{L_\ell \cap L_{\ell_1}}(q) \sum\limits_{s \in (L_{\ell_1} - \ell_1) \cap (L_\ell - \ell)}
		\norm{ K(\ell_1)_{q,s+\ell_1} a_{-s-\ell_1} a_{q-\ell_1} \xi_\lambda}
		\norm{ K^m(\ell)_{q,s+\ell} a_{-s-\ell} a_{q-\ell} \xi_\lambda } \nonumber\\
	&\leq \sum_{\ell,\ell_1 \in \Zstar}
		\norm{K(\ell_1)}_{\max,2}
		\norm{  a_{q-\ell_1} \xi_\lambda}
		e(q)^{-1} \kF^{-1} (C \hat{V}(\ell))^m
	\Bigg( \sum_{s\in \Z^3} \norm{a_{-s-\ell} a_{q-\ell} \xi_\lambda }^2 \Bigg)^{\half} \nonumber\\
	&\leq e(q)^{-1} \kF^{-\frac 32}
		\Bigg( \sum_{\ell_1 \in \Zstar} \hat{V}(\ell_1)^2 \Bigg)^{\half}
		\Bigg( \sum_{\ell \in \Zstar} (C \hat{V}(\ell))^{2m} \Bigg)^{\half}
		\Bigg( \sum_{s, \ell, \ell_1 \in \Z^3} \norm{ a_{q-\ell_1} \xi_\lambda}^2
		\norm{a_{-s-\ell} a_{q-\ell} \xi_\lambda }^2 \Bigg)^{\half} \nonumber\\
	&\leq C^m e(q)^{-1} \kF^{-\frac 32}
		\norm{(\NN+1) \xi_\lambda }^2
		\;. \label{eq:estEQ2131_Coulomb_bis}
\end{align}
For $ \sum_\ell \hat{V}(\ell) < \infty $, an analogous bound to~\eqref{eq:estEQ2131} also applies to $ |\I_m(\ell, \ell_1)| $. For the case $ 1 \le j \le m-1 $, we have $ m \ge 2 $, so $ \sum_\ell \hat{V}(\ell)^m < \infty $:
\begin{align}
	&\sum_{\ell,\ell_1 \in \Zstar} \mathds{1}_{L_\ell}(q) |\I_j(\ell, \ell_1)| \nonumber\\
	&\leq \sum_{\ell,\ell_1 \in \Zstar} \mathds{1}_{L_\ell}(q) \sum\limits_{\substack{r\in L_{\ell} \cap L_{\ell_1}\\ s \in (L_{\ell} - \ell) \cap (L_{\ell_1} - \ell_1)}}
		\norm{ K^{j}(\ell)_{q,s+\ell} a_{-s-\ell_1} a_{r-\ell_1} (\NN+1)^{\half} \xi_\lambda} \nonumber \times \\ 
	&\quad \times \norm{ K^{m-j}(\ell)_{r,q} K(\ell_1)_{r,s+\ell_1} a_{-s-\ell}a_{r-\ell} (\NN+1)^{-\half} \xi_\lambda } \nonumber\\
	&\leq \sum_{\ell,\ell_1 \in \Zstar} \mathds{1}_{L_\ell}(q)
		\sum_{r\in \Z^3} \sum_{s \in L_{\ell} - \ell}
		|K^{j}(\ell)_{q,s+\ell}|
		\norm{ a_{-s-\ell_1} a_{r-\ell_1} (\NN+1)^{\half} \xi_\lambda} \times \nonumber\\
	&\quad \times (C \hat{V}(\ell))^{m-j} \hat{V}(\ell_1) \kF^{-2} e(q)^{-1}
		\norm{ a_{-s-\ell}a_{r-\ell} (\NN+1)^{-\half} \xi_\lambda } \nonumber\\
	&\leq \sum_{\ell \in \Zstar} \mathds{1}_{L_\ell}(q) (C \hat{V}(\ell))^{m-j} \kF^{-2} e(q)^{-1}
		\sum_{s \in L_{\ell} - \ell}
		|K^{j}(\ell)_{q,s+\ell}| \times \nonumber\\
	&\quad \times \Bigg( \sum_{r,\ell_1 \in \Z^3} \norm{a_{r-\ell_1} a_{-s-\ell_1} (\NN+1)^{\half} \xi_\lambda}^2 \Bigg)^{\half}
		\Bigg( \sum_{r \in \Z^3} \norm{a_{r-\ell} a_{-s-\ell} (\NN+1)^{-\half} \xi_\lambda }^2 \Bigg)^{\half} \nonumber\\
	&\leq C^m \kF^{-2} e(q)^{-1}
		\norm{(\NN+1)^{\frac 32} \xi_\lambda} \Xi^{\half} \;. \label{eq:estEQ2132}
\end{align} 
This establishes the desired bound for $ E^{m,5}_{Q_2}(q) $.

Regarding $ E^{m,6}_{Q_2}(q) $ and $ E^{m,7}_{Q_2}(q) $, the case $ 1 \le j \le m-1 $ is treated as in~\eqref{eq:estEQ2132} and the case $ j \in \{0,m\} $ for $ \sum_\ell \hat{V}(\ell) < \infty $ is treated as in~\eqref{eq:estEQ2131}. The general case $ j \in \{0,m\} $ is treated as in~\eqref{eq:estEQ2131_Coulomb_bis} for $ E^{m,6}_{Q_2}(q) $ and as in~\eqref{eq:estEQ2131_Coulomb} for $ E^{m,7}_{Q_2}(q) $. 
\end{proof}

\begin{lemma} \label{lem:EQ218}
Let $ \sum_{\ell \in \Zstar} \hat{V}(\ell)^2 < \infty $. For $\xi_\lambda = e^{-\lambda S} \Omega$, there exists a constant $ C > 0 $ such that for all $ \lambda \in [0,1] $, odd $ m \in \NNN $, and $ q \in B_{\F}^c $,
\begin{align}
	\abs{\eva{\xi_\lambda,\left(E^{m,8}_{Q_2}(q) + E^{m,9}_{Q_2}(q) + \mathrm{h.c.}\right) \xi_\lambda }}
	\leq C^m \left( \kF^{-\frac 32} \Xi^\half
		+ \kF^{-1} \Xi \right)
		e(q)^{-1} \norm{(\NN+1)^{\half} \xi_\lambda} \;.\label{eq:estEQ218}
\end{align}
\end{lemma}

\begin{proof}
We first focus on $ E^{m,8}_{Q_2}(q) $.
Splitting the anticommutator in $ E^{m,8}_{Q_2}(q) $ by \eqref{eq:q-q} we get
\begin{equation} \label{eq:EQ2181}
\begin{aligned}
	& \abs{\eva{\xi_\lambda,\left(E^{m,8}_{Q_2}(q) + \mathrm{h.c.}\right) \xi_\lambda }}
	\le 4 \sum_{j=0}^m {{m}\choose j} \sum_{\ell,\ell_1 \in \Zstar}\!\! \mathds{1}_{L_\ell}(q) |\I_j(\ell, \ell_1)| \;,
	\end{aligned}
\end{equation}
where
\begin{equation}
\begin{aligned}
	& \I_j(\ell, \ell_1)
	\coloneq \sum_{\substack{r\in L_{\ell} \cap L_{\ell_1}\\\cap (-L_{\ell}+\ell+\ell_1) \\ \cap (-L_{\ell_1}+\ell+\ell_1)}}
		\eva{\xi_\lambda, K^{m-j}(\ell)_{r,q} K^{j}(\ell)_{q,-r+\ell+\ell_1} K(\ell_1)_{r,-r+\ell+\ell_1} a^*_{r-\ell_1} a_{r-\ell_1} \xi_\lambda} \;. \\
\end{aligned}
\end{equation}
For the case $ j = 0 $,
\begin{align}
	&\sum_{\ell,\ell_1 \in \Zstar} \mathds{1}_{L_\ell}(q) |\I_0(\ell, \ell_1)| \nonumber\\
	&\leq \sum_{\ell,\ell_1 \in \Zstar} \mathds{1}_{L_\ell \cap L_{\ell_1} \cap (-L_{\ell_1} + \ell + \ell_1) \cap (-L_\ell + \ell + \ell_1)}(q)
		\norm{ K(\ell_1)_{-q+\ell+\ell_1,q} a_{-q+\ell} \xi_\lambda}
		\norm{ K^m(\ell)_{-q+\ell+\ell_1,q} a_{-q+\ell} \xi_\lambda } \nonumber\\
	&\leq C \sum_{\ell \in \Zstar} \kF^{-1} e(q)^{-1}
		\Bigg( \sum_{\ell_1 \in \Zstar} \hat{V}(\ell_1)^2 \Bigg)^{\half}
	\Bigg( \sum_{\ell_1 \in \Zstar} \mathds{1}_{L_\ell \cap (-L_\ell + \ell + \ell_1)}(q) |K^m(\ell)_{-q+\ell+\ell_1,q}|^2 \Bigg)^{\half}
		\norm{ a_{-q+\ell} \xi_\lambda }^2 \nonumber\\
	&\leq \kF^{-\frac 32} e(q)^{-1} \Xi^\half
		\Bigg( \sum_{\ell \in \Zstar} (C \hat{V}(\ell))^{2m} \Bigg)^{\half}
		\Bigg( \sum_{\ell \in \Z^3} \norm{a_{-q+\ell} \xi_\lambda }^2 \Bigg)^{\half} \nonumber\\
	&\leq C^m \kF^{-\frac 32} e(q)^{-1} \Xi^{\half} \norm{(\NN+1)^{\half} \xi_\lambda} \;. \label{eq:estEQ2181}
\end{align}
An analogous bound holds for $ j = m $. Finally, for the case $ 1 \le j \le m-1 $, which only happens for $ m \ge 2 $,
\begin{align}
	&\sum_{\ell,\ell_1 \in \Zstar} \mathds{1}_{L_\ell}(q) |\I_j(\ell, \ell_1)| \nonumber\\
	&\leq \sum_{\ell,\ell_1 \in \Zstar} \mathds{1}_{L_\ell}(q) \sum\limits_{\substack{r\in L_{\ell} \cap L_{\ell_1} \cap (-L_{\ell}+\ell+\ell_1) \\ \cap (-L_{\ell_1}+\ell+\ell_1)}} \norm{ K(\ell_1)_{r,-r+\ell+\ell_1} a_{r-\ell_1} \xi_\lambda} \norm{ K^{m-j}(\ell)_{r,q} K^j(\ell)_{q,-r+\ell+\ell_1} a_{r-\ell_1} \xi_\lambda} \nonumber\\
	&\leq \sum_{\ell \in \Zstar} \mathds{1}_{L_\ell}(q)
		\Bigg( \sum_{\ell_1 \in \Zstar} \hat{V}(\ell_1)^2 \Bigg)^{\half}
		\Xi^{\half} \kF^{-\frac 32} e(q)^{-1} (C \hat{V}(\ell))^j
		\Bigg( \sum_{r \in L_{\ell}} |K^{m-j}(\ell)_{q,r}|^2 
		\sum_{\ell_1 \in \Z^3} \norm{a_{r-\ell_1} \xi_\lambda}^2 \Bigg)^{\half} \nonumber\\
	&\leq C^m \kF^{-2} \Xi^{\half} e(q)^{-1} \norm{(\NN+1)^{\half} \xi_\lambda} \;. \label{eq:estEQ2182}
\end{align}
For $ E^{m,9}_{Q_2}(q) $, the bound is analogous up to the following modification: The term for $ j = 0 $ is
\begin{equation}
	\I_0(\ell,\ell_1)
	\coloneq \mathds{1}_{L_{\ell_1} \cap (-L_{\ell_1} + \ell+\ell_1) \cap (-L_\ell + \ell+\ell_1)}(q)
		K^m(\ell)_{q,-q+\ell+\ell_1}
		K(\ell_1)_{q,-q+\ell+\ell_1}
		\eva{\xi_\lambda, a_{-q}^* a_{-q} \xi_\lambda} \,,
\end{equation}
and we bound it by
\begin{align}
	\sum_{\ell,\ell_1 \in \Zstar} \mathds{1}_{L_\ell}(q) |\I_0(\ell,\ell_1)|
	&\leq \Bigg( \sum_{\ell, \ell_1 \in \Zstar} \mathds{1}_{L_\ell \cap (-L_\ell + \ell+\ell_1)}(q) |K^m(\ell)_{q,-q+\ell+\ell_1}|^2 \Bigg)^{\half} \times \nonumber\\
	&\quad \times \Bigg( \sum_{\ell, \ell_1 \in \Zstar} \mathds{1}_{L_{\ell_1} \cap (-L_{\ell_1} + \ell+\ell_1)}(q) |K(\ell_1)_{q,-q+\ell+\ell_1}|^2 \Bigg)^{\half}
	\norm{ a_{-q} \xi_\lambda }^2 \nonumber\\
	&\leq \Bigg( \sum_{\ell\in \Zstar} (C \hat{V}(\ell))^{2m} \kF^{-1} e(q)^{-1} \Bigg)^{\half}
		\Bigg( \sum_{\ell_1 \in \Zstar} \hat{V}(\ell_1)^2 \kF^{-1} e(q)^{-1} \Bigg)^{\half}
	\Xi \nonumber\\
	&\leq C^m
		\kF^{-1} e(q)^{-1} \Xi \;.\label{eq:estEQ2181_bis}
\end{align}
This concludes the proof.
\end{proof}

\begin{lemma} \label{lem:EQ212}
Let $ \sum_{\ell \in \Zstar} \hat{V}(\ell)^2 < \infty $. For $\xi_\lambda = e^{-\lambda S} \Omega$, there exists a constant $ C > 0 $ such that for all $ \lambda \in [0,1] $, odd $ m \in \NNN $, and $ q \in B_{\F}^c $,
\begin{equation}
	\abs{\eva{\xi_\lambda,\left(E^{m,10}_{Q_2}(q) + E^{m,11}_{Q_2}(q) + \mathrm{h.c.}\right) \xi_\lambda }}
	\leq C^m \kF^{-1} \Xi e(q)^{-1} \;. \label{eq:estEQ212}
\end{equation}
\end{lemma}

\begin{proof}
We show the estimate for $ E^{m,10}_{Q_2}(q) $, the term $ E^{m,11}_{Q_2}(q) $ is analogous.
Splitting the multi-anticommutator in $ E^{m,10}_{Q_2}(q) $ by \eqref{eq:q-q} yields
\begin{equation} \label{eq:EQ2121}
\begin{aligned}
	& \abs{\eva{\xi_\lambda,\left(E^{m,10}_{Q_2}(q) + \mathrm{h.c.}\right) \xi_\lambda }}
	\le 2 \sum_{j=0}^{m+1} {{m+1}\choose j} \sum_{\ell \in \Zstar}\!\! \mathds{1}_{L_\ell}(q) |\I_j(\ell)| \;,
\end{aligned}
\end{equation}
where
\begin{equation}
\begin{aligned}
	& \I_j(\ell)
	\coloneq \sum_{r\in L_{\ell}}
		\eva{\xi_\lambda, K^{m+1-j}(\ell)_{r,q} K^{j}(\ell)_{q,r} a^*_{r-\ell} a_{r-\ell} \xi_\lambda} \;. \\
\end{aligned}
\end{equation}
Applying the Cauchy--Schwarz inequality and Lemma~\ref{lem:normsk} results in
\begin{equation}
\begin{split}
	\sum_{\ell \in \Zstar} \mathds{1}_{L_\ell}(q) |\I_0(\ell)|
	& \leq \sum_{\ell \in \Zstar} \mathds{1}_{L_\ell}(q) \norm{ K^{m+1}(\ell)_{q,q} a_{q-\ell} \xi_\lambda}\norm{ a_{q-\ell} \xi_\lambda } \\
&	\leq \sum_{\ell \in \Zstar} (C \hat{V}(\ell))^{m+1}
		\kF^{-1} e(q)^{-1} \Xi \;.
\end{split}
		\label{eq:estEQ2121}
\end{equation}
Since $ m+1 \ge 2 $, we get $ \sum_{\ell \in \Zstar} (C \hat{V}(\ell))^{m+1} \le C^m < \infty $. The estimate for $ |\I_{m+1}(\ell)| $ is analogous. Finally, for $ 1 \le j \le m $, we have
\begin{equation}
\begin{split}
	\sum_{\ell \in \Zstar} \mathds{1}_{L_\ell}(q) |\I_j(\ell)|
	& \leq \sum_{\ell \in \Zstar} \mathds{1}_{L_\ell}(q)
		\sum\limits_{r \in L_{\ell}}
		\norm{ K^{m+1-j}(\ell)_{r,q} a_{r-\ell}\xi_\lambda}
		\norm{ K^j(\ell)_{q,r} a_{r-\ell} \xi_\lambda } \\
&	\leq C^m \kF^{-1} e(q)^{-1} \Xi \;.
	\end{split}
	\label{eq:estEQ2122}
\end{equation}
This concludes the proof.
\end{proof}

\begin{lemma}[Exchange contribution] \label{lem:estnqex}
Let $ \sum_{\ell \in \Zstar} \hat{V}(\ell)^2 |\ell|^\alpha < \infty $ with $ \alpha \in (0,2) $, and recall the definition  \eqref{eq:nqexm} of $ n^{\ex,m}(q) $. For $\xi_\lambda = e^{-\lambda S} \Omega$, given $ \varepsilon > 0 $, there exists $ C_\varepsilon, C > 0 $ such that for all odd $m\in \Nbb$ and for all $ q \in B_{\F}^c $ we have
\begin{equation}
\begin{aligned}
	\abs{n^{\ex,m}(q)}
	&\leq C_\varepsilon C^m \kF^{-\frac 32 + \varepsilon} e(q)^{-\frac 32}
 	\qquad \text{for all } m > 1 \;, \\
 	\abs{n^{\ex,1}(q)}
 	&\leq C_\varepsilon \kF^{-1 - \frac{\alpha}{2} + \varepsilon} e(q)^{-1}  \;.
\label{eq:estnqex_Coulomb}
\end{aligned}
\end{equation}
If $ \sum_{\ell \in \Zstar} \hat{V}(\ell) < \infty $, we have the following bound for all odd $ m \in \mathbb{N} $:
\begin{equation}
	\abs{n^{\ex,m}(q)}
	\leq C^m \kF^{-2} e(q)^{-2} \;. \label{eq:estnqex}
\end{equation}
\end{lemma}
\begin{proof}
For $ m > 1 $, we expand $ n^{\ex,m}(q) $ using \eqref{eq:q-q} to get
\begin{equation}
\begin{aligned}
	|n^{\ex,m}(q)|
	\leq \I + \II
\end{aligned}
\end{equation}
with
\begin{displaymath}
\begin{aligned}
	\I
	&:= 4 \sum_{\ell,\ell_1 \in \Zstar}
		\mathds{1}_{L_\ell \cap L_{\ell_1} \cap (-L_\ell + \ell + \ell_1) \cap (-L_{\ell_1} + \ell + \ell_1)}(q)
		\abs{K^m(\ell)_{q,-q+\ell+\ell_1}}
		\abs{K(\ell_1)_{q,-q+\ell+\ell_1}} \\
	\II
	&:= 2 \sum_{1 \le j \le m-1} {{m}\choose j} \sum_{\ell,\ell_1 \in \Zstar}
		\mathds{1}_{L_\ell}(q)
		\sum_{\substack{r\in L_{\ell} \cap L_{\ell_1}\\ \cap (-L_{\ell}+\ell+\ell_1) \\ \cap (-L_{\ell_1}+\ell+\ell_1 )}}
		\abs{K^{m-j}(\ell)_{r,q}}
		\abs{K^j(\ell)_{q,-r+\ell+\ell_1}}
		\abs{K(\ell_1)_{r,-r+\ell+\ell_1}} \;.
\end{aligned}
\end{displaymath}
Since $ m > 1 $, $ \sum_{\ell \in \Zstar} \hat{V}(\ell)^m < \infty $ holds. Lemma~\ref{lem:normsk} and $  \lambda_{\ell,q} \geq \frac{1}{2} e(q) $ imply
\begin{align}
	\I
	&\leq C \sum_{\ell,\ell_1 \in \Zstar}
		\mathds{1}_{L_\ell \cap L_{\ell_1} \cap (-L_\ell + \ell + \ell_1) \cap (-L_{\ell_1} + \ell + \ell_1)}(q)
		\frac{\kF^{-2} \hat{V}(\ell)^m \hat{V}(\ell_1)}{(\lambda_{\ell,q} + \lambda_{\ell,-q+\ell+\ell_1}) (\lambda_{\ell_1,q} + \lambda_{\ell_1,-q+\ell+\ell_1})} \nonumber\\
	&\leq C \kF^{-2} e(q)^{-\frac 32} \sum_{\ell,\ell_1 \in \Zstar} \hat{V}(\ell)^m \hat{V}(\ell_1)
		\mathds{1}_{-L_\ell + \ell + \ell_1}(q)
		\lambda_{\ell,-q+\ell+\ell_1}^{-\frac 12} \nonumber\\ 
	&\leq C \kF^{-2} e(q)^{-\frac 32} \sum_{\ell \in \Zstar} \hat{V}(\ell)^m
		\Bigg( \sum_{\ell_1 \in \Zstar} \hat{V}(\ell_1)^2 \Bigg)^{\half}
		\Bigg( \sum_{\ell_1 \in \Z^3} \mathds{1}_{-L_\ell + \ell + \ell_1}(q)
		\lambda_{\ell,-q+\ell+\ell_1}^{-1} \Bigg)^{\half} \nonumber\\
	&\leq C \kF^{-\frac 32} e(q)^{-\frac 32} \;,
\end{align}
where we used~\eqref{eq:lambdainverse}. Likewise, for the summands of $ \II $, with $ \lambda_{\ell_1,r} \ge \frac{1}{2} e(r) $ we have
\begin{align}
	&\sum_{\ell,\ell_1 \in \Zstar}
		\mathds{1}_{L_\ell}(q)
		\sum_{r\in L_{\ell} \cap L_{\ell_1} \cap (-L_{\ell}+\ell+\ell_1) \cap (-L_{\ell_1}+\ell+\ell_1 )}
		\abs{K^{m-j}(\ell)_{r,q}}
		\abs{K^j(\ell)_{q,-r+\ell+\ell_1}}
		\abs{K(\ell_1)_{r,-r+\ell+\ell_1}} \nonumber\\
	&\leq \kF^{-3} \sum_{\ell,\ell_1 \in \Zstar} \sum_{r \in L_\ell \cap L_{\ell_1} \cap (-L_{\ell}+\ell+\ell_1)}
		(C \hat{V}(\ell))^m \hat{V}(\ell_1)
		\mathds{1}_{L_\ell}(q)
		\lambda_{\ell,q}^{-1} \lambda_{\ell,q}^{-\half} \lambda_{\ell,-r+\ell+\ell_1}^{-\half} \lambda_{\ell_1,r}^{-1} \nonumber\\
	&\leq \kF^{-3} e(q)^{-\frac 32} \sum_{\ell \in \Zstar} \sum_{r \in L_\ell}
		(C \hat{V}(\ell))^m
		\Bigg( \sum_{\ell_1 \in \Zstar} \hat{V}(\ell_1)^2  \Bigg)^{\half}
		\Bigg( \sum_{\ell_1 \in \Z^3} \mathds{1}_{-L_{\ell}+\ell+\ell_1}(r) \lambda_{\ell,-r+\ell+\ell_1}^{-1} \Bigg)^{\half}
		 e(r)^{-1} \nonumber\\
	&\leq \kF^{-\frac 52} e(q)^{-\frac 32} \sum_{\ell \in \Zstar} (C \hat{V}(\ell))^m
	\sum_{r \in L_\ell} e(r)^{-1}
	\leq C^m C_\varepsilon \kF^{-\frac 32 + \varepsilon} e(q)^{-\frac 32} \;.
\end{align}
Summing over $ j $, with $ \sum_{0 \le j \le m} {{m}\choose j} = 2^m $, concludes the estimate for $ \II $.

\medskip

For $ m = 1 $, recalling \eqref{eq:Theta} and \eqref{eq:Pq} and that $K(\ell)$ is symmetric, as well as the bound for the matrix elements of $K(\ell)$ given in Lemma~\ref{lem:normsk}, we obtain
\begin{align*}
	\abs{n^{\ex,1}(q)}
	& = \Big\lvert 4 \sum_{\ell,\ell_1 \in \Zstar}\mathds{1}_{L_\ell \cap L_{\ell_1} \cap (-L_{\ell}+\ell+\ell_1) \cap (-L_{\ell_1}+\ell+\ell_1)}(q)\, K(\ell)_{q,-q+\ell+\ell_1}K(\ell_1)_{q,-q+\ell+\ell_1} \Big\rvert \\
	&\leq C \sum_{\ell,\ell_1 \in \Zstar} 
		\mathds{1}_{L_\ell \cap L_{\ell_1} \cap (-L_{\ell}+\ell+\ell_1) \cap (-L_{\ell_1}+\ell+\ell_1)}(q)
		\frac{ \hat{V}(\ell) \kF^{-1}}{\lambda_{\ell,q} + \lambda_{\ell,-q+\ell+\ell_1}}
		\frac{\hat{V}(\ell_1) \kF^{-1}}{\lambda_{\ell_1,q} + \lambda_{\ell_1,-q+\ell+\ell_1}} \;.
\end{align*}
Using $ \lambda_{\ell,q} \ge C e(q)^{\half} e(q-\ell)^{\half} $ and the Cauchy--Schwartz inequality we get
\begin{align}
	|n^{\ex,1}(q)|
	&\leq C \kF^{-2}
		\sum_{\ell \in \Zstar} \mathds{1}_{L_\ell}(q) \frac{\hat{V}(\ell)}{\lambda_{\ell,q}}
		\sum_{\ell_1 \in \Zstar} \mathds{1}_{L_{\ell_1}}(q) \frac{\hat{V}(\ell_1)}{\lambda_{\ell_1,q}} \nonumber\\
	&\leq C \kF^{-2} e(q)^{-1}
		\Bigg( \sum_{\ell \in \Zstar} \hat{V}(\ell)^2 |\ell|^{\alpha} \Bigg)
		\Bigg( \sum_{\ell \in \Zstar} \mathds{1}_{L_\ell}(q) e(q-\ell)^{-1} |\ell|^{-\alpha} \Bigg) \;.	\tagg{eq:591}
\end{align}
The first sum is finite by assumption.

To bound the second sum in \eqref{eq:591}, we write the summation region as a union of pairwise disjoint sets, $ S := \{ \ell \in \Zstar : q \in L_\ell \} = S_1 \cup S_2 \cup S_3 $, where
\begin{displaymath}
\begin{split}
	S_1 & := \{ \ell \in S : |\ell| \le \kF^{1/2} \textnormal{ and } ||q-\ell|-\kF| \le 2 \} \;, \\
	S_2 & := \{ \ell \in S : |\ell| \le \kF^{1/2}  \textnormal{ and } ||q-\ell|-\kF| > 2  \} \;, \\
	S_3 & := \{ \ell \in S : |\ell| > \kF^{1/2} \} \;.
\end{split}
\end{displaymath}
We suppress the dependence on the fixed momentum $q$ (with $\lvert q\rvert > \kF$ according to the definition of $L_\ell$) in this notation.

For $ \ell \in S_1 $, we use $ e(q-\ell) \ge \frac{1}{2} $ to obtain
\begin{equation}
	\sum_{\ell \in S_1} e(q-\ell)^{-1} |\ell|^{-\alpha}
	\leq 2 \sum_{\ell \in S_1} |\ell|^{-\alpha}
	=  \int_{\Rbb^3} f(\ell') \di \ell'
\end{equation}
where we introduced a function $f: \Rbb^3 \to \Rbb$ which is constant on cubes centered on $\ell$ by
\[
	f(\ell') := 2 \sum_{\ell \in S_1} \lvert \ell \rvert^{-\alpha} \chi_{C_\ell}(\ell')\;, \qquad C_\ell := \Big[-\frac{1}{2},\frac{1}{2} \Big]^3 + \ell \;.
\]
Next we construct a function $g: \Rbb^3 \to \Rbb$ a follows:
\begin{itemize}
   \item If $\lvert \ell'\rvert \leq 1$, define $g(\ell') := f(\ell')$. Note that $\ell \in S_1$ in particular requires $L_\ell \not= \emptyset$; this implies $\ell \not= 0$ and therefore, having integer coordinates, $\lvert \ell \rvert \geq 1$. So $g$ is bounded.
   \item If $\lvert \ell' \rvert > 1$, we extend $S_1$ from discrete momenta to continuous momenta, and then enlarge it further by a distance $\sqrt{3}/2$, namely
\begin{equation}\label{eq:S1primeprime}
\begin{split}
   S'_1 := \bigg\{ \ell \in \Rbb^3: & \ \lvert q - \ell \rvert \leq \kF + \tfrac{\sqrt{3}}{2} \textnormal{ and } \lvert \ell \rvert \leq \kF^{1/2} + \tfrac{\sqrt{3}}{2} \\
   & \ \textnormal{ and } \lvert \lvert q-\ell \rvert -\kF\rvert \leq 2 + \tfrac{\sqrt{3}}{2} \bigg\} \;.
\end{split}
\end{equation}
Then we set
	\[
	   g(\ell') := \begin{cases}
	                  \left( \lvert \ell'\rvert - \frac{\sqrt{3}}{2} \right)^{-\alpha} & \text{for } \ell' \in S_1' \;, \\
	                  0 & \text{for } \ell' \not\in S_1' \;.
	               \end{cases}
	\]
\end{itemize}
Enlarging the set to include all points up to a distance $\sqrt{3}/2$ is important because this way, $\ell' \in \operatorname{supp}(f) = \bigcup_{\ell \in S_1} C_\ell$ implies $\ell' \in S_1'$. So for all $\ell' \in \Rbb^3$ with $\lvert \ell'\rvert > 1$ we have
\begin{equation}\label{eq:flessg}
\begin{split}
   f(\ell') & \leq 2 \sum_{\ell \in S_1} \big\lvert \lvert \ell' \rvert  -  \lvert \ell - \ell' \rvert \big\rvert^{-\alpha} \chi_{C_\ell}(\ell')
   \leq 2  \big\lvert \lvert \ell' \rvert  -  \tfrac{\sqrt{3}}{2} \big\rvert^{-\alpha} \sum_{\ell \in S_1} \chi_{C_\ell}(\ell') \\ &
   \leq 2  \Big( \lvert \ell' \rvert  -  \tfrac{\sqrt{3}}{2} \Big)^{-\alpha} \chi_{S_1'}(\ell') = 2 g(\ell') \;,
\end{split}
\end{equation}
where we bounded $\lvert \ell - \ell' \rvert$ by half the length of the diagonal of the cube $C_\ell$ which contains $\ell'$.

We estimate the integral over $f$ as follows. Close to the origin we have
\begin{equation}\label{eq:A}
	\int_{\lvert \ell'\rvert \leq 1} f(\ell') \di \ell' \leq \int_{\lvert \ell'\rvert \leq 1} g(\ell') \di \ell' \leq C \;;
\end{equation}
instead for larger $\ell'$, by means of \eqref{eq:flessg}, we have
\[
	\int_{\lvert \ell' \rvert > 1} f(\ell') \di \ell' \leq 2 \int_{S_1' \cap \{\lvert
	\ell' \rvert > 1\}} g(\ell') \di \ell' \;;
\]
and to compute this integral, without loss of generality we can assume $q = \lvert q \rvert e_3$, with $e_3$ the third canonical unit vector. We introduce spherical coordinates
\[
   \ell' = \begin{pmatrix}r \sin \theta \cos \varphi\\ r \sin \theta \sin \varphi \\ r \cos\theta\end{pmatrix} \;, \quad r \in (0,\infty),\ \theta \in [0,\pi],\ \varphi \in [0,2\pi) \;.
\]
and thus obtain
\begin{align}\label{eq:inte}
\begin{split}
 & \int_{S_1' \cap \{\lvert
	\ell' \rvert > 1\}} g(\ell') \di \ell' \\
	& = \int_1^{\infty} \di r\; r^2 \left( r - \frac{\sqrt{3}}{2} \right)^{-\alpha} \int_{0}^\pi \di \theta \sin(\theta) \int_0^{2\pi} \di\varphi\; \chi_{S_1'}\Big(\begin{pmatrix}r \sin \theta \cos \varphi\\ r \sin \theta \sin \varphi \\ r \cos\theta\end{pmatrix}\Big) \;.
\end{split}
\end{align}
To prove that the integral over $\theta$ is of order $1/r$, consider \eqref{eq:S1primeprime}:  obviously
 \begin{align*}
   S_1' = \Big\{ \ell  = \begin{pmatrix}r \sin \theta \cos \varphi\\ r \sin \theta \sin \varphi \\ r \cos\theta\end{pmatrix}: \kF - 2 - \tfrac{\sqrt{3}}{2} \leq \lvert q - \ell \rvert \leq \kF + \tfrac{\sqrt{3}}{2} \textnormal{ and } r \leq \kF^{1/2} + \tfrac{\sqrt{3}}{2} \Big\} \;.
 \end{align*}
Expanding the square of $\lvert q -\ell \rvert$, the first condition defining this set becomes
\[
		\left( \kF-2 - \tfrac{\sqrt{3}}{2} \right)^2 \leq r^2 + \lvert q\rvert^2 - 2 \lvert q\rvert r \cos \theta \leq \left( \kF+ \tfrac{\sqrt{3}}{2} \right)^2 \;.
\]
Solving the two inequalities for $\cos(\theta)$ we get
\[
   u_0 \leq \cos \theta \leq u_0 + \frac{(2 + \sqrt{3})(\kF - 1)}{\lvert q\rvert r} \;,
\]
where $u_0 := (r^2 + q^2 - \kF^2 - \tfrac{3}{4} - \kF \sqrt{3})(2 \lvert q\rvert  r)^{-1}$. Recall that $q \in L_\ell$, which contains the condition $q \in \BFc$, and thus $\lvert q \rvert \geq \kF - 1$. So $u_0 \leq \cos \theta \leq u_0 + (2 + \sqrt{3})/r$ and thus
\begin{align} \label{eq:B}
\begin{split}
\int_{S_1' \cap \{\lvert
	\ell' \rvert > 1\}} g(\ell') \di \ell'
	& = \int_1^{\kF^{1/2} + \frac{\sqrt{3}}{2}} \di r\; r^2 \left( r - \frac{\sqrt{3}}{2} \right)^{-\alpha} \int_{u_0}^{u_0 + (2+\sqrt{3})/r} \di u \;  2\pi \\
	& = \int_1^{\kF^{1/2} + \frac{\sqrt{3}}{2}} \di r\; r^2 \left( r - \frac{\sqrt{3}}{2} \right)^{-\alpha} \frac{(2 + \sqrt{3})2\pi}{r} \;.
\end{split}
\end{align}
Summing \eqref{eq:A} and \eqref{eq:B}, and using $r \geq 1$, we get
\begin{align*}
	 \sum_{\ell \in S_1} e(q-\ell)^{-1} |\ell|^{-\alpha}
	&
% 	\leq \int_{\mathbb{R}^3} f(\ell') \di \ell'
% \leq C + 2 \int_1^{\kF^{1/2} + \frac{\sqrt{3}}{2}} \di r \left( r - \frac{\sqrt{3}}{2} \right)^{-\alpha} r^2 \frac{2\pi(2+\sqrt{3})}{r}
% \\	&
\leq C + C \int_1^{\kF^{1/2} + \frac{\sqrt{3}}{2}} \di r \, r^{1-\alpha}  \Big( 1 - \frac{\sqrt{3}}{2} \Big)^{-\alpha}
	\le C \kF^{1 - \frac{\alpha}{2}} \;.
\end{align*}

For $ \ell \in S_2 $, we have $ e(q-\ell) \ge C \kF $ and use the sum--of--squares function $r_3(n)$, which counts the number of representations of an integer $n$ as the sum of three squares, to write
\begin{equation}
	\sum_{\ell \in S_2} e(q-\ell)^{-1} |\ell|^{-\alpha} \leq \frac{C}{\kF} \sum_{\ell \in S_2}  |\ell|^{-\alpha}
	= \frac{C}{\kF} \sum_{m = 1}^{\kF}  m^{-\alpha/2} \sum_{\ell \in S_2} \delta_{|\ell|^2, m}
	\leq \frac{C}{\kF} \sum_{m = 1}^{\kF}  m^{-\alpha/2} r_3(m) \;.
\end{equation}
The last inequality ignores the restriction of $\ell$ to $S_2$. It is a classic result (see, e.\,g., \cite{BRS17}) that $r_3(m) \leq C_\varepsilon m^{1/2 + \varepsilon}$. Thus
\begin{equation}
	\sum_{\ell \in S_2} e(q-\ell)^{-1} |\ell|^{-\alpha}
	\leq \frac{C_\varepsilon}{\kF} \sum_{m = 1}^{\kF}  m^{-\alpha/2}  m^{1/2 + \varepsilon}
	\leq \frac{C_\varepsilon}{\kF} \int_1^{\kF} m^{-\alpha/2 + 1/2 + \varepsilon} \di m
	\leq C_\varepsilon \kF^{\frac{1}{2} - \frac{\alpha}{2} + \varepsilon} \;.
\end{equation}

\medskip

Finally, for $ \ell \in S_3 $, we have $ |\ell|^{-\alpha} \le \kF^{-\alpha/2} $ and thus, using \eqref{eq:lambdainverse}, the estimate
\begin{equation}\label{eq:above}
	\sum_{\ell \in S_3} e(q-\ell)^{-1} |\ell|^{-\alpha}
	\le \kF^{-\alpha/2} \sum_{\ell \in S_3} e(q-\ell)^{-1}
	\le C_\varepsilon \kF^{1 + \varepsilon - \alpha/2} \;.
\end{equation}

Taking together the three cases we get
\begin{equation}
	\sum_{\ell \in \Zstar} \mathds{1}_{L_\ell}(q) e(q-\ell)^{-1} |\ell|^{-\alpha}
	\le C_\varepsilon \kF^{1 + \varepsilon - \frac{\alpha}{2}}
\end{equation}
	which implies
\begin{equation}	\label{eq:nex-final}
	|n^{\ex,1}(q)|
	\leq C_\varepsilon \kF^{-1 - \frac{\alpha}{2} + \varepsilon} e(q)^{-1} \;.
\end{equation}
A weaker, but still sufficient and simpler estimate is given in the footnote\footnote{
Write  $ S = S_1 \cup S_2$, where $S_1  := \{ \ell \in S : |\ell| \le \kF^{1/3} \} $ and $S_2  := \{ \ell \in S : |\ell| > \kF^{1/3} \} $. For $ \ell \in S_1 $, we use $ e(q-\ell) \ge \frac{1}{2} $ to obtain
\begin{equation}
	\sum_{\ell \in S_1} e(q-\ell)^{-1} |\ell|^{-\alpha} \leq 2 \sum_{\ell \in S_1}  |\ell|^{-\alpha}
	= 2 \sum_{m = 1}^{\kF^{2/3}}  m^{-\alpha/2} \sum_{\ell \in S_1} \delta_{|\ell|^2, m}
	\leq 2 \sum_{m = 1}^{\kF^{2/3}}  m^{-\alpha/2} r_3(m) \;.
\end{equation}
(The last inequality is far from optimal because it ignores the restriction of $\ell$ to $S_1$.) By the classic bound
\begin{equation}
	\sum_{\ell \in S_1} e(q-\ell)^{-1} |\ell|^{-\alpha}
	\leq 2 C_\varepsilon \sum_{m = 1}^{\kF^{2/3}}  m^{-\alpha/2}  m^{1/2 + \varepsilon}
	\leq C_\varepsilon \int_1^{\kF^{2/3}} m^{-\alpha/2 + 1/2 + \varepsilon} \di m
	\leq C_\varepsilon \kF^{-\alpha/3 + 1 + \varepsilon} \;.
\end{equation}
For $ \ell \in S_2 $, we have $ |\ell|^{-\alpha} \le \kF^{-\alpha/3} $, from where we can proceed as in \eqref{eq:above} and find a bound by $ C_\varepsilon \kF^{1 + \varepsilon - \alpha/3}$. Summing the two cases implies \eqref{eq:nex-final} with the exponent $\alpha/2$ replaced by $\alpha/3$.
}.

\bigskip

If the stronger hypothesis $ \sum_{\ell \in \Zstar} \hat{V}(\ell) < \infty $ holds, we get \eqref{eq:estnqex} using Lemma~\ref{lem:normsk} and $ \sum_{0 \le j \le m} {{m}\choose j} = 2^m $ via
\begin{align}
	\I + \II
	&\leq C^m \kF^{-2} e(q)^{-2} \sum_{\ell,\ell_1 \in \Zstar}
		\hat{V}(\ell)^m
		\hat{V}(\ell_1)
	+ C^m \kF^{-2} e(q)^{-2} \sum_{\ell,\ell_1 \in \Zstar}
		\hat{V}(\ell)^m
		\norm{K(\ell_1)}_{\max,1} \nonumber\\
	&\leq C^m
		\Bigg( \sum\limits_{\ell \in \Zstar} \hat{V}(\ell)^m \Bigg)
		\Bigg( \sum\limits_{\ell_1 \in \Zstar} \hat{V}(\ell_1) \Bigg)
		\kF^{-2} e(q)^{-2}
	\leq C^m \kF^{-2} e(q)^{-2} \;. \qedhere \nonumber
\end{align}
\end{proof}

\begin{proof}[Proof of Proposition~\ref{prop:finEQ2est}]
We sum the estimates from Lemmas~\ref{lem:EQ211}--\ref{lem:estnqex} and use $ e(q) \ge \half $, $ \Xi \le 1 $, and $ \norm{(\NN+1)^{5/2} \xi_\lambda} \le C $ from Lemma~\ref{lem:gronNest}.
\end{proof}

\begin{proof}[Proof of Proposition~\ref{prop:finalEmest}]
We use Propositions~\ref{prop:finEQ1est} and~\ref{prop:finEQ2est} in \eqref{eq:errEm2}, taking the supremum over $ \lambda \in [0,1] $.
\end{proof}

\section{Analysis of the Leading-Order Term}
\label{sec:leading_order_analysis}

In this section we show that the first term in~\eqref{eq:finexpan} equals $ n^{\RPA}(q) $ defined in \eqref{eq:nqb}. Moreover we establish the scaling $ n^{\RPA}(q) \sim C \kF^{-1} $.

\begin{lemma}[Integral formula for $ n^{\RPA}(q) $] \label{lem:nqb_integralrecovery}
Let $q \in B^c_{\F}$, then
\begin{equation} \label{eq:nqb_integralrecovery}
	\half\sum_{\ell\in \Zstar}\mathds{1}_{L_\ell}(q) \big( \cosh(2K(\ell)) - 1 \big)_{q,q} = n^{\RPA}(q)\;.
\end{equation}
\end{lemma}

\begin{proof}
We drop the $ \ell $-dependence of $K(\ell) $, $ h(\ell) $ and $ P(\ell) = |v_\ell \rangle \langle v_\ell| $. Obviously
\begin{equation} \label{eq:coshrewriting}
	\cosh(2K)-1
	= \half\big((e^{-2K}-1)-(1-e^{2K})\big) \;.
\end{equation}
Using the notation $ P_w = |w \rangle \langle w| $, so $ P = P_v $, from \eqref{eq:K} we get
\begin{equation} \label{eq:e-2k}
	e^{-2K} = h^{-\half} \big(h^2 +2P_{h^{\half} v}\big)^{\half} h^{-\half} \;, \qquad
	e^{2K} = h^{\half} \big(h^2 +2P_{h^{\half} v}\big)^{-\half} h^{\half} \;.
\end{equation}
We then express $ (e^{-2K}-1)_{q,q} $ and $ (1-e^{2K})_{q,q} $ using the identities
\begin{equation} \label{eq:intid}
	A^\half = \frac{2}{\pi} \int_0^\infty \left(1- \frac{t^2}{A+t^2}\right) \mathrm{d}t \;,\qquad
	A^{-\half} = \frac{2}{\pi} \int_0^\infty \frac{\mathrm{d}t}{A+t^2} \;,
\end{equation}
for any symmetric, invertible matrix $ A $, as well as the Sherman--Morrison formula
\begin{equation} \label{eq:shermor}
	(A+cP_w)^{-1} = A^{-1} - \frac{c}{1+c\eva{w, A^{-1}w}}P_{A^{-1}w} \;,
\end{equation}
for any $ c \in \C $ and $ w \in \ell^2(L_\ell) $ such that the denominator is non-vanishing. We begin with
\begin{align}
	\big(h^2 +2P_{h^{\half} v}\big)^{\half}
	&= \frac{2}{\pi} \int_0^\infty \Bigg( 1- \frac{t^2}{t^2+h^2} + \frac{2 t^2}{1+ 2 \big\langle h^{\half} v ,(t^2+h^2)^{-1} h^\half v \big\rangle } P_{(t^2+h^2)^{-1}h^{\half} v} \Bigg) \mathrm{d}t \nonumber\\
	&= h + \frac{2}{\pi} \int_0^\infty \frac{2t^2}{1+ 2 \big\langle h^{\half} v ,(t^2+h^2)^{-1} h^\half v \big\rangle }  P_{(t^2+h^2)^{-1}h^{\half} v}\mathrm{d}t \;.
\end{align}
Using the canonical basis vectors $ (e_p)_{p \in L_\ell} $ with $ h e_q = \lambda_{\ell,q} e_q $ and $ g_\ell = \langle e_p,v \rangle^2 $, this implies
\begin{align}
	(e^{-2K}-1)_{q,q}
	&= \eva{e_q, h^{-\half} \big(h^2 +2P_{h^{\half} v}\big)^{\half} h^{-\half} e_q} - 1\nonumber\\
	&= \frac{2}{\pi} \int_0^\infty \frac{2t^2}{1+ 2 \big\langle h^{\half} v ,(t^2+h^2)^{-1} h^\half v \big\rangle } \eva{e_q,h^{-\half} P_{(t^2+h^2)^{-1}h^{\half} v}h^{-\half} e_q}\mathrm{d}t\nonumber\\
	&= \frac{2}{\pi} \int_0^\infty \frac{2g_\ell t^2 (t^2+\lambda^2_{\ell,q})^{-2}}{1+ 2g_\ell\sum_{p \in L_\ell}\lambda_{\ell,p}(t^2+\lambda^2_{\ell,p})^{-1} } \mathrm{d}t \;. \label{eq:e-2k_integral}
\end{align}
Similarly we arrive at
\begin{equation} \label{eq:e2kfin}
	(1-e^{2K})_{q,q}
	= \frac{2}{\pi} \int_0^\infty \frac{2g_\ell \lambda_{\ell,q}^2 (t^2+\lambda^2_{\ell,q})^{-2}}{1+ 2g_\ell\sum_{p \in L_{\ell}}\lambda_{\ell,p}(t^2+\lambda^2_{\ell,p})^{-1} } \mathrm{d}t \;.
\end{equation}
Summing both terms we obtain the claimed result.
\end{proof}

\begin{lemma}[Estimate for the RPA contribution] \label{lem:nqb_bounds}
For any potential $ \hat{V} \in \ell^1(\Zstar) $, there exists some $ C > 0 $
such that for all $ q \in \Z^3 $ we have
\begin{equation} \label{eq:nqb_upperbound}
	n^{\RPA}(q)
	\le C \kF^{-1} e(q)^{-1} \;.
\end{equation}
\end{lemma}

\begin{proof}
We focus on the case $ q \in B_{\F}^c $; $ q \in B_{\F} $ is treated analogously. We use \eqref{eq:nqb_integralrecovery}, expand the $ \cosh $, and use Lemma~\ref{lem:normsk}, as well as $ \lambda_{\ell,q} \ge \frac{1}{2} e(q) $ to get
\[
	n^{\RPA}(q)
	\le \half \sum_{\ell \in \Zstar} \mathds{1}_{L_\ell}(q) \sum_{m=1}^{\infty} \frac{4^m |(K(\ell)^{2m})_{q,q}|}{(2m)!}
	\le \sum_{\ell \in \Zstar} \frac{\kF^{-1}}{\lambda_{\ell,q}} \sum_{m=1}^{\infty} \frac{C^m \hat{V}(\ell)^{2m}}{(2m)!}
	\le C \frac{\kF^{-1}}{e(q)} \;. \qedhere
\]
\end{proof}

\section{Conclusion of the Proof of Proposition~\ref{prop:main}}
\label{sec:mainthmproof}

\begin{proof}[Proof of Proposition~\ref{prop:main}]
We focus on the case $ q \in B_{\F}^c $, since $ q \in B_{\F} $ is completely analogous. Recall the trial state $ \Psi_N = R e^{-S} \Omega $, so Proposition~\ref{prop:finexpan} yields
\begin{align*}
	\eva{\Omega, e^{S} a_q^* a_q e^{-S} \Omega} 
	&= \half\sum_{\ell\in \Zstar}\mathds{1}_{L_\ell}(q) \sum_{\substack{m=2\\m:\textnormal{ even}}}^n \frac{((2K(\ell))^m)_{q,q}}{m!}
		+ \half \sum_{m=1}^{n-1} \eva{\Omega, E_m(P^q)\Omega}\nonumber\\
	&\quad +\half \int_{\Delta^n} \di^n\underline{\lambda} \;
		\eva{\Omega, e^{\lambda_n S}Q_{\sigma(n)}(\Theta^n_{K}(P^q)) e^{-\lambda_n S} \Omega} \;,
\end{align*}
for any $ n \in \N $. As $ n \to \infty $, the third term vanishes by Proposition~\ref{prop:headerr}, while the first one by Lemma~\ref{eq:nqb_integralrecovery} converges to
\begin{equation*}
	\half\sum_{\ell\in \Zstar}\mathds{1}_{L_\ell}(q) \big( \cosh(2K(\ell)) - 1 \big)_{q,q}
	= n^{\RPA}(q) \;.
\end{equation*}
We estimate the $ E_m(P^q) $-terms by Proposition~\ref{prop:finalEmest} to get
\begin{align} \label{eq:main_errorbound_with_Xi}
	& \abs{\eva{\Omega, e^{S} a_q^* a_q e^{-S} \Omega} - n^{\textnormal{RPA}}(q) - \frac{1}{4} n^{\ex,1}(q) } \nonumber\\
	&\le C_\varepsilon \sum_{m=1}^\infty \frac{C^m}{m!}
		\bigg( e(q)^{-1}\left( \kF^{-\frac 32 + \varepsilon}
		+ \kF^{-1 - \frac{\alpha \gamma}{2}}
		+ \kF^{-1 + \frac{3-\alpha}{2} \gamma} \Xi^\half
		+ \kF^{-1+\varepsilon} \Xi^\half \right) \nonumber\\
	&\hspace{7em} + e(q)^{-\half} \kF^{-1} \sup_{\lambda \in [0,1]} \eva{\Omega, e^{\lambda S} a_q^* a_q e^{-\lambda S} \Omega}^{\half} \bigg) \;.
\end{align}
We observe also that the expectation value $\eva{\Omega, e^{\lambda S} a_q^* a_q e^{-\lambda S} \Omega}$ with $\lambda \in [0,1]$ can be expanded the same way; the only difference is that everywhere the matrices $K(k)$ are replaced by $\lambda K(k)$. In particular, all our estimates remain valid since $\lambda \leq 1$, and therefore, we can easily take the supremum over $\lambda \in [0,1]$ in the all following steps.
\paragraph{Global bootstrap:} To estimate $ \Xi = \sup_{q \in \Z^3} \sup_{\lambda \in [0,1]} \eva{\Omega, e^{\lambda S} a_q^* a_q e^{- \lambda S} \Omega} $, observe that \eqref{eq:main_errorbound_with_Xi} holds uniformly in $ q \in B_{\F}^c $, and it remains true for $ q \in B_{\F} $. Moreover, we have  $ \sup_{\lambda \in [0,1]} \eva{\Omega, e^{\lambda S} a_q^* a_q e^{-\lambda S} \Omega} \le \Xi $. So setting $ \gamma := 0 $ and $\varepsilon := \frac{\alpha}{4}$ in~\eqref{eq:main_errorbound_with_Xi}, using $e(q)^{-1} \leq 2$, taking the supremum over $ q \in \Z^3 $ and $ \lambda \in [0,1] $, and estimating $ n^{\RPA}(q) $ by Lemma~\ref{lem:nqb_bounds} and $n^{\ex,1}(q)$ by Lemma~\ref{lem:estnqex}, we get
\begin{align} \label{eq:Xibound}
	\Xi
	&\le \sup_{q \in \Z^3} n^{\RPA}(q) + \sup_{q \in \Z^3} \frac{n^{\ex,1}(q)}{4}
		+ C_\varepsilon \left( \kF^{-1}
		+ \kF^{-1 + \varepsilon} \Xi^\half \right)
	\le  \left( \kF^{-1} + C \kF^{-1 + \frac{\alpha}{4}} \Xi^{1/2}
	\right) \nonumber
\end{align}
which by the Cauchy--Schwartz inequality implies
\begin{align}
		\Xi \leq C \kF^{-1} + C \kF^{-1+\frac{\alpha}{2}} \Xi \;.
\end{align}
This can be resolved for
\begin{align}
	\Xi
	 \le \frac{C}{1-C \kF^{-1 + \frac{\alpha}{2}}} \kF^{-1} \leq C \kF^{-1}
\end{align}
where in the last step we replaced the fraction by a (larger) constant, which is a valid upper bound for any $\kF^{-1}$ large enough.
\paragraph{Local bootstrap:} Plugging this bound into~\eqref{eq:main_errorbound_with_Xi}, again with $\gamma=0$ and $\varepsilon = \frac{\alpha}{4}$, and taking only the supremum over $ \lambda \in [0,1] $ (not over $q$), gives
\begin{align}
	&\sup_{\lambda \in [0,1]} \eva{\Omega, e^{\lambda S} a_q^* a_q e^{-\lambda S} \Omega} \nonumber\\
	&\le n^{\RPA}(q) + \frac{n^{\ex,1}(q)}{4}
	+ C \kF^{-1} e(q)^{-1}
		+ C \kF^{-1} e(q)^{-\half} \Big(\sup_{\lambda \in [0,1]} \eva{\Omega, e^{\lambda S} a_q^* a_q e^{-\lambda S} \Omega} \Big)^{\half}
\end{align}
which, estimating $ n^{\RPA}(q) $ by Lemma~\ref{lem:nqb_bounds}  and $n^{\ex,1}(q)$ by Lemma~\ref{lem:estnqex}, implies
\begin{displaymath}
	\sup_{\lambda \in [0,1]} \eva{\Omega, e^{\lambda S} a_q^* a_q e^{-\lambda S} \Omega}
	\le
	C \kF^{-1} e(q)^{-1}
% 	+ {C \kF^{-1 - \frac{\alpha}{4}} e(q)^{-1}}
		+ C \kF^{-1} \Big( e(q)^{-1}  + \sup_{\lambda \in [0,1]} \eva{\Omega, e^{\lambda S} a_q^* a_q e^{-\lambda S} \Omega} \Big) \;.
\end{displaymath}
This can be resolved for $\sup_{\lambda \in [0,1]} \eva{\Omega, e^{\lambda S} a_q^* a_q e^{-\lambda S} \Omega}$ to yield
\begin{align}\label{eq:aqaq_bound}
	\sup_{\lambda \in [0,1]} \eva{\Omega, e^{\lambda S} a_q^* a_q e^{-\lambda S} \Omega}
% 	\le \frac{2C}{1-C \kF^{-1}} \kF^{-1} e(q)^{-1}
	\leq C \kF^{-1} e(q)^{-1} \;.
\end{align}

\paragraph{Conclusion of the bootstrap:} Inserting \eqref{eq:Xibound} and~\eqref{eq:aqaq_bound} into~\eqref{eq:main_errorbound_with_Xi} and choosing the optimal value $ \gamma = \frac 13 $ implies that for all $\varepsilon > 0$ there exists $C_\varepsilon$ such that we have
\begin{align}
	\abs{\eva{\Omega, e^{S} a_q^* a_q e^{-S} \Omega} - n^{\textnormal{RPA}}(q) -\frac{n^{\ex,1}}{4}} \le C_\varepsilon
		\kF^{-1} e(q)^{-1}\left( \kF^{-\frac{1}{2} + \varepsilon}
		+ \kF^{- \frac{\alpha}{6}}\right) \;.
\end{align}
This is the final result~\eqref{eq:main_prop_2}.

\bigskip

Under the assumption of \eqref{hyp:ell1}, Proposition~\ref{prop:finalEmest} provides us with
\begin{align} \label{eq:main_errorbound_with_Xi_subcoulomb}
	& \abs{\eva{\Omega, e^{S} a_q^* a_q e^{-S} \Omega} - n^{\textnormal{RPA}}(q)
- \frac{n^{\ex,1}(q)}{4}
	}  \\
	& \le
	C_\varepsilon
		\Big(  e(q)^{-1} \left( \kF^{-2}
		+ \kF^{-\frac{3}{2}} \Xi^\half
		+ \kF^{-1} \Xi^{1-\varepsilon} \right)  + e(q)^{-\half} \kF^{-1} \Xi^{\half - \varepsilon} \sup_{\lambda \in [0,1]} \eva{\Omega, e^{\lambda S} a_q^* a_q e^{-\lambda S} \Omega}^{\half} \Big) \,. \nonumber
\end{align}
Estimating $ \Xi $ and $ \eva{\Omega, e^{\lambda S} a_q^* a_q e^{-\lambda S} \Omega} $ by Lemma~\ref{lem:estnqex},~\eqref{eq:Xibound}, and~\eqref{eq:aqaq_bound}, we obtain the improved error bound~\eqref{eq:main_prop_improvederror}.
\end{proof}

\section*{Acknowledgments}
SL would like to thank Phan Thành Nam for discussions. NB was supported by the European Union through the ERC Starting Grant \textsc{FermiMath}, grant agreement nr.~101040991. SL was partially supported by the ERC Starting Grant \textsc{FermiMath}, grant agreement nr.~101040991, and partially by the ERC Advanced Grant MathBEC, grant agreement nr.~101095820. Views and opinions expressed are those of the authors and do not necessarily reflect those of the European Union or the European Research Council Executive Agency. Neither the European Union nor the granting authority can be held responsible for them. The authors were partially supported by Gruppo Nazionale per la Fisica Matematica in Italy.

\section*{Statements and Declarations}
The authors have no competing interests to declare.

\section*{Data Availability}
As purely mathematical research, there are no datasets related to this article.

\footnotesize
\newcommand{\etalchar}[1]{$^{#1}$}

\end{document}